\documentclass[12pt]{article}
\usepackage{amsmath}
\usepackage{graphicx,psfrag,epsf}
\usepackage{enumerate}
\usepackage{natbib}
\usepackage{url} 

\newcommand{\blind}{0}

\usepackage[table,dvipsnames]{xcolor}
\usepackage{amsthm}
\usepackage{mathtools}
\usepackage{collcell}
\usepackage{dsfont}
\usepackage{rotating}
\usepackage{wrapfig}
\usepackage[utf8]{inputenc}
\usepackage{tikz}
\usepackage{lettrine}
\usepackage{dblfloatfix}
\usepackage[english]{babel} 
\usepackage{amsfonts}
\usepackage{subfig}
\graphicspath{ {fig/} }
\usepackage[hyperfootnotes=false]{hyperref}

\usepackage{lipsum}

\newcommand\blfootnote[1]{%
	\begingroup
	\renewcommand\thefootnote{}\footnote{#1}%
	\addtocounter{footnote}{-1}%
	\endgroup
}
\newtheorem*{proposition}{Proposition}

\begin{document}

	\def\spacingset#1{\renewcommand{\baselinestretch}%
		{#1}\small\normalsize} \spacingset{1}

	
	\if0\blind
	{
		\title{\bf Econometric modelling and forecasting of intraday electricity prices}
		\author{Michał Narajewski\hspace{.2cm}\\
			University of Duisburg-Essen\\
			and \\
			Florian Ziel \\
			University of Duisburg-Essen}
		\maketitle
	} \fi
	
	\if1\blind
	{
		\bigskip
		\bigskip
		\bigskip
		\begin{center}
			{\LARGE\bf Title}
		\end{center}
		\medskip
	} \fi
	
	\bigskip
	\begin{abstract}
		In the following paper, we analyse the ID$_3$-Price in the German Intraday Continuous electricity market using an econometric time series model. A multivariate approach is conducted for hourly and quarter-hourly products separately. We estimate the model using lasso and elastic net techniques and perform an out-of-sample, very short-term forecasting study. The model's performance is compared with benchmark models and is discussed in detail. Forecasting results provide new insights to the German Intraday Continuous electricity market regarding its efficiency and to the ID$_3$-Price behaviour. 
		\blfootnote{This research article was partially supported by the German Research Foundation (DFG, Germany) and the National Science Center (NCN, Poland) through BEETHOVEN grant no. 2016/23/G/HS4/01005.\\
		\textcopyright 2019. This manuscript version is made available under the CC-BY-NC-ND 4.0 license \url{http://creativecommons.org/licenses/by-nc-nd/4.0/}}
	\end{abstract}
	
	\noindent%
	\vfill
	
	\newpage
	\spacingset{1.45} 
	
	
	\section{Introduction}
	The constant development of the weather-dependent renewable energy production in Germany requires a flexible market in which power plants can balance their production forecast errors that may be caused by changing, unpredicted weather conditions. The introduction of intraday electricity markets addresses these problems and lets market participants trade energy continuously until 30 minutes before the delivery begins in the whole market and until 5 minutes before the delivery begins within respective control zones. Although the intraday markets' popularity grows rapidly, the corresponding literature does not follow its pace.
		
	While the electricity price forecasting (EPF) in day-ahead markets is willingly researched, there are, to the best of our knowledge, only a few articles regarding the forecasting of intraday electricity prices. To be specific, \cite{andrade2017probabilistic} performed a probabilistic price forecasting of electricity prices, and \cite{monteiro2016short} carried out a forecasting of intraday electricity prices using artificial neural networks. Both these papers are based on the Spanish market data. Recently \cite{uniejewski2018understanding} conducted research regarding the forecasting of intraday electricity prices that is close to our direction. They carried out a very-short term price forecasting of the ID$_3$-Price index for hourly products in the EPEX German Intraday Continuous market. There is definitely more literature on the intraday electricity markets regarding other topics than the EPF. \cite{ziel2017modeling}, \cite{pape2016fundamentals} or \cite{gonzalez2015impact} investigate the impact of fundamental regressors on the formation of intraday prices. On the other hand, \cite{Kiesel2017} or \cite{aid2016optimal} focus their research on bidding behaviour in the intraday market.
	
	The following paper aims to take a closer look at the electricity price formation in the intraday market. We want to understand better the intraday market itself and the processes that drive the price formation of both hourly and quarter-hourly Intraday Continuous products. Therefore, we focus our attention on the ID$_3$-Price index. We model it in a~multivariate manner, which is a well-known technique in the electricity price forecasting, see \cite{weron2014electricity}. We utilize an autoregressive approach, but we also make use of the continuity of the Intraday Continuous market. Our goal is to take advantage of all the information that is available on the market. Additionally, as external regressors, we take into consideration the results of the Day-Ahead and Intraday Auctions, and the imbalance volume from the balancing market. Let us note that we do not make use of any fundamental regressors, e.g. wind or solar forecast errors.
	
	In the next section, we shortly explain the intraday market rules and the function of the ID$_3$-Price index. We briefly analyze the aforementioned ID$_3$, and, based on it, we define a more general intraday price measure called~$_x\text{ID}_y$. Then, descriptive statistics are presented and the stationarity of the ID$_3$ prices is examined. In the third section, we discuss a variance stabilization transformation, following the recommendations of \cite{Uniejewski2018}, and we describe the model estimation techniques, i.e. the least absolute shrinkage and selection operator (lasso) of \cite{Tibshirani1996} and the elastic net regularization of \cite{Zou05regularizationand}. Then, we propose a full information model and present benchmark models. In the fourth section, we describe a forecasting study, utilized error measures, the \cite{diebold1995comparing} test, and a measure of the importance of coefficients as in \cite{Ziel2016}. In the fifth section, we present results, and we conduct an in-depth discussion of these. We compare the forecasts of the considered models, and we take a closer look at the variable selection for the best performing models. We close the following paper with a conclusion.
	
	\section{Market Description}
	
	 Trading in the German Intraday Continuous market begins every day at 15:00 for hourly products and at 16:00 for quarter-hourly products of the following day. The Intraday Continuous market is preceded by the Day-Ahead Auction and the Intraday Auction, which take place daily at 12:00 and 15:00, respectively, see \cite{EPEX2018}. For a better visualisation see Figure \ref{fig:market}. We describe there only the aforementioned products, but there is more trading taking place daily in the German electricity market, for instance, a forward market, a balancing auction, or an EXAA auction. For more details see \cite{Viehmann2017}.
	
	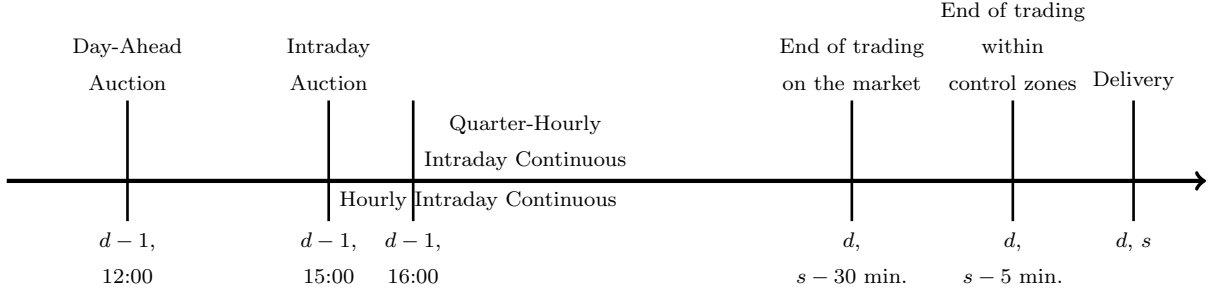
\begin{figure*}[!t]
		\begin{tikzpicture}[scale=1.07]
			\draw [->] [ultra thick] (0,0) -- (14.9,0);
			\draw [line width = 1] (1.5,1) -- (1.5, -0.5);
			\node [align = center, below, font = \scriptsize] at (1.5, -0.5) {$d-1$,\\ 12:00};
			\node [align = center, above, font = \scriptsize] at (1.5,1) {Day-Ahead\\Auction};
			\draw [line width = 1] (4,1) -- (4, -0.5);
			\node [align = center, below, font = \scriptsize] at (4, -0.5) {$d-1$,\\ 15:00};
			\node [align = center, above, font = \scriptsize] at (4,1) {Intraday\\Auction};
			\node [align = center, below right, font = \scriptsize] at (4,0) {Hourly Intraday Continuous};
			\draw [line width = 1] (5.05,1) -- (5.05, -0.5);
			\node [align = center, below, font = \scriptsize] at (5.05, -0.5) {$d-1$,\\ 16:00};
			\node [align = center, above right, font = \scriptsize] at (5.05,0) {Quarter-Hourly\\ Intraday Continuous};
			
			\draw [line width = 1] (10.5,1) -- (10.5, -0.5);
			\node [align = center, above, font = \scriptsize] at (10.5,1) {End of trading\\ on the market};
			\node [align = center, below, font = \scriptsize] at (10.5, -0.5) {$d$,\\ $s - 30$ min.};
			
			\draw [line width = 1] (12.5,1) -- (12.5, -0.5);
			\node [align = center, above, font = \scriptsize] at (12.5,1) {End of trading\\ within\\ control zones};
			\node [align = center, below, font = \scriptsize] at (12.5, -0.5) {$d$,\\ $s - 5$ min.};
			\draw [line width = 1] (14,1) -- (14, -0.5);
			\node [align = center, above, font = \scriptsize] at (14,1) {Delivery};
			\node [align = center, below, font = \scriptsize] at (14, -0.5) {$d$, $s$};
 		\end{tikzpicture}
		\caption{The daily routine of the German electricity market. $d$ corresponds to the day of the delivery and $s$ corresponds to the hour of the delivery.}
		\label{fig:market}
	\end{figure*}

	In the forward and day-ahead markets, a term "Price" is pretty straightforward, but concerning the intraday market it is not that clear what one means when speaking of an "Intraday Price". Considering the last transaction's price as the product's current price may be misleading. The volatility of the prices is highly dependent on the volume of traded energy — the smaller the volume is, the more scattered the prices can be. This pattern often results in temporal jumps of prices. Thus, several price measures were introduced by EPEX: Price Index, ID$_3$-Price, and ID$_1$-Price. The measures are applied to each product separately, so it gives us 24 values of each for the hourly products and 96 values for the quarter-hourly products. The Price Index is a volume-weighted average of the prices of the transactions in the whole trading period, ID$_3$-Price is a volume-weighted average of the prices of the transactions during the last 3 hours of trading, and the ID$_1$-Price is analogous to the ID$_3$-Price averaging the last hour of the prices instead of 3 hours. The transactions that take place in the respective control zones later than 30 minutes before the delivery are not taken into account in the calculation of these indices.
	
	\subsection{ID$_3$-Price}
	
	In the following article, we focus our attention on the ID$_3$-Price because of the importance of this index. It serves as an underlying for the German Intraday Cap/Floor Futures. With these financial instruments, market participants can hedge against positive or negative price spikes in the German electricity market, see \cite{EEX2018}. Let $b(d,s)$ be a start of the delivery of product $s$ on day~$d$. By $\mathbb{T}_3^{d,s} = \left[b(d,s) - 3, b(d,s) - 0.5\right)$ we denote a time frame between 3 hours and 30 minutes before the start of the delivery, where $[x,y)$ stands for a half-open interval. The EPEX definition of the ID$_3$-Price is as follows.
	
	\begin{equation}
		\text{EPEX ID}_3^{d,s} := \frac{1}{\sum_{k \in \mathbb{T}_3^{d,s}\cap \mathcal{T}^{d,s}} V_k^{d,s}} \sum_{k \in \mathbb{T}_3^{d,s}\cap \mathcal{T}^{d,s}} V_k^{d,s}P_k^{d,s},
		\label{eq:ID3}
	\end{equation}
		where $\mathcal{T}^{d,s}$ is a set of timestamps of transactions regarding the product $s$ on day $d$, $V_k^{d,s}$ and $P_k^{d,s}$ are the volume and the price of $k$-th trade within the transaction set $\mathbb{T}_3^{d,s}\cap \mathcal{T}^{d,s}$ respectively. 
		
		For the calculation of the ID$_3$-Price domestic and cross-border transactions are taken into account while the so-called cross-trades, i.e. trades within the same counterparty, are excluded. In the case of no trades within the $\mathbb{T}_3^{d,s}$ period, the averaging window is extended to the whole trading period of the product $s$ on day $d$. If no trades at all are present, then for quarter-hourly products the respective Intraday Auction value is used and for hourly products, the respective Day-Ahead Auction result is used.
		
		In the purpose of our analysis, we want to reconstruct the EPEX ID$_3$ as well as it is possible. Unfortunately, the data that is available to market participants do not consist of the information whether each transaction was a cross-trade or not. Besides, we disregard the block trades, which are not that common in Intraday Continuous market and are associated only with a small volume of traded energy. Since we aim at a very short-term price forecasting, we want to be able to use all the information available on the market at the time of forecasting. The price measures constructed by EPEX tell us the price level either for the full period of trading or the last few hours of trading. To get to know the price value of a product at a particular time during the trading period, we define an $_x\text{ID}_y$ function as follows.
		
		\begin{equation}
			{}^{}_x\text{ID}_y^{d,s} := \frac{1}{\sum_{k \in \mathbb{T}_{x,y}^{d,s}\cap \mathcal{T}^{d,s}} V_k^{d,s}} \sum_{k \in \mathbb{T}_{x,y}^{d,s}\cap \mathcal{T}^{d,s}} V_k^{d,s}P_k^{d,s},	
			\label{eq:xIDy}
		\end{equation}
		where $\mathbb{T}_{x,y}^{d,s} = \left[b(d,s) - x - y, b(d,s) - x \right)$, $x \ge 0$ and $y > 0$.

		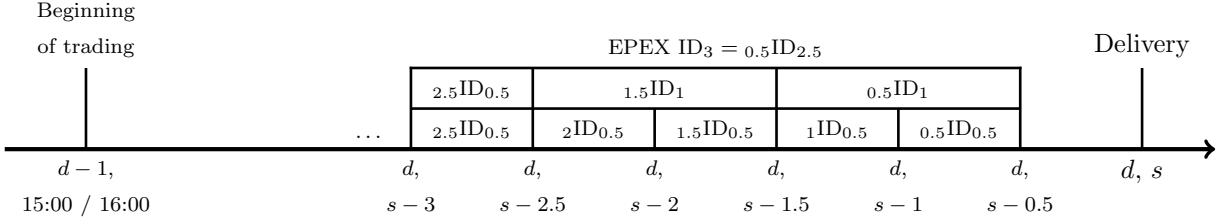
\begin{figure*}[t!]
			\begin{tikzpicture}[scale=1.08]
			\draw [->] [ultra thick] (0,0) -- (14.9,0);
			
			\node [align = center, below, font = \scriptsize] at (1, 0) {$d-1$,\\ 15:00 / 16:00};
			\node [align = center, below, font = \scriptsize] at (5, 0) {$d$,\\$s - 3$};
			\node [align = center, below, font = \scriptsize] at (6.5, 0) {$d$,\\$s - 2.5$};
			\node [align = center, below, font = \scriptsize] at (8, 0) {$d$,\\$s - 2$};
			\node [align = center, below, font = \scriptsize] at (9.5, 0) {$d$,\\$s - 1.5$};
			\node [align = center, below, font = \scriptsize] at (11, 0) {$d$,\\$s - 1$};
			\node [align = center, below, font = \scriptsize] at (12.5, 0) {$d$,\\$s - 0.5$};
			\node [align = center, below, font = \footnotesize] at (14, 0) {$d$, $s$};
			
			\draw [line width = 1] (1, 1) -- (1,0);
			\draw [line width = 1] (5, 1.015) -- (5,0);
			\draw [line width = 1] (6.5, 1) -- (6.5,0);
			\draw [line width = 1] (8, 0.5) -- (8,0);
			\draw [line width = 1] (9.5,1) -- (9.5, 0);
			\draw [line width = 1] (11,0.5) -- (11, 0);
			\draw [line width = 1] (12.5,1.015) -- (12.5, 0);
			\draw [line width = 1] (14,1) -- (14, 0);
		
			\draw [line width = 1] (4.985, 1) -- (12.515,1);
			\draw [line width = 1] (5, 0.5) -- (12.5,0.5);
			
			\node [align = center, above, font = \scriptsize] at (1,1) {Beginning \\of trading};
			\node [align = center, above, font = \scriptsize] at (4.5,0) {\dots};
			\node [align = center, above, font = \footnotesize] at (14,1) {Delivery};
			
			\node [align = center, above, font = \scriptsize] at (5.75,0) {$_{2.5}\text{ID}_{0.5}$};
			\node [align = center, above, font = \scriptsize] at (7.25,0) {$_{2}\text{ID}_{0.5}$};
			\node [align = center, above, font = \scriptsize] at (8.75,0) {$_{1.5}\text{ID}_{0.5}$};
			\node [align = center, above, font = \scriptsize] at (10.25,0) {$_{1}\text{ID}_{0.5}$};
			\node [align = center, above, font = \scriptsize] at (11.75,0) {$_{0.5}\text{ID}_{0.5}$};
			
			\node [align = center, above, font = \scriptsize] at (5.75,0.5) {$_{2.5}\text{ID}_{0.5}$};
			\node [align = center, above, font = \scriptsize] at (8,.5) {$_{1.5}\text{ID}_{1}$};
			\node [align = center, above, font = \scriptsize] at (11,0.5) {$_{0.5}\text{ID}_{1}$};
			
			\node [align = center, above, font = \scriptsize] at (8.75,1) {EPEX ID$_3 =   {}_{0.5}\text{ID}_{2.5}$};		
			\end{tikzpicture}
			\caption{Illustration of combining $_x\text{ID}_y$ on longer time frames using multiple $_x\text{ID}_y$ of shorter time frames. $d$~corresponds to the day of delivery and $s$ corresponds to the hour of delivery.}
			\label{fig:xIDyproperty}
		\end{figure*}

	For the calculation of the $_x\text{ID}_y$ we use the same transaction types as EPEX does in the calculation of their indices, but we change its behaviour in the case of no trades in the considered time frame~$\mathbb{T}_{x,y}^{d,s}$. That is to say, in the case of no trades instead of extending the averaging window to the whole trading period, we set the $_x\text{ID}_y$ price to the price of the last transaction that occurred before the time frame $\mathbb{T}_{x,y}^{d,s}$. If no trades are present before the considered time frame, we use the Intraday Auction and Day-Ahead Auction values, similarly as it is done for the EPEX ID$_3$.
One can see that with the definition (\ref{eq:xIDy}) in most cases $\text{EPEX ID}_3 = {}_{0.5}\text{ID}_{2.5}$ and $\text{EPEX ID}_1 = {}_{0.5}\text{ID}_{0.5}$. These would differ only in the case of no transactions in the considered time frame, which is a rare event. Thus, in purpose of our analysis when mentioning ID$_3$, we will have $_{0.5}\text{ID}_{2.5}$ in our minds.

Let us note that the $_x\text{ID}_y$ possesses a so-called weighted additivity property. This means that if we consider a disjoint split of the $\mathbb{T}_{x,y}^{d,s}$ period
\begin{equation}
\begin{aligned}
\mathbb{T}_{x,y}^{d,s}  =  \dot{\bigcup_{j}} \mathbb{T}_{x_j,y_j}^{d,s} 
= \dot{\bigcup_{j}} \left[b(d,s) - x_j - y_j, b(d,s) - x_j \right)		
\end{aligned}
\end{equation}
for $j \in \{0, 1, \dots, J\}$, where $x_0 + y_0 = x + y$ and $x_J = x$. Then

\begin{equation}
{}_{x}\text{ID}^{d,s}_y = \frac{  \sum_{j} {}_{x_j}\text{ID}^{d,s}_{y_j} \mathbb{V}_{x_j,y_j}^{d,s}}{\sum_{j} \mathbb{V}_{x_j,y_j}^{d,s}},
\end{equation}
where $\mathbb{V}_{x,y}^{d,s} = \sum_{k \in \mathbb{T}_{x,y}^{d,s} \cap \mathcal{T}^{d,s}} V_k^{d,s}$. The proof can be found in the Appendix. This property is useful in a computational optimization and can be helpful in a better understanding of the relation between the $_x\text{ID}_y$ for different $x$ and $y$. An~example is shown in Figure~\ref{fig:xIDyproperty}. Naturally, we can split the $\mathbb{T}_{x,y}^{d,s}$ period to time frames that are not equally long and we can continue constructing $_x\text{ID}_y$ until the beginning of trading.
	

\subsection{Descriptive statistics}
\begin{wrapfigure}[10]{r}{0.48\textwidth} 
	\includegraphics[width = 1\linewidth]{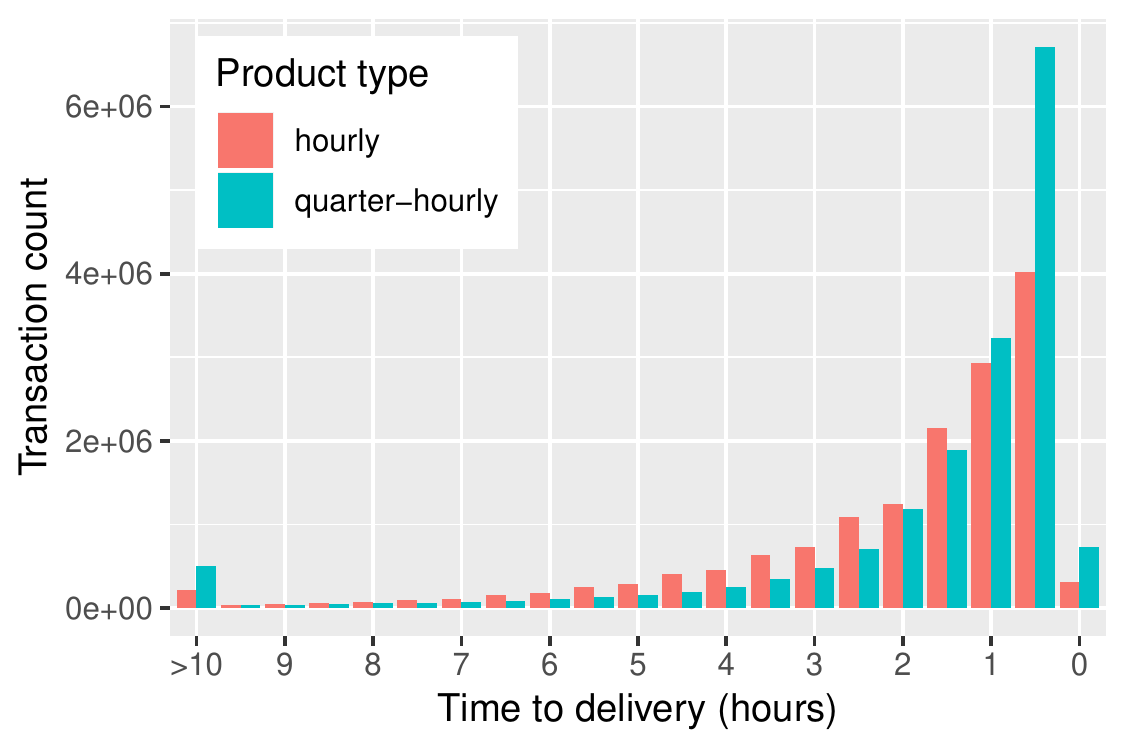}
	\caption{Distribution of transactions in the intraday market over time to delivery.}
	\label{fig:numberoftrans}
\end{wrapfigure}
In the purpose of our analysis, we use the data regarding the Intraday Continuous transactions. The data consist of hourly and quarter-hourly products, and they span the date range from 01.01.2015 to 29.09.2018. Besides, we use the corresponding data regarding the Day-Ahead Auction, Intraday Auction, and Balancing Volume as external regressors. Figure \ref{fig:numberoftrans} shows the distribution of transactions in the German intraday market depending on the time to the delivery. Let us note that most of the transactions of both considered types happen in the last hours before the delivery of the product. This distribution is even more skewed for quarter-hourly products. Over $70\%$ of all hourly and over $80\%$ of all quarter-hourly trades take place during the ID$_3$ time frame. Thus, considering the ID$_3$ as a price measure for the intraday market is a reasonable idea.
\begin{table}[!b]
\centering
\begingroup\small
\begin{tabular}{rrrrrrr}
	\hline
	& Min. & 1st Qu. & Median & Mean & 3rd Qu. & Max. \\ 
	\hline
	hourly & 6 & 280 & 421 & 472.19 & 603 & 19479 \\ 
	quarter-hourly & 0 & 70 & 109 & 129.72 & 172 & 1434 \\ 
	\hline
\end{tabular}
\endgroup
\caption{Summary statistics of the number of transactions traded in the hourly and quarter-hourly German Intraday Continuous market per product}
\label{tab:transactions}
\end{table}

\begin{table}[!t]
	\centering
	\begingroup\small
	\begin{tabular}{rrrrrrr}
		\hline
		& Min. & 1st Qu. & Median & Mean & 3rd Qu. & Max. \\ 
		\hline
		Day-Ahead Auction & 14948.8 & 23528.58 & 26842.05 & 27596.18 & 31008.43 & 51465.5 \\ 
		Intraday Hourly & 12.3 & 2517.8 & 3546.3 & 3773.52 & 4781.85 & 15173.2 \\ 
		Intraday Auction & 4.22 & 77.62 & 117.72 & 142.54 & 182.35 & 1239 \\ 
		Intraday Quarter-Hourly & 0 & 69.22 & 114.11 & 128.92 & 171.41 & 939.25 \\ 
		\hline
		\hline
	\end{tabular}
	\endgroup
	\caption{Summary statistics of the energy volume (MWh) traded in the German Day-Ahead and Intraday markets per product}
	\label{tab:volumes}
\end{table}

Table~\ref{tab:transactions} presents the basic summary statistics regarding the number of transactions that are traded in the hourly and quarter-hourly Intraday Continuous market. We observe that, on the average, there are around 472 transactions per product in the hourly intraday market, and at the same time, on the average, there are around 130 transactions per product in the quarter-hourly intraday market. Bearing in mind that there are 4 times more quarter-hourly products than the hourly ones, it is clear that, on the average, there is more trading taking place in the quarter-hourly intraday market. Moreover, we see that it is not only theoretically possible that there are no trades on a particular product. Let us note that there are also instances with a huge number of trades, comparing to the average.  In Table~\ref{tab:volumes}, we report the basic summary statistics of the volume traded in the intraday market and we compare them with the statistics of the volume traded in the day-ahead market. Obviously, the day-ahead market is incomparably bigger in terms of the traded volume than the intraday market, but the German Intraday Continuous market gains the relevance every year, what is depicted in Figure \ref{fig:volumetrend}.

\begin{wrapfigure}{r}{0.3\textwidth} 
	\includegraphics[width = 1\linewidth]{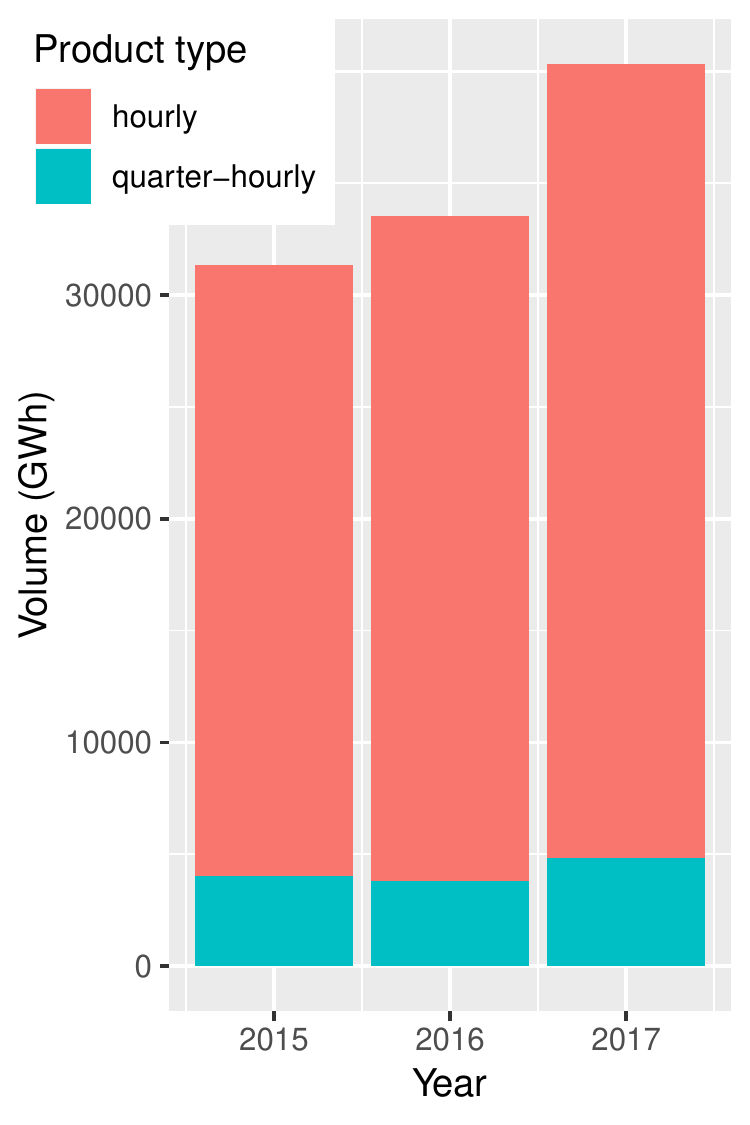}
	\caption{Energy volume traded in the Intraday Continuous market over years}
	\label{fig:volumetrend}
\end{wrapfigure}

	\begin{figure}[b!]
	\centering
	\subfloat{
		\includegraphics[width = 0.65\linewidth]{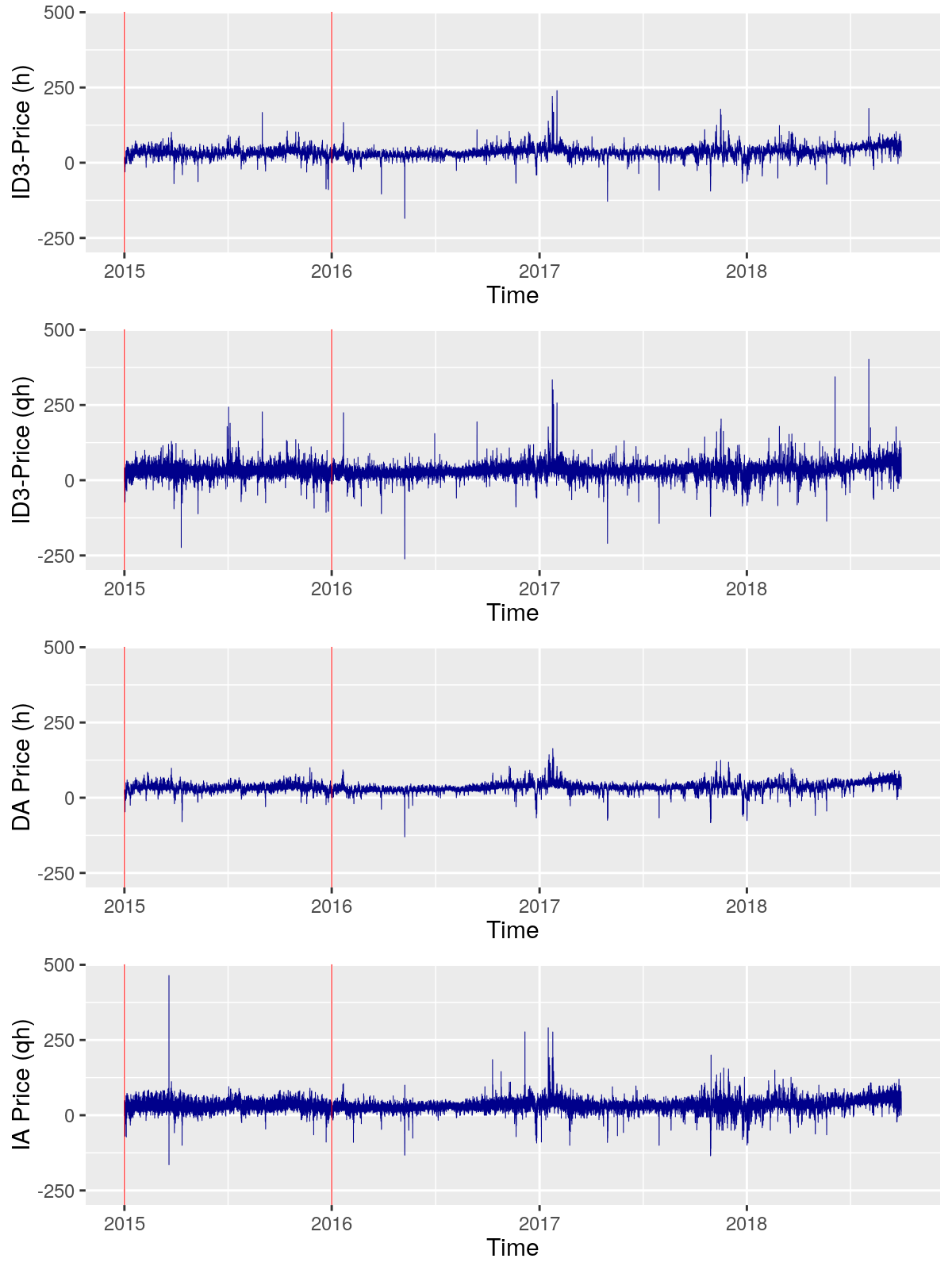}
	}
	~ 
	\subfloat{
		\includegraphics[width = 0.33\linewidth]{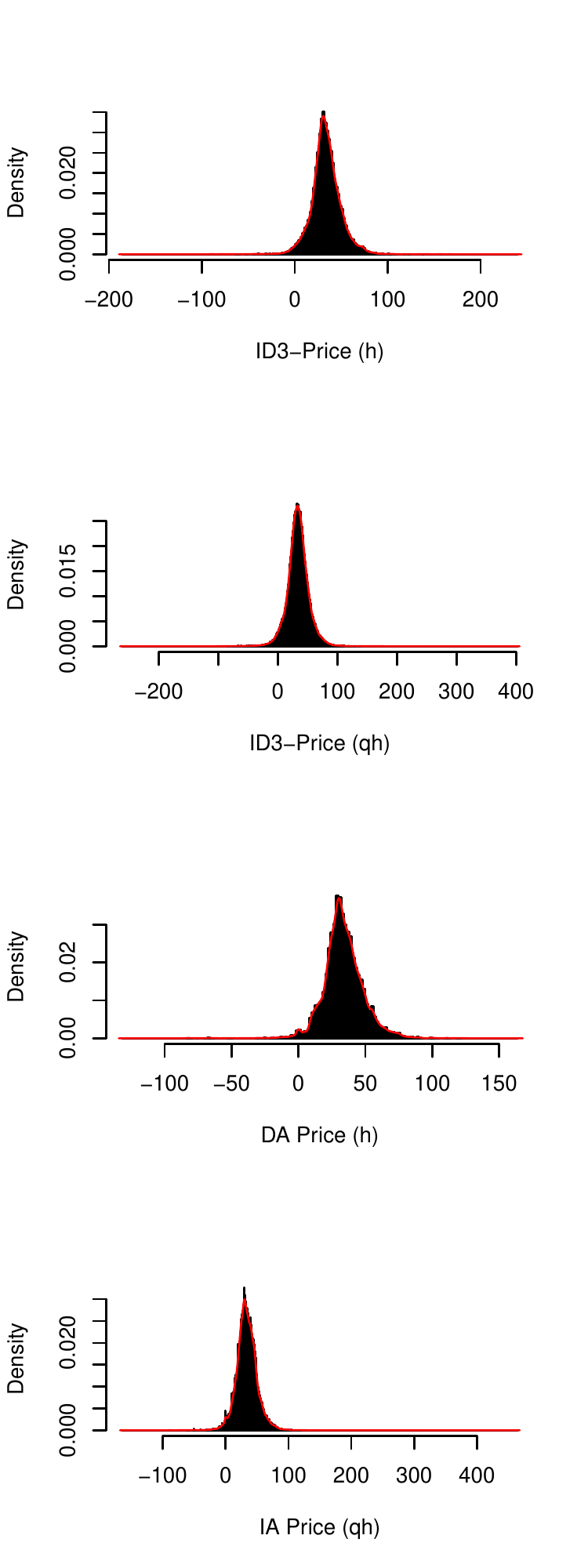}
	}
	\caption{ID$_3$-Price of hourly products (top), ID$_3$-Price of quarter-hourly products (second top), Day-Ahead Price (third top) and Intraday Auction Price (bottom) over time. Red lines indicate the initial rolling window period. The histograms with kernel density estimates are presented on the right.}
	\label{fig:ID3overtime}
\end{figure}

	Figure \ref{fig:ID3overtime} presents the ID$_3$-Price, the Day-Ahead Price, and the Intraday Auction Price of hourly and quarter-hourly products over time. Let us note that the variance of prices is substantial and the outliers occurrence is quite often. Still, the quarter-hourly products tend to exhibit higher variance of prices and outliers appearance frequency than the hourly products. This behaviour is even better visible in Figure~\ref{fig:ID3weekly}. We present there the weekly sample mean of: the ID$_3$-Price for both considered product types, the Day-Ahead Price, and the Intraday Auction Price. Let us note that Figure \ref{fig:ID3weeklyh} is smoother than Figure \ref{fig:ID3weeklyqh}.	The latter one exhibits a so-called jigsaw pattern, which is broadly explained by \cite{Kiesel2017}. Based on these plots, it is obvious that the ID$_3$ price for hourly products is less volatile. It is worth mentioning that, weekly, the intraday prices perform similarly to the day-ahead prices. The latter ones are well-described in the literature, see e.g. \cite{Ziel2015}. In Figure \ref{fig:ID3weekly}, we observe that the prices on the average behave almost identically from Tuesday to Friday, regardless of the product type. The prices on Monday are very similar to those between Tuesday and Friday, despite the night hours. Analogously to the day-ahead prices, we observe a weekend effect, which means that the prices are lower on Saturday, and on Sunday they are even lower than on Saturday.
	
	\begin{figure*}[b!]
		\centering
		\subfloat[]{
			\includegraphics[width = 0.475\linewidth]{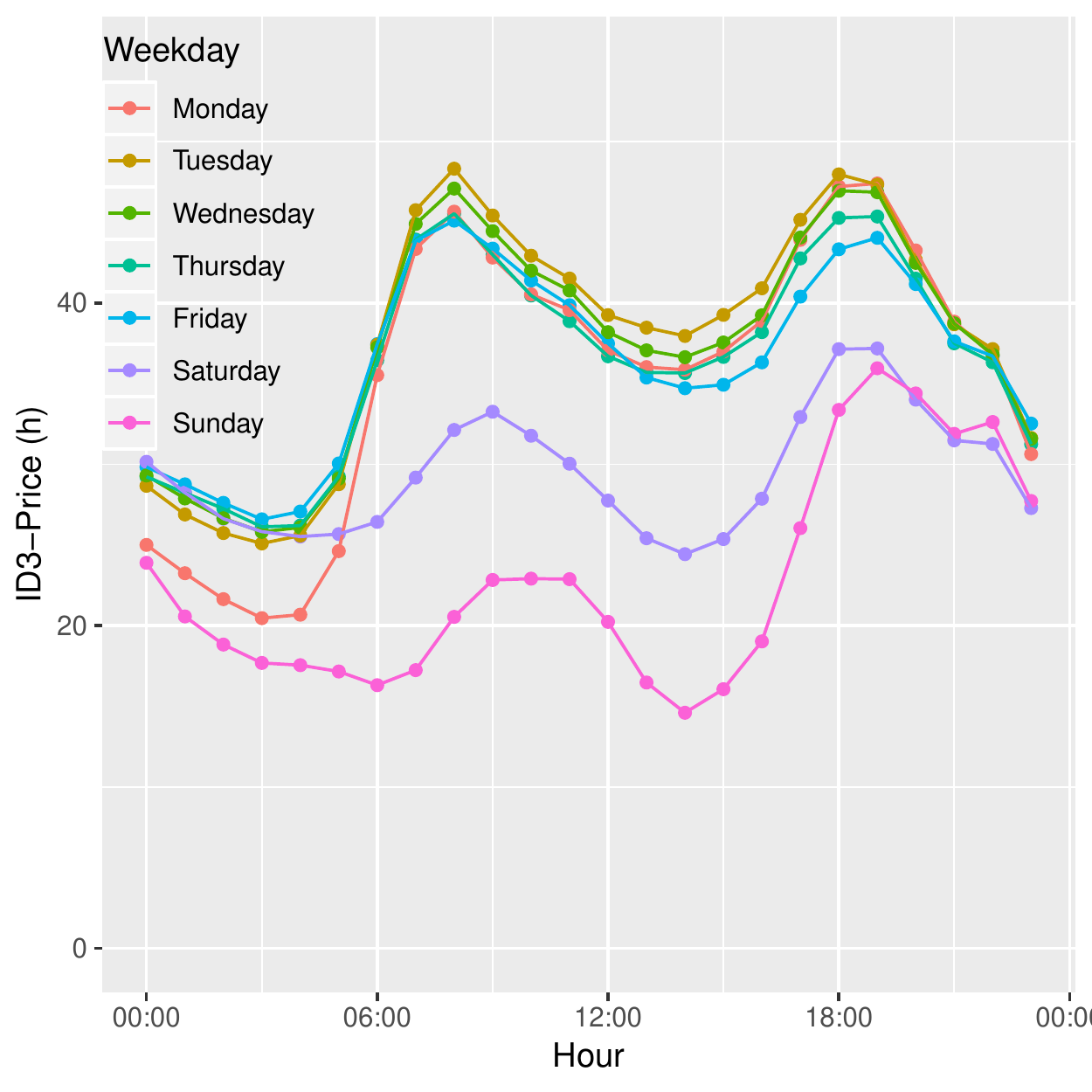}
			\label{fig:ID3weeklyh}
		}
		~ 
		\subfloat[]{
			\includegraphics[width = 0.475\linewidth]{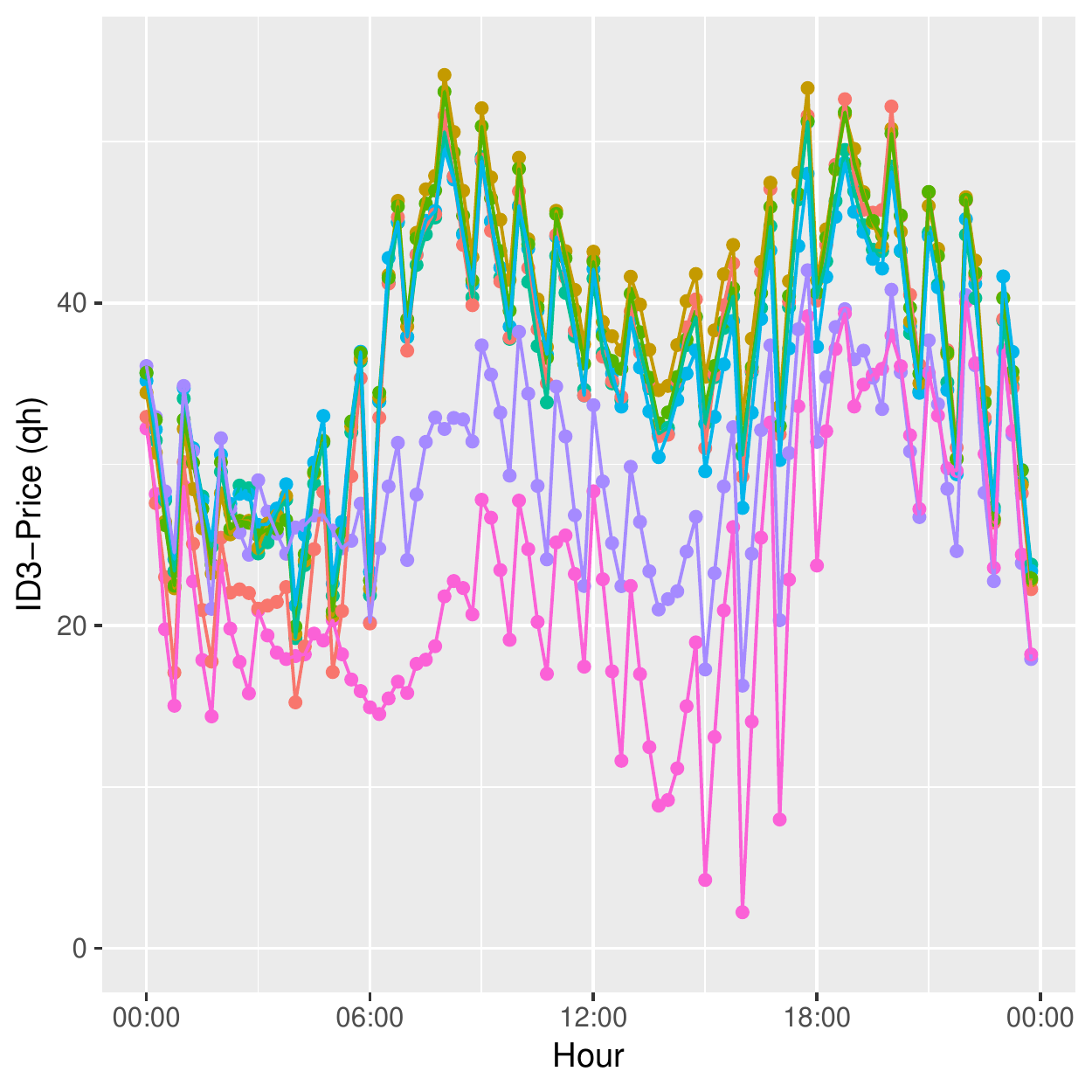}
			\label{fig:ID3weeklyqh}
		}
		\hspace{0mm}
		\subfloat[]{
			\includegraphics[width = 0.475\linewidth]{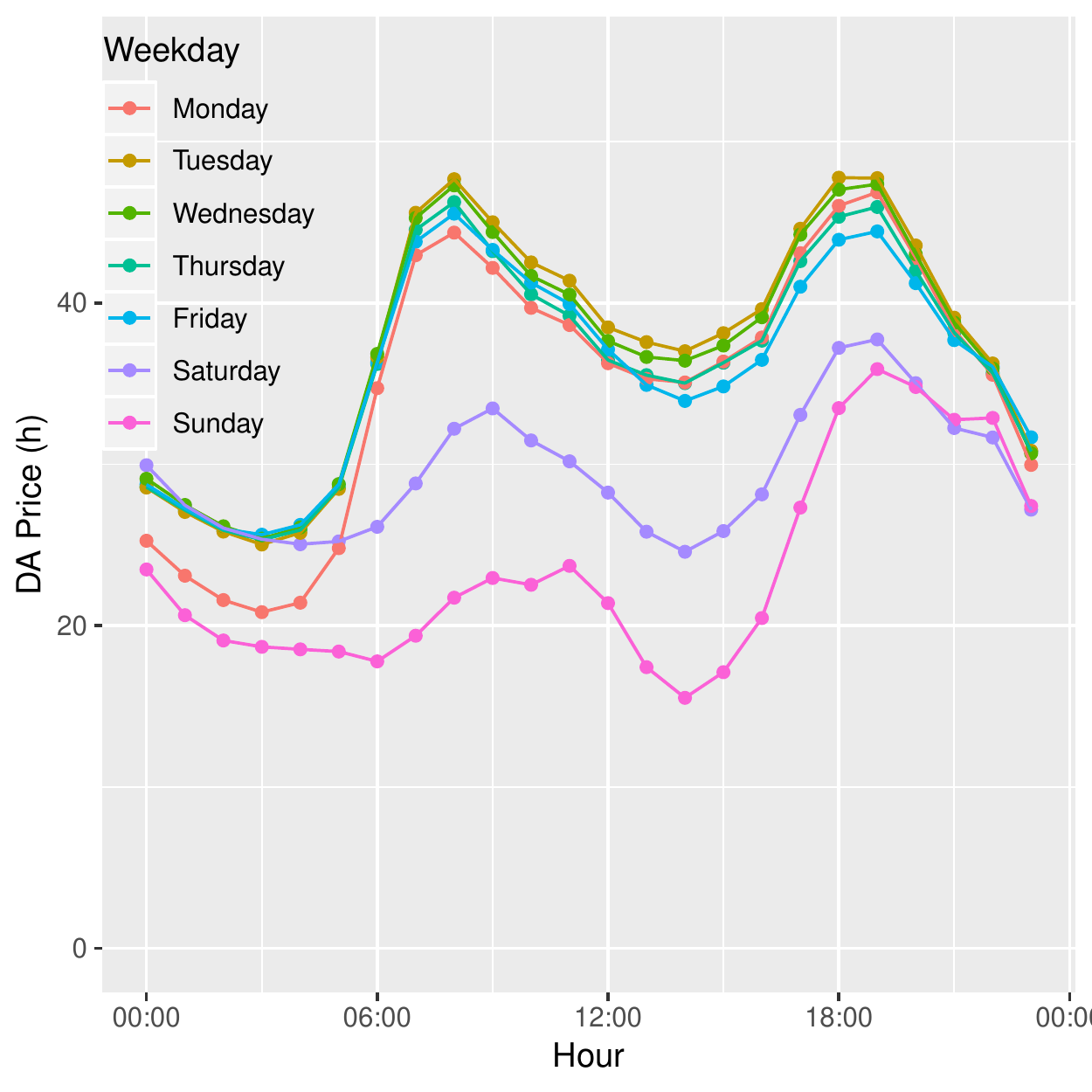}
			\label{fig:DAweeklyh}
		}
		\subfloat[]{
			\includegraphics[width = 0.475\linewidth]{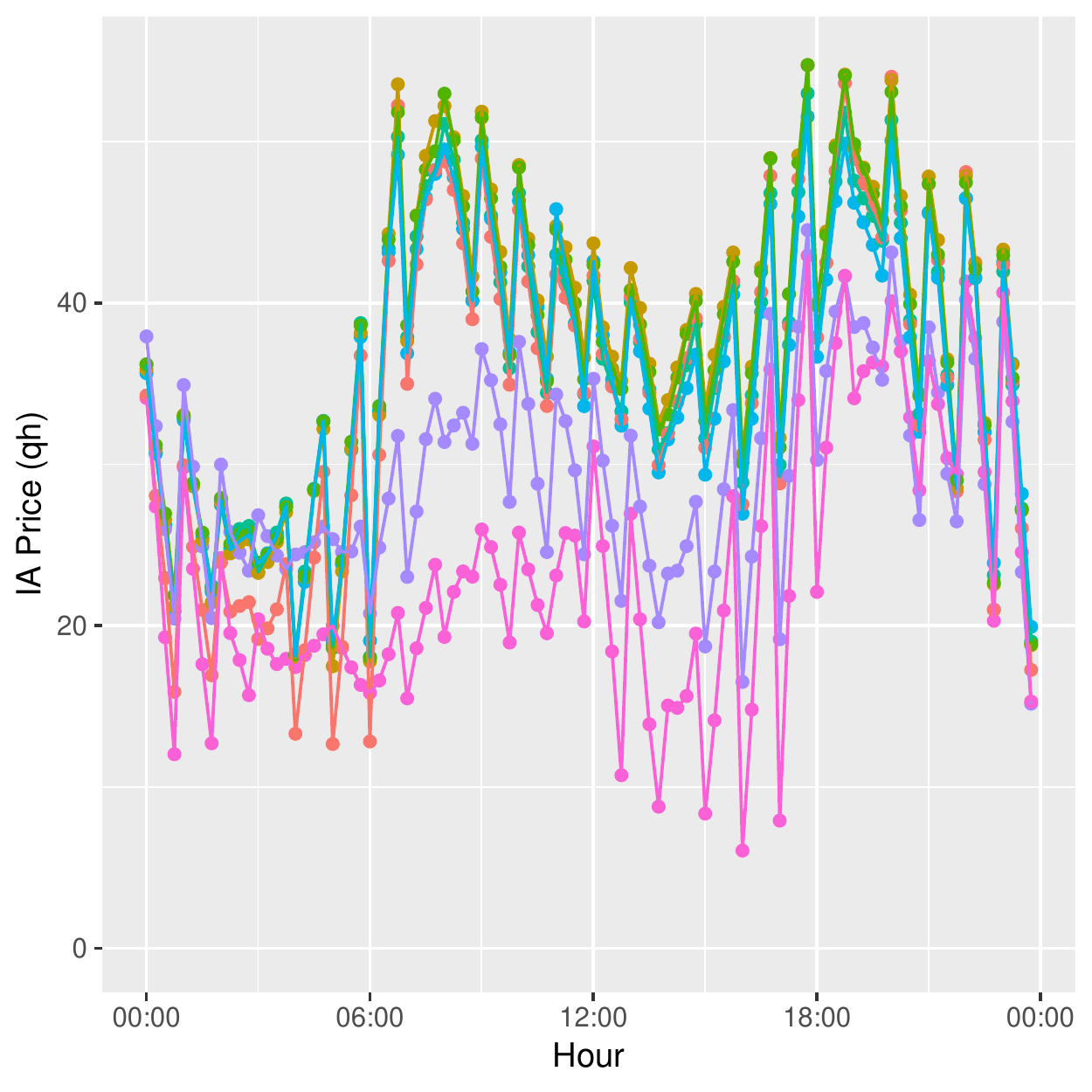}
			\label{fig:IAweeklyqh}
		}
		\caption{Weekly sample mean of: (a) ID$_3$-Price for hourly products, (b) ID$_3$-Price for quarter-hourly products, (c) Day-Ahead Price and (d) Intraday Auction Price.}
		\label{fig:ID3weekly}
	\end{figure*}

	In order to gain meaningful insights regarding the stationarity of the ID$_3$ prices, we perform three tests. The first one is the Dickey-Fuller test (\citealp{dickey1979distribution}), the second is the Augmented Dickey-Fuller test (\citealp{said1984testing}) and the third test is the Phillips-Perron (\citealp{phillips1988testing}). All of them evaluate the null hypothesis that a unit root is present in an autoregressive model against the alternative that the considered data is stationary or trend-stationary. The Dickey-Fuller test is the basic one, while the Augmented Dickey-Fuller and Phillips-Perron tests are extensions to the first one. The tests are applied to the ID$_3$ prices of every product separately. This means that we run each test 120 times. Resulting p-values are smaller than 0.01 for every test and every ID$_3$ series. Due to this, we reject the null hypothesis that there is a unit root in any of the ID$_3$ price series. 
	
	\section{Model description}
	
	\subsection{Data transformation}
	The ID$_3$ index represents high volatility and exhibits price spikes, which can be seen in Figure \ref{fig:ID3overtime}. This may lead to biased model estimation and inaccurate forecasts. \cite{Uniejewski2018} have shown that the usage of the variance stabilizing methods results in higher quality forecasts. Thus, before the model estimation, we apply a median normalization and an asinh transformation to stabilize the variance. These are not new to the electricity price forecasting and are used in many research papers.
	
	The median normalization of price $P_{d,s}$ is given by the formula 
	\begin{equation}
	p_{d,s} = \frac{1}{\text{MAD}(\mathbf{P}_{s})/z_{0.75}}(P_{d,s} - \text{Med}(\mathbf{P}_{s})),
	\label{eq:mediannorm}
	\end{equation}
	where $\text{Med}(\mathbf{P}_{s})$ is the median of $P_{d,s}$ in the $D =365$-day calibration sample, $\text{MAD}(\mathbf{P}_{s})$ is the median absolute deviation around the sample median in the calibration sample and $z_{0.75}$ is the $75\%$ quantile of the standard normal distribution. Here we adjust the MAD dividing it by $z_{0.75}$ to ensure its asymptotical consistency to the standard deviation. We use the median normalization because of its robustness, which is useful when dealing with heavy tailed data.
	
	The area hyperbolic sine (asinh) transformation is given by the formula
	\begin{equation}
    Y_{d,s} = \text{asinh}(p_{d,s}) = \log \left(p_{d,s} + \sqrt{p_{d,s}^2+1}\right),
	\label{eq:asinh}
	\end{equation}
	where $p_t$ is the normalized price. If we worked on a market with strictly positive electricity prices, a logarithmic transformation would be sufficient. Since the German electricity market allows for negative prices, this is no longer an option. The asinh transformation can handle all real values and has a logarithmic tail behaviour as the $\log$, but for both positive and negative values. Thus, it solves all issues concerning the heavy tails in the data.
	 In the paper of \cite{Uniejewski2018} it is shown that, considering the quality of the forecasts, it performs very well among the other variance stabilization methods. The only problem with this transformation is that it is non-linear and the backward transformation is not that obvious as usual. To be specific, as \cite{Uniejewski2018} mentioned, $\sinh (\mathbb{E}(Y_{d,s})) \neq \mathbb{E} (\sinh(Y_{d,s})) = \mathbb{E} (p_{d,s})$. In the literature (e.g. \citealp{uniejewski2018understanding}; \citealp{ziel2018day}) this problem is often either ignored or it is assumed, that the values of $Y_{d,s}$ are close to 0, where the asinh is approximately linear. In this paper we take two approaches to this problem: in the first one, we do it the mathematically incorrect way, i.e. we assume that
	\begin{equation}
		\widehat{\mathbb{E}(P_{d,s})} = \widehat{\mathbb{E} (p_{d,s})} \cdot \widehat{b} + \widehat{a}  \approx  \sinh (\widehat{\mathbb{E}(Y_{d,s})}) \cdot \widehat{b} + \widehat{a}, 
		\label{eq:incorrectway}
	\end{equation}
	where $\widehat{\mathbb{E}(Y_{d,s})}$ is the corresponding forecast of $Y_{d,s}$, $\widehat{b}$ and $\widehat{a}$ are the adjusted sample MAD and the sample median from the equation (\ref{eq:mediannorm}), respectively. In the second approach, we do it the correct way
	\begin{equation}
		\widehat{\mathbb{E}(P_{d,s})} = \widehat{\mathbb{E} (p_{d,s})} \cdot  \widehat{b} + \widehat{a} =  \int \sinh(x) d \widehat{F}_{Y_{d,s}} \cdot  \widehat{b} + \widehat{a},
		\label{eq:correctway}
	\end{equation}
	where $\widehat{F}_{Y_{d,s}}$ is an empirical cumulative distribution function of $Y_{d,s}$. We estimate the cumulative distribution function by $\widehat{F}_{Y_{d,s}}(t)  =  \frac{1}{D} \sum_{j=-D+d}^{d-1} \mathds{1}_{\widehat{\mathbb{E}(Y_{d,s})} + \widehat{\varepsilon}_{j,s} \le t}$, where  $\widehat{\varepsilon}_{j,s}$ are the in-sample residuals. Therefore, the correct backward transformation (\ref{eq:correctway}) comes to
	\begin{equation}
		\widehat{\mathbb{E}(P_{d,s})} = \int \sinh(x) d \widehat{F}_{Y_{d,s}} \cdot  \widehat{b} + \widehat{a} = \frac{1}{D} \sum_{j = -D + d}^{d-1} \sinh\left(\widehat{\mathbb{E}(Y_{d,s})} + \widehat{\varepsilon}_{j,s}\right)  \cdot  \widehat{b} + \widehat{a}.
	\end{equation}
	
	\subsection{Estimation techniques}\label{sec:estimationtechnique}
	In the following paper, we consider only linear models, thus we utilize 3 estimation methods: the ordinary least squares (OLS), the least absolute shrinkage and selection operator (lasso) and the elastic net, which is a linear combination of the lasso and ridge regressions. We use the OLS estimation only with very simple models, while the lasso and elastic net methods with more complex models that contain a very big number of regressors. The OLS is a standard, well-known estimation method of linear models, therefore we focus our attention on the latter ones.
	
	The lasso method, which was introduced by \cite{Tibshirani1996}, is a regularized model estimation technique. It is often used in the literature in the sake of variable selection, e.g. by \cite{Ziel2016}, \cite{Uniejewski2018efficient} or \cite{uniejewski2018understanding}. Thanks to the lasso's shrinkage property, we can easily handle models with many parameters. Let us assume that we possess a model in an OLS representation as following
	\begin{equation}
	Y_{d,s} = \boldsymbol{X}^{'}_{d,s} \boldsymbol{\beta}_{s}  + \varepsilon_{d,s},
	\label{eq:modelrewritten}
	\end{equation}
	where  $\boldsymbol{X}^{'}_{d,s}$ is a vector of the input regressors and $\boldsymbol{\beta}_{s}$ is a vector of the corresponding coefficients. Let us note that we perform a median normalization and an asinh transformation on all input regressors $\boldsymbol{X}^{'}_{d,s}$ and for all regressors we calculate the MAD around the sample median in the calibration sample excluding those observations that are equal to the corresponding median. This operation does not change much for the continuous variables, but helps to preserve the dummy variables.  Since the lasso technique requires the regressors to be additionally standardized, i.e. with 0 mean and the variance equal to 1, we introduce it with
	\begin{equation}
	Y_{d,s} = \widetilde{\boldsymbol{X}}^{'}_{d,s} \widetilde{\boldsymbol{\beta}}_{s}  + \widetilde{\varepsilon}_{d,s}.
	\label{eq:modellassostd}
	\end{equation}
	We perform this scaling using the corresponding sample mean and standard deviation. Having $\widetilde{\boldsymbol{\beta}}_{s}$, we can easily calculate $\boldsymbol{\beta}_{s}$ of (\ref{eq:modelrewritten}) by rescaling. The lasso estimation method is a penalized regression approach that uses an $L_1$ penalty on the number of parameters. Let D be the number of observable days. Then the lasso estimator $\widehat{\widetilde{\beta}}_s^{\text{lasso}}$ is given by
	\begin{equation}
	\widehat{\widetilde{\beta}}_s^{\text{lasso}} = \arg\min_{\boldsymbol{\beta}} \left\{ \sum_{d = 1}^{D}(Y_{d,s} - \widetilde{\boldsymbol{X}}^{'}_{d,s} \boldsymbol{\beta})^2 + \lambda_s \sum_{i =1}^{p} |\beta_i| \right \},
	\label{eq:lasso}
	\end{equation}
	where $\lambda_s$ is a tunable parameter and $p$ is a number of regressors.
	
	The elastic net method, which was introduced by \cite{Zou05regularizationand}, can be considered as a correction of the lasso method that overcomes some of the latter's limitations. The difference to the lasso estimator is that the elastic net linearly combines the $L_1$ and $L_2$ penalties of the lasso and ridge methods. The elastic net estimator $\widehat{\widetilde{\beta}}_s^{\text{elnet}}$ is given by
	\begin{equation}
	\widehat{\widetilde{\beta}}_s^{\text{elnet}} = \arg\min_{\boldsymbol{\beta}} \left\{ \sum_{d = 1}^{D}(Y_{d,s} - \widetilde{\boldsymbol{X}}^{'}_{d,s} \boldsymbol{\beta})^2 +  \alpha \lambda_s \sum_{i =1}^{p} |\beta_i| + \frac{1-\alpha}{2} \lambda_s \sum_{i = 1 }^{p} \beta_i^2 \right\},
	\end{equation}
	where $\alpha$ is an elastic net mixing parameter. Let us note that if we set $\alpha = 1$, then the estimator $\widehat{\widetilde{\beta}}_s^{\text{elnet}}$ becomes in fact the lasso one. Subsequently, we fix the $\alpha$ parameter to $0.5$, so the elastic net method uses the lasso and ridge penalties evenly.
	
	A crucial parameter for the lasso and elastic net estimators is the $\lambda_s$. The larger the value is, the more variables are included in the model, so a proper tuning exercise of this parameter is essential. In the literature appear many approaches, but we utilize the one described by \cite{Ziel2016}. That is to say, since the estimation algorithm is very fast, we utilize an exponential grid $\Lambda = \{ \lambda_i = 2^i|i \in \mathcal{G} \}$, where $\mathcal{G}$ is an equidistant grid from -10 to 4 of length 100 and for each out-of-sample iteration we compute the model for all $\lambda_i$. In each iteration, we choose the tuning parameter $\lambda_s \in \Lambda$ based on the minimization of the Bayesian information criterion (BIC). As the BIC is regarded as a conservative information criterion, it is suitable for the high-dimensional regression setting that we are considering. For the implementation of the lasso and elastic net methods we use the \verb|R| package \verb|glmnet| developed by \cite{friedman2010glmnet}.
	
	\subsection{Full information models}
	As we mentioned before, our goal is to build a model that uses all the information that is available on the market at the time of forecasting, which is in our case 3 hours and 15 minutes before the delivery. To be specific, we want to forecast the value of the ID$_3$ just before its time interval. For this purpose, we construct for each product a linear model
	\begin{equation}
	\label{eq:ourmodel}
	\begin{aligned}
	\text{ID}_3^{d,s}  = &  \sum_{j \in \{-1,0,1\}} \sum_{k = 1}^{24} \sum_{x \in \mathcal{I}_{H}(j,k)} \beta_{j,k,x}^{(1)} \: \prescript{H}{x}{\text{ID}}_{0.25}^{d-j, k}  + \sum_{j = 2}^{14} \sum_{k=1}^{24} \beta_{j,k}^{(2)} \: \prescript{H}{}{\text{ID}}_{3}^{d-j, k} \\
	& + \sum_{j \in \{-1,0,1\}} \sum_{k = 1}^{96} \sum_{x \in \mathcal{I}_{QH}(j,k)} \beta_{j,k,x}^{(3)} \: { } \prescript{QH}{x}{\text{ID}}_{0.25}^{d-j, k}  + \sum_{j = 2}^{14} \sum_{k=1}^{96} \beta_{j,k}^{(4)} \: \prescript{QH}{}{\text{ID}}_{3}^{d-j, k} \\
	& + \sum_{j = -1}^{14} \sum_{k = 1}^{24} \beta_{j,k}^{(5)} \text{DA}^{d-j, k} + \sum_{j = -1}^{14} \sum_{k = 1}^{96} \beta_{j,k}^{(6)} \text{IA}^{d-j, k} + \sum_{j=1}^{7} \beta_j^{(7)}\text{DoW}_j^d   \\ 
	&+ \sum_{j = 0}^{14} \sum_{k = 1}^{96} \beta_{j,k}^{(8)} \text{BV}^{d-j, k} + \varepsilon^{d,s},
	\end{aligned}
	\end{equation}
	where $\mathcal{I}_{i}(j,k) = \{ 0, 0.25, 0.5, \dots, b(d-j,k) - c_i(d-j) - 0.25, b(d-j,k) - c_i(d-j) \}$ for $i \in \{H,QH\}$ and $c_i(d)$ stands for the beginning of trading of a product type $i$ on day $d$. Model (\ref{eq:ourmodel}) consists of 8 main components, excluding an error term. The first one is a set of all $_{x}\text{ID}_{0.25}$ values between the beginning of trading and the time of forecasting. That is to say, we split this time frame, similarly as is presented in Figure \ref{fig:xIDyproperty}, but we use denser, 15 minutes grid. We perform this for all hourly products of the previous day, the current day, and of the next day, if the trading of the following day products has already begun, i.e. after 15:00 the current day.

	\begin{sidewaysfigure}
		\includegraphics[width=1.02\linewidth, height=.4\linewidth]{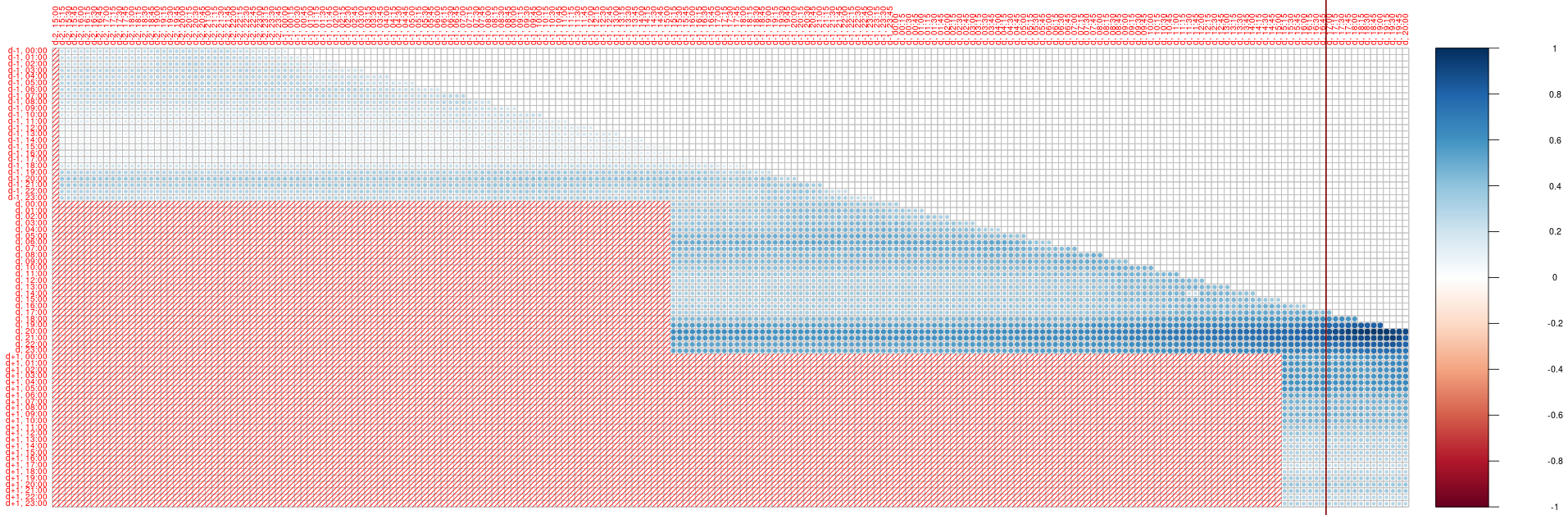}
		\caption{Correlation between the ID$_3$ of an hourly product with the delivery at 20:00 and $_{x}\text{ID}_{0.25}$ of hourly products with delivery times specified on the Y-axis. Time points $t$ are depicted on the X-axis, and they determine the value of $x$ lag. Each dot represents a single correlation at a corresponding time point.}
		\label{fig:ID3corplot}
	\end{sidewaysfigure}
	
	An illustration of this set is shown in Figure \ref{fig:ID3corplot}. This figure presents a correlation between the ID$_3$ of an hourly product with the delivery at 20:00 and all available values of the above described $_{x}\text{ID}_{0.25}$. In red, we crossed out the time points, for which the trading has not yet begun. White triangle in the top-right area of the plot indicates the time points, for which the trading is already finished. In this example, we would forecast at 16:45, that is 3 hours and 15 minutes before the delivery. The time of forecasting is highlighted by a dark-red vertical line. Everything that is on the left of this line, is the information that we get from the first component of our model. Let us note that in this example trading for the following day has begun. Thus, we can use in the model the first volume-weighted prices of the following day. Let us also notice that the correlations are quite high, especially if we look at the most recent value of the product with the delivery at 20:00, which is $_{3.25}\text{ID}_{0.25}$.
	
	The second component of the model are the values of ID$_3$ of all hourly products for all days between two weeks and two days before the current day. These values can supply our model with some weekly or bi-weekly information. The third and fourth components of the model are quarter-hourly equivalents of the first and second components, respectively. The fifth component is a set of the Day-Ahead Auction results. It consists of the day-ahead prices of all hourly products for the last two weeks, the current day and the following day if they are already announced. The sixth component of the model is an Intraday Auction equivalent of the fifth one and the seventh component consists of the day of the week dummies. 
	
	The eighth component is a set of quarter-hourly balancing volumes. Market participants often look at the balancing market when trading in the intraday market, so some information from this market may be useful for our model. Thus, we include the quarter-hourly balancing volumes from two weeks to 30 minutes before the time of forecasting. As a balancing volume, we use the sum of imbalances of all German Transmission System Operators. This data is published every quarter-hour, 15 minutes after the end of the delivery, e.g. the imbalance volume for the delivery between 16:15 and 16:30 is announced at 16:45. Therefore, in the model, the most recent balancing volume is the one with the start of the delivery 30 minutes before the time of forecasting. Let us note that the reason for using balancing volumes instead of balancing prices is the fact that the latter ones are published with a delay of two months (\citealp{Viehmann2017}). It is worth to mention that in the Intraday Continuous market, there are time-to-maturity effects which affect heavily the price formation. However, we are only interested forecasting the ID$_3$, thus the time-to-maturity is always the same in this study, i.e. 3 hours and 15 minutes. For that reason, we cannot use it as a regressor.
	
	Each of the aforementioned components supplies our model with a very big amount of regressors, especially the $_{x}\text{ID}_{0.25}$ components. Due to this, the model contains at least 16580 explanatory variables (in the case of products with delivery at 00:00) and at most~26259 (in the case of a quarter-hourly product with delivery at 23:45). Thus, in the purpose of the model estimation, we utilize the lasso and elastic net methods. We expect the most recent price of the corresponding product $s$, i.e. ${}_{3.25}^{ }\text{ID}_{0.25}^{d,s}$ to be the most informative for the model, so we want to favour this regressor. Therefore, we perform the model estimation in three ways. The first way, we do not penalize the model for the size of the corresponding coefficient $\widehat{\beta}^{(i)}_{0,s,3.25}$, where $i \in \{1,3\}$, depending on the product type. The second way, we fix the corresponding coefficient to 1, i.e. $\widehat{\beta}^{(i)}_{0,s,3.25} =1$. The third way, we do not interfere in the coefficient estimation. 	To summarize, we estimate our model in 3 ways, using 2 methods and 2 approaches to the backward transformation. In total, this gives us 12 versions of this model. We abbreviate them with \textbf{FI.X.Y.Z}, where \textbf{FI} stands for full information, \textbf{X} $\in \{\text{lasso},\text{elnet}\}$ indicate the estimation method, \textbf{Y} $\in \{\text{notpen},\text{penal}, \text{fixed} \}$ describe the way of the coefficient estimation and \textbf{Z} $\in \{\text{IC},\text{C}\}$ indicate the approach to the backward transformation of the asinh. For instance, \textbf{FI.lasso.penal.C} stands for the full information model estimated using the lasso with a standard penalty and correctly back-transformed.
	
	\subsection{Benchmark models}
	The first benchmark model that we utilize is the corresponding day-ahead price, i.e. 
	\begin{equation}
		\widehat{\text{ID}}_3^{d,s} = \text{DA}^{d,s} \text{ or } \widehat{\text{ID}}_3^{d,s} = \text{IA}^{d,s},
	\end{equation}
	depending on the product type. We denote it as \textbf{Naive.DA}. This model is based on the assumption that both the intraday prices and the day-ahead prices are determined by the same factors. It has been already used as a benchmark model in the forecasting of intraday prices by \cite{uniejewski2018understanding}. The second benchmark model is a new approach, but at the same time, a very intuitive one. To be specific, we define it by the aforementioned most recent 15-minutes price of the corresponding product, i.e.
	\begin{equation}
		\widehat{\text{ID}}_3^{d,s} = {}_{3.25}^{ }\text{ID}_{0.25}^{d,s}.
		\label{eq:naiveMR}
	\end{equation}
	Based on Figure \ref{fig:ID3corplot} we expect it to be a good benchmark model, and we denote it as \textbf{Naive.MR1}. The third benchmark model is a little modification to the second one. We take into account 2.5 hours of the most recent transactions instead of 15 minutes. We denote it by \textbf{Naive.MR2} and the model formula is given by 
	\begin{equation}
		\widehat{\text{ID}}_3^{d,s} = {}_{3.25}^{ }\text{ID}_{2.5}^{d,s}.
	\end{equation}
	
	The fourth benchmark model is based on the \textbf{ARX}  model by \cite{uniejewski2018understanding}, which, on the other hand, is inspired by the \textbf{expert$_{\textbf{DoW,nl}}$} model of \cite{ziel2018day} and is given by the formula
	\begin{equation}
	\begin{aligned}
		\text{ID}_3^{d,s} = &  \beta_1 \text{ID}_3^{d-\mathds{1}_{(s \le 4)},(s-4)\bmod S +1} + \beta_2 \text{ID}_3^{d-1,s} + \beta_3 \text{ID}_3^{d-2,s} + \beta_4 \text{ID}_3^{d-7,s} + \beta_5 {}_{3.25}^{ }\text{ID}_{0.25}^{d,s}\\ & + \beta_6 \text{DA}^{d,s} + \sum_{j = 1}^{7} \beta_{6+j}\text{DoW}_j^d + \varepsilon^{d,s},
	\end{aligned}
	\label{eq:ARX}
	\end{equation}
	where $S$ is a number of products of a given type. Model (\ref{eq:ARX}) holds for hourly products, in the case of quarter-hourly products we swap $\text{DA}^{d,s}$ with $\text{IA}^{d,s}$. Let us note that $\text{ID}_3^{d-\mathds{1}_{(s \le 4)},(s-4)\bmod S +1}$ is the most recent observed $\text{ID}_3$ value at the time of forecasting. $\text{ID}_3^{d-1,s}$, $\text{ID}_3^{d-2,s}$ and $\text{ID}_3^{d-7,s}$ account for the autoregressive effects of the previous days, i.e. the same product yesterday, two days ago and a week ago. The difference between our modification of the \textbf{ARX} model and the one used by \cite{uniejewski2018understanding} is the most recent value. We consider three versions of the \textbf{ARX} model. In order to understand the importance of modelling using transformed data, we apply this model to non-transformed prices and asinh-transformed prices with an incorrect and a correct backward transformation. We denote them by \textbf{ARX.non}, \textbf{ARX.asinhIC} and \textbf{ARX.asinhC}, respectively.
	
	The last, but not the least benchmark is a slightly modified lasso-estimated model of \cite{uniejewski2018understanding}. The mentioned modification is an addition of the respective quarter-hourly products to its formula. The research of \cite{uniejewski2018understanding} was done only for the hourly products, therefore the quarter-hourly ones are not considered there. Since the lasso technique can easily handle a big amount of regressors, instead of constructing two separate models for different product types, we consider one that contains more information. The formula is as follows
	\begin{equation}
	\begin{aligned}
	\text{ID}_3^{d,s} = & \sum_{i, j = d-7, s}^{d, s-4}\beta_{i,j}^{(1)}\prescript{H}{}{\text{ID}}_3^{i,j} + \sum_{i, j = d-7, s}^{d+1, 24}\beta_{i,j}^{(2)}\text{DA}^{i,j} + \sum_{j=1}^{7} \beta_j^{(3)}\text{DoW}_j^d\\ 
	& + \sum_{i, j = d-7, s}^{d, s-13}\beta_{i,j}^{(4)}\prescript{QH}{}{\text{ID}}_3^{i,j}  + \sum_{i, j = d-7, s}^{d+1, 96}\beta_{i,j}^{(5)}\text{IA}^{i,j} + \varepsilon^{d,s}.
	\end{aligned}
	\end{equation}
	The original notation was adjusted, so it fits our convention. Let us note that this model uses the values of the Day-Ahead and Intraday Auction prices if they are already published, similarly as we do in the full information model. Due to more regressors present in the model, to avoid overestimation, we do not use their best performing model, i.e. \textbf{LASSO($\mathbf{\lambda_6}$)}, where $\mathbf{\lambda_6} = 10^{-\frac{13}{6}}$. Instead, we tune the $\lambda_s$ parameter as described in Section \ref{sec:estimationtechnique}, but using the same $\Lambda$ grid as in \cite{uniejewski2018understanding}, i.e. $\Lambda = \{10^{-\frac{19-i}{6}}| i \in \{1,\dots, 10\}\}$. We denote this model as \textbf{Lasso.AR}.
	\section{Forecasting Study and Evaluation}
		
	Since we deal with time-series data, we utilize a rolling window scheme. This approach is taken in the majority of electricity price forecasting studies. For a meaningful forecast evaluation, we consider a $D = 365$-day window size. The initial data range is highlighted with red lines in Figure \ref{fig:ID3overtime}. In our research, we apply only multivariate models, which in our case results in 120 models (24 for hourly and 96 for quarter-hourly products). The aim of this paper is a very short term forecasting, i.e. we want to forecast the ID$_3$-Price 3 hours and 15 minutes before the delivery of the corresponding product. This means that for each product we fit a model to the data from 01.01.2015 to 31.12.2015 and we forecast the ID$_3$-Price for the next day. Then we move our window forward by one day and repeat the exercise until the end of the out-of-sample set. Our out-of-sample data span the date range from 01.01.2016 to 29.09.2018, which gives us $N = 1003$ days of meaningful forecasts.
	
	We utilize the mean absolute error (MAE) and the root mean squared error (RMSE) as the forecasting measures. The RMSE is the optimal least square problems measure, yet it is sensitive to outliers. Thus, we apply also the MAE, which is more robust, but it is designed to measure the performance of forecasting the median, while we forecast the mean. The MAE and the RMSE are given by
	
	\begin{equation}
		\text{MAE} = \frac{1}{S\cdot N}\sum_{i = 1}^{N} \sum_{s=1}^{S} \left|\text{ID}_3^{D+i,s} - \widehat{\text{ID}}_3^{D+i,s} \right|,
	\end{equation}
	
	\begin{equation}
		\text{RMSE} = \sqrt{\frac{1}{S \cdot N} \sum_{i = 1}^{N} \sum_{s=1}^{S} \left|\text{ID}_3^{D+i,s} - \widehat{\text{ID}}_3^{D+i,s} \right|^2}.
	\end{equation}
	Let us note that we evaluate the forecasts separately for hourly and quarter-hourly products. Moreover, in purpose of better understanding models' performance over the day, we calculate also the MAE and the RMSE for each product. We denote them by $\text{MAE}_s$ and $\text{RMSE}_s$. They are defined as follows
	\begin{equation}
	\text{MAE}_s = \frac{1}{N}\sum_{i = 1}^{N} \left|\text{ID}_3^{D+i,s} - \widehat{\text{ID}}_3^{D+i,s} \right|,
	\end{equation}
	
	\begin{equation}
	\text{RMSE}_s = \sqrt{\frac{1}{N} \sum_{i = 1}^{N} \left|\text{ID}_3^{D+i,s} - \widehat{\text{ID}}_3^{D+i,s} \right|^2}.
	\end{equation}
	
	The RMSE and the MAE are widely used in the EPF literature to provide a ranking of models, e.g. by \cite{Ziel2016} or by \cite{uniejewski2018understanding}. However, these cannot draw statistically significant conclusions on the outperformance of the forecasts of the considered models. Therefore, we also calculate the \cite{diebold1995comparing} test, which tests forecasts of model $A$ against forecasts of model $B$. In the following paper, we compute the multivariate version of the DM test as in \cite{ziel2018day}. This is in contrast to the majority of the EPF literature, where the DM test is performed separately for each product, see \cite{weron2014electricity}. The multivariate DM test results in only one statistic for each model that is computed based on the $S$-dimensional vector of errors for each day, where $S \in \{24,96\}$ for hourly and quarter-hourly products, respectively. Therefore, denote $\widehat{\varepsilon}_{A,d} = [\widehat{\varepsilon}_{A,d,1}, \widehat{\varepsilon}_{A,d,2}, \dots, \widehat{\varepsilon}_{A,d,S}]'$ and $\widehat{\varepsilon}_{B,d} = [\widehat{\varepsilon}_{B,d,1}, \widehat{\varepsilon}_{B,d,2}, \dots, \widehat{\varepsilon}_{B,d,S}]'$ the vectors of the out-of-sample errors for day d of the models $A$ and $B$, respectively. The multivariate loss differential series 
	\begin{equation}
		\Delta_{A,B,d} = ||\widehat{\varepsilon}_{A,d}||_i  - ||\widehat{\varepsilon}_{B,d}||_i  
	\end{equation}
	defines the difference of errors in $||\cdot||_i$ norm, i.e. $||\widehat{\varepsilon}_{A,d}||_i = \left(\sum_{s=1}^{S} |\widehat{\varepsilon}_{A,d,s}|^i \right)^{1/i} $, where $i \in \{1,2\}$. For each model pair, we compute the $p$-value of two one-sided DM tests. The first one is with the null hypothesis $\mathcal{H}_0: \mathbb{E}(\Delta_{A,B,d}) \le 0$, i.e. the outperformance of the forecasts of model B by the forecasts of model A. The second test is with the reverse null hypothesis $\mathcal{H}_0^R: \mathbb{E}(\Delta_{A,B,d}) \ge 0$, i.e. the outperformance of the forecasts of model A by those of model B. Let us note that these tests are complementary, and we perform them using two norms: $||\cdot||_1$ and $||\cdot||_2$. Naturally, we assume that the loss differential series is covariance stationary.

	For a better understanding of the full information model, we perform a coefficient analysis of the best performing one. We can easily study the variable selection thanks to the lasso and elastic net shrinkage and regularization properties. As the measure of the importance of a standardized parameter $\widetilde{\beta}_{s,i}$ we consider the fraction of the absolute standardized parameter of a model to the sum of all absolute standardized parameters as in \cite{Ziel2016}
	\begin{equation}
		\iota_{s,i} = \frac{|\widetilde{\beta}_{s,i}|}{\sum_{j=1}^{p}|\widetilde{\beta}_{s,j}|}.
	\end{equation}
	Naturally, $0 \le \iota_{s,i} \le 1$ and $\sum_{i=1}^{p}\iota_{s,i} = 1$. The larger $\iota_{s,i}$, the larger the relative impact of the corresponding parameter to the ID$_3$-Price of product $s$. We estimate the values of $\iota_{s,i}$ by applying the plug-in principle to the estimators $\widehat{\widetilde{\beta}}_s$ and compute the corresponding sample mean across the rolling window. 
	
		\section{Results and Discussion}
	\subsection{Forecast evaluation}
	
	\begin{table}[ht!]
		\centering
		\begingroup\small
		\begin{tabular}{rrrrr}
			\hline
			& $\text{MAE}_H$ & $\text{RMSE}_H$ & $\text{MAE}_{QH}$ & $\text{RMSE}_{QH}$ \\ 
			\hline
			\textbf{Naive.DA} & \cellcolor[rgb]{1,0.621,0.5} {5.0042 \scriptsize{(.04)}} & \cellcolor[rgb]{1,0.621,0.5} {7.9653 \scriptsize{(.146)}} & \cellcolor[rgb]{1,0.621,0.5} {7.643 \scriptsize{(.028)}} & \cellcolor[rgb]{1,0.621,0.5} {11.701 \scriptsize{(.122)}} \\ 
			\textbf{Naive.MR1} & \cellcolor[rgb]{0.503,0.901,0.5} {\underline{3.3343} \scriptsize{(.026)}} & \cellcolor[rgb]{0.504,0.901,0.5} {\underline{5.278} \scriptsize{(.113)}} & \cellcolor[rgb]{1,0.589,0.5} {7.706 \scriptsize{(.029)}} & \cellcolor[rgb]{1,0.616,0.5} {11.714 \scriptsize{(.185)}} \\ 
			\textbf{Naive.MR2} & \cellcolor[rgb]{0.789,0.996,0.5} {3.5006 \scriptsize{(.028)}} & \cellcolor[rgb]{0.823,1,0.5} {5.5877 \scriptsize{(.12)}} & \cellcolor[rgb]{1,0.689,0.5} {7.511 \scriptsize{(.028)}} & \cellcolor[rgb]{1,0.725,0.5} {11.433 \scriptsize{(.178)}} \\ 
			\textbf{ARX.non} & \cellcolor[rgb]{0.794,0.998,0.5} {3.5033 \scriptsize{(.027)}} & \cellcolor[rgb]{0.68,0.96,0.5} {\underline{5.4421} \scriptsize{(.099)}} & \cellcolor[rgb]{1,0.964,0.5} {6.977 \scriptsize{(.026)}} & \cellcolor[rgb]{0.655,0.952,0.5} {\underline{10.529} \scriptsize{(.121)}} \\ 
			\textbf{ARX.asinhIC} & \cellcolor[rgb]{0.688,0.963,0.5} {3.4416 \scriptsize{(.027)}} & \cellcolor[rgb]{0.733,0.978,0.5} {\underline{5.4925} \scriptsize{(.116)}} & \cellcolor[rgb]{0.9,1,0.5} {6.858 \scriptsize{(.027)}} & \cellcolor[rgb]{0.855,1,0.5} {\underline{10.627} \scriptsize{(.128)}} \\ 
			\textbf{ARX.asinhC} & \cellcolor[rgb]{0.679,0.96,0.5} {3.4364 \scriptsize{(.028)}} & \cellcolor[rgb]{0.658,0.953,0.5} {\underline{5.4217} \scriptsize{(.109)}} & \cellcolor[rgb]{0.979,1,0.5} {6.896 \scriptsize{(.024)}} & \cellcolor[rgb]{0.5,0.9,0.5} {\textbf{10.462} \scriptsize{(.118)}} \\ 
			\textbf{Lasso.AR} & \cellcolor[rgb]{1,0.776,0.5} {4.4619 \scriptsize{(.037)}} & \cellcolor[rgb]{1,0.766,0.5} {7.1529 \scriptsize{(.156)}} & \cellcolor[rgb]{1,0.896,0.5} {7.109 \scriptsize{(.027)}} & \cellcolor[rgb]{1,0.847,0.5} {11.116 \scriptsize{(.131)}} \\ 
			\textbf{FI.lasso.notpen.IC} & \cellcolor[rgb]{0.571,0.924,0.5} {\underline{3.3738} \scriptsize{(.027)}} & \cellcolor[rgb]{0.655,0.952,0.5} {\underline{5.4191} \scriptsize{(.115)}} & \cellcolor[rgb]{1,0.942,0.5} {7.02 \scriptsize{(.027)}} & \cellcolor[rgb]{1,0.93,0.5} {10.902 \scriptsize{(.131)}} \\ 
			\textbf{FI.lasso.notpen.C} & \cellcolor[rgb]{0.537,0.912,0.5} {\underline{3.3538} \scriptsize{(.026)}} & \cellcolor[rgb]{0.552,0.917,0.5} {\underline{5.3229} \scriptsize{(.111)}} & \cellcolor[rgb]{0.968,1,0.5} {6.891 \scriptsize{(.026)}} & \cellcolor[rgb]{0.688,0.963,0.5} {\underline{10.543} \scriptsize{(.131)}} \\ 
			\textbf{FI.lasso.fixed.IC} & \cellcolor[rgb]{0.5,0.9,0.5} {\textbf{3.3325} \scriptsize{(.027)}} & \cellcolor[rgb]{0.5,0.9,0.5} {\textbf{5.2739} \scriptsize{(.112)}} & \cellcolor[rgb]{1,0.664,0.5} {7.56 \scriptsize{(.028)}} & \cellcolor[rgb]{1,0.691,0.5} {11.52 \scriptsize{(.179)}} \\ 
			\textbf{FI.lasso.fixed.C} & \cellcolor[rgb]{0.511,0.904,0.5} {\underline{3.3391} \scriptsize{(.027)}} & \cellcolor[rgb]{0.501,0.9,0.5} {\underline{5.2751} \scriptsize{(.111)}} & \cellcolor[rgb]{1,0.674,0.5} {7.541 \scriptsize{(.028)}} & \cellcolor[rgb]{1,0.704,0.5} {11.486 \scriptsize{(.178)}} \\ 
			\textbf{FI.lasso.penal.IC} & \cellcolor[rgb]{1,0.984,0.5} {3.7364 \scriptsize{(.033)}} & \cellcolor[rgb]{1,0.921,0.5} {6.2817 \scriptsize{(.144)}} & \cellcolor[rgb]{1,0.934,0.5} {7.035 \scriptsize{(.028)}} & \cellcolor[rgb]{1,0.795,0.5} {11.251 \scriptsize{(.132)}} \\ 
			\textbf{FI.lasso.penal.C} & \cellcolor[rgb]{0.884,1,0.5} {3.5803 \scriptsize{(.031)}} & \cellcolor[rgb]{1,0.968,0.5} {6.0132 \scriptsize{(.141)}} & \cellcolor[rgb]{0.525,0.908,0.5} {\underline{6.72} \scriptsize{(.027)}} & \cellcolor[rgb]{0.973,1,0.5} {10.703 \scriptsize{(.137)}} \\ 
			\textbf{FI.elnet.notpen.IC} & \cellcolor[rgb]{0.572,0.924,0.5} {\underline{3.3741} \scriptsize{(.027)}} & \cellcolor[rgb]{0.655,0.952,0.5} {\underline{5.4194} \scriptsize{(.113)}} & \cellcolor[rgb]{1,0.79,0.5} {7.314 \scriptsize{(.028)}} & \cellcolor[rgb]{1,0.764,0.5} {11.33 \scriptsize{(.127)}} \\ 
			\textbf{FI.elnet.notpen.C} & \cellcolor[rgb]{0.537,0.912,0.5} {\underline{3.3543} \scriptsize{(.027)}} & \cellcolor[rgb]{0.552,0.917,0.5} {\underline{5.323} \scriptsize{(.115)}} & \cellcolor[rgb]{1,0.857,0.5} {7.185 \scriptsize{(.027)}} & \cellcolor[rgb]{1,0.909,0.5} {10.957 \scriptsize{(.125)}} \\ 
			\textbf{FI.elnet.fixed.IC} & \cellcolor[rgb]{0.501,0.9,0.5} {\underline{3.3331} \scriptsize{(.026)}} & \cellcolor[rgb]{0.5,0.9,0.5} {\underline{5.2742} \scriptsize{(.111)}} & \cellcolor[rgb]{1,0.638,0.5} {7.61 \scriptsize{(.028)}} & \cellcolor[rgb]{1,0.665,0.5} {11.588 \scriptsize{(.171)}} \\ 
			\textbf{FI.elnet.fixed.C} & \cellcolor[rgb]{0.513,0.904,0.5} {\underline{3.3401} \scriptsize{(.026)}} & \cellcolor[rgb]{0.502,0.901,0.5} {\underline{5.2754} \scriptsize{(.113)}} & \cellcolor[rgb]{1,0.644,0.5} {7.599 \scriptsize{(.028)}} & \cellcolor[rgb]{1,0.674,0.5} {11.565 \scriptsize{(.177)}} \\ 
			\textbf{FI.elnet.penal.IC} & \cellcolor[rgb]{1,0.998,0.5} {3.6878 \scriptsize{(.033)}} & \cellcolor[rgb]{1,0.952,0.5} {6.1059 \scriptsize{(.141)}} & \cellcolor[rgb]{1,0.951,0.5} {7.002 \scriptsize{(.027)}} & \cellcolor[rgb]{1,0.839,0.5} {11.138 \scriptsize{(.126)}} \\ 
			\textbf{FI.elnet.penal.C} & \cellcolor[rgb]{0.859,1,0.5} {3.5588 \scriptsize{(.03)}} & \cellcolor[rgb]{1,0.994,0.5} {5.8723 \scriptsize{(.131)}} & \cellcolor[rgb]{0.5,0.9,0.5} {\textbf{6.712} \scriptsize{(.026)}} & \cellcolor[rgb]{0.852,1,0.5} {\underline{10.625} \scriptsize{(.129)}} \\ 
			\hline
		\end{tabular}
		\endgroup
		\caption{MAE and RMSE values for the considered models. The corresponding estimated standard deviations are given in parenthesis. Bolded values indicate the lowest error in each column. The models that are not significantly worse than the best (indicated by the 2-sigma range of the best model) are underlined.}
		\label{tab:maermse}
	\end{table}
	
	\begin{figure*}[b!]
		\centering
		\subfloat[]{
			\includegraphics[width = 0.475\linewidth]{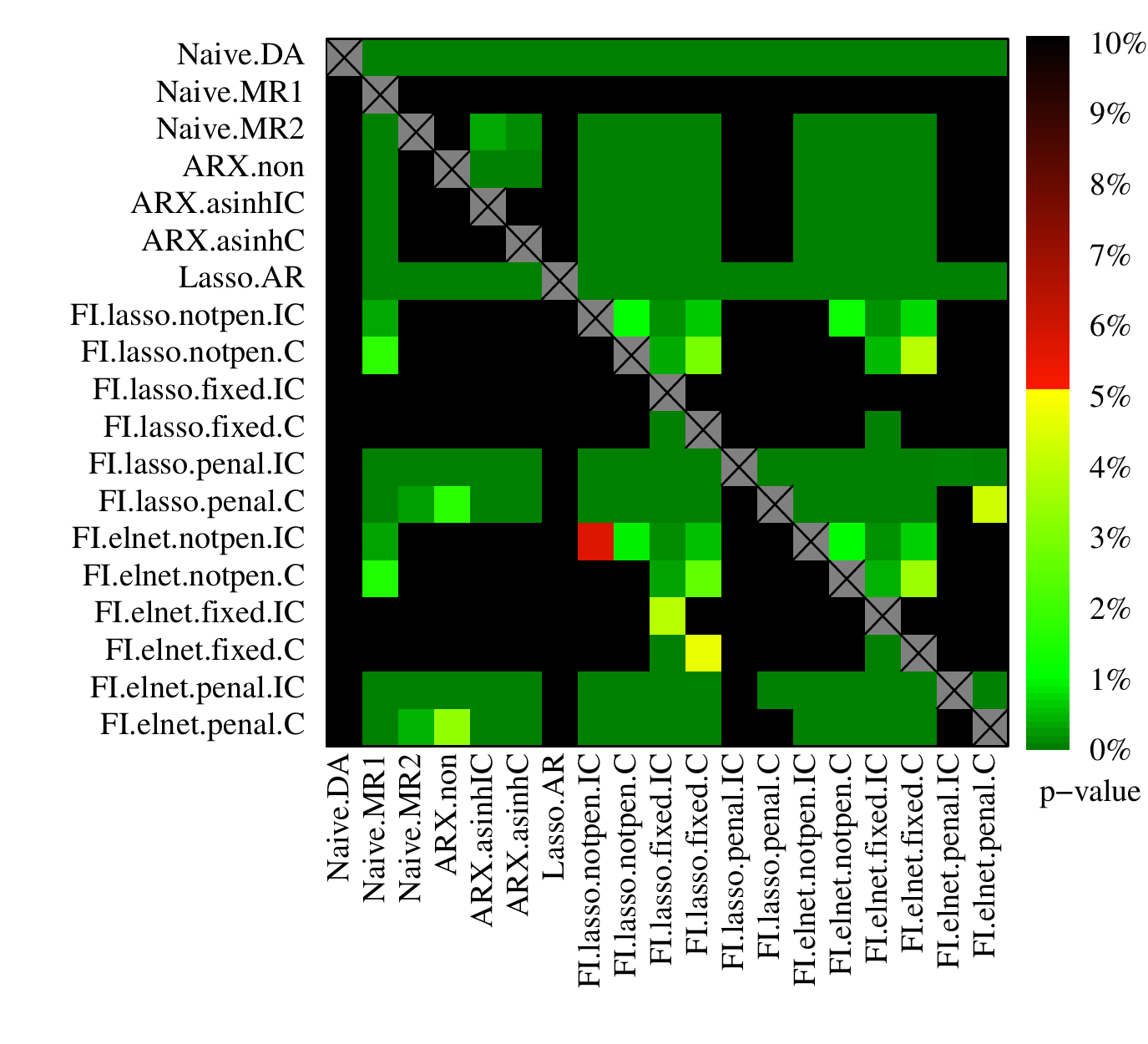}
			\label{fig:DM_RMSE_h}
		}	
		\subfloat[]{
			\includegraphics[width = 0.475\linewidth]{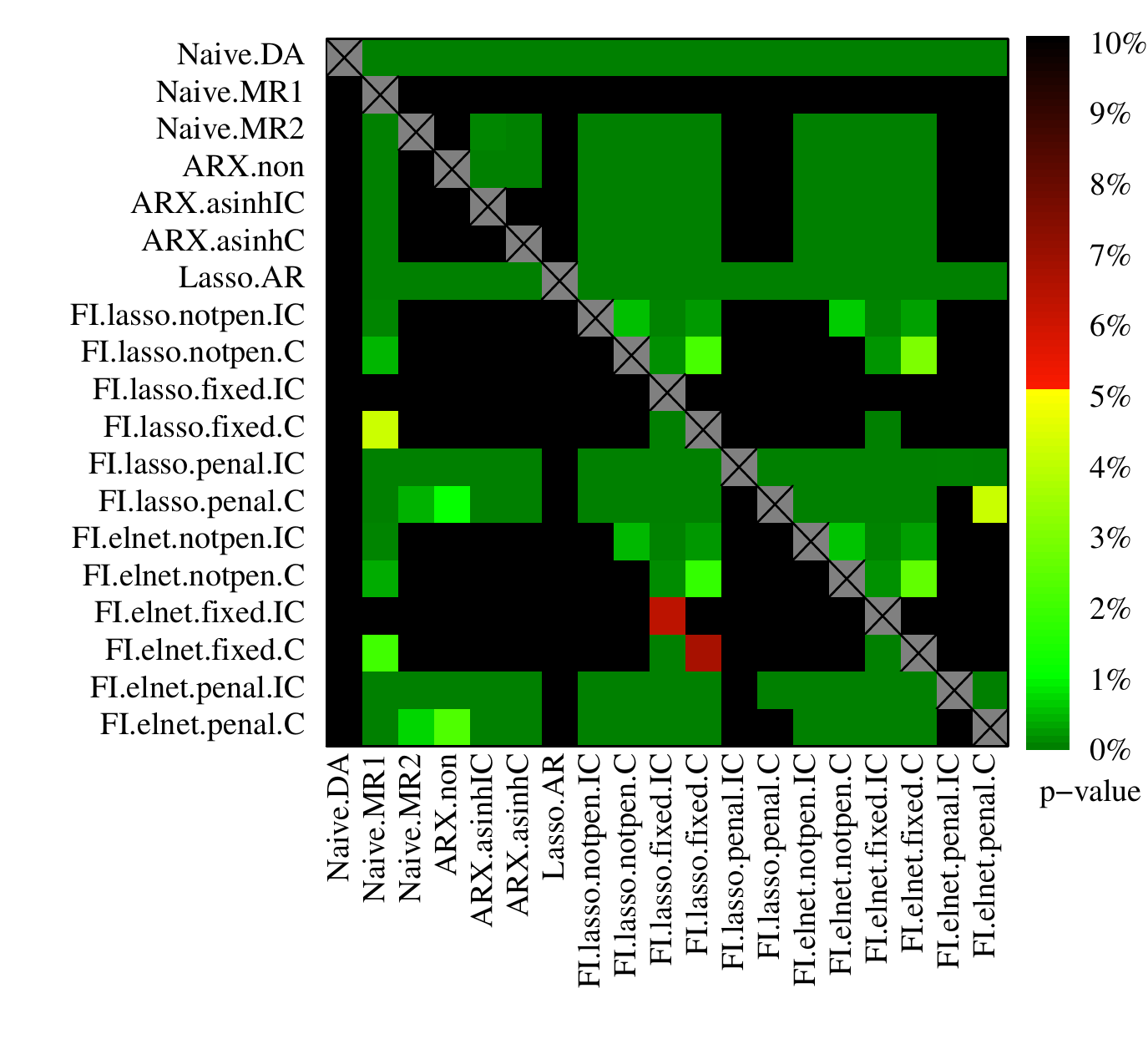}
			\label{fig:DM_MAE_h}
		}
		\caption{Results of the Diebold-Mariano test for the hourly products. (a) presents the $p$-values for the $||\cdot||_2$ norm, (b) the values for the $||\cdot||_1$ norm. The figures use a heat map to indicate the range of the $p$-values. The closer they are to zero ($\to$ dark green), the more significant the difference is between forecasts of X-axis model (better) and forecasts of the Y-axis model (worse).}
		\label{fig:DMtest_h}
	\end{figure*}

	In Table \ref{tab:maermse}, we present the MAE and RMSE values for all considered models. Next to the error values, we report the corresponding standard deviations estimated using 1000 bootstrap replications. We calculate the errors separately for the hourly and quarter-hourly products. Figures \ref{fig:DMtest_h} and \ref{fig:DMtest_qh} present the $p$-values of the DM test for the hourly and quarter-hourly products, respectively. The results for the hourly products are quite surprising. The \textbf{FI.lasso.fixed.IC} model gives the lowest error for both MAE and RMSE, but these values are not significantly different from the errors of the naive model (\ref{eq:naiveMR}). Based on the DM test results shown in Figure \ref{fig:DMtest_h}, we observe that none of the considered models is significantly better than the aforementioned naive. Let us recall that the \textbf{Naive.MR1} models the ID$_3$ price with a weighted-average price of the transactions that take place not earlier than 15 minutes before the forecasting time. This is an indication of a weak-form efficiency of the market, as based on the market data there is no arbitrage possible for risk-neutral traders, similarly to established financial markets. The same behaviour was captured by \cite{ziel2015forecasting} in the German-Austrian day-ahead markets. To be specific, they found the naive EXAA model to be the best for the results of the EPEX Day-Ahead Auction. The naive EXAA model uses as the predictor simply the price of the EXAA, which publishes the results 1 hour and 40 minutes before the submission in the EPEX Day-Ahead Auction.

	\begin{figure*}[t!]
		\centering
		\subfloat[]{
			\includegraphics[width = 0.475\linewidth]{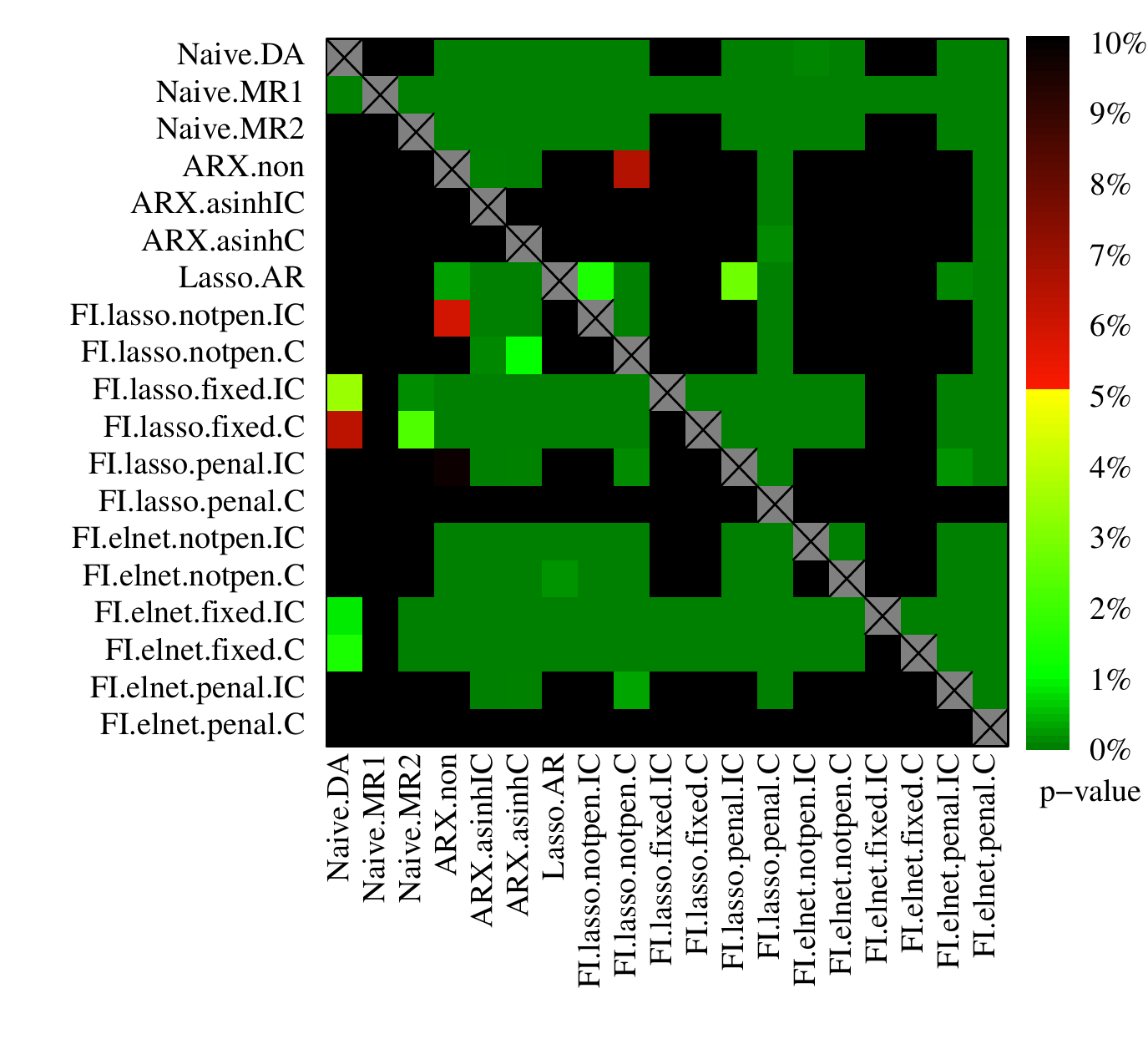}
			\label{fig:DM_RMSE_qh}
		}%
		~ 
		\subfloat[]{
			\includegraphics[width = 0.475\linewidth]{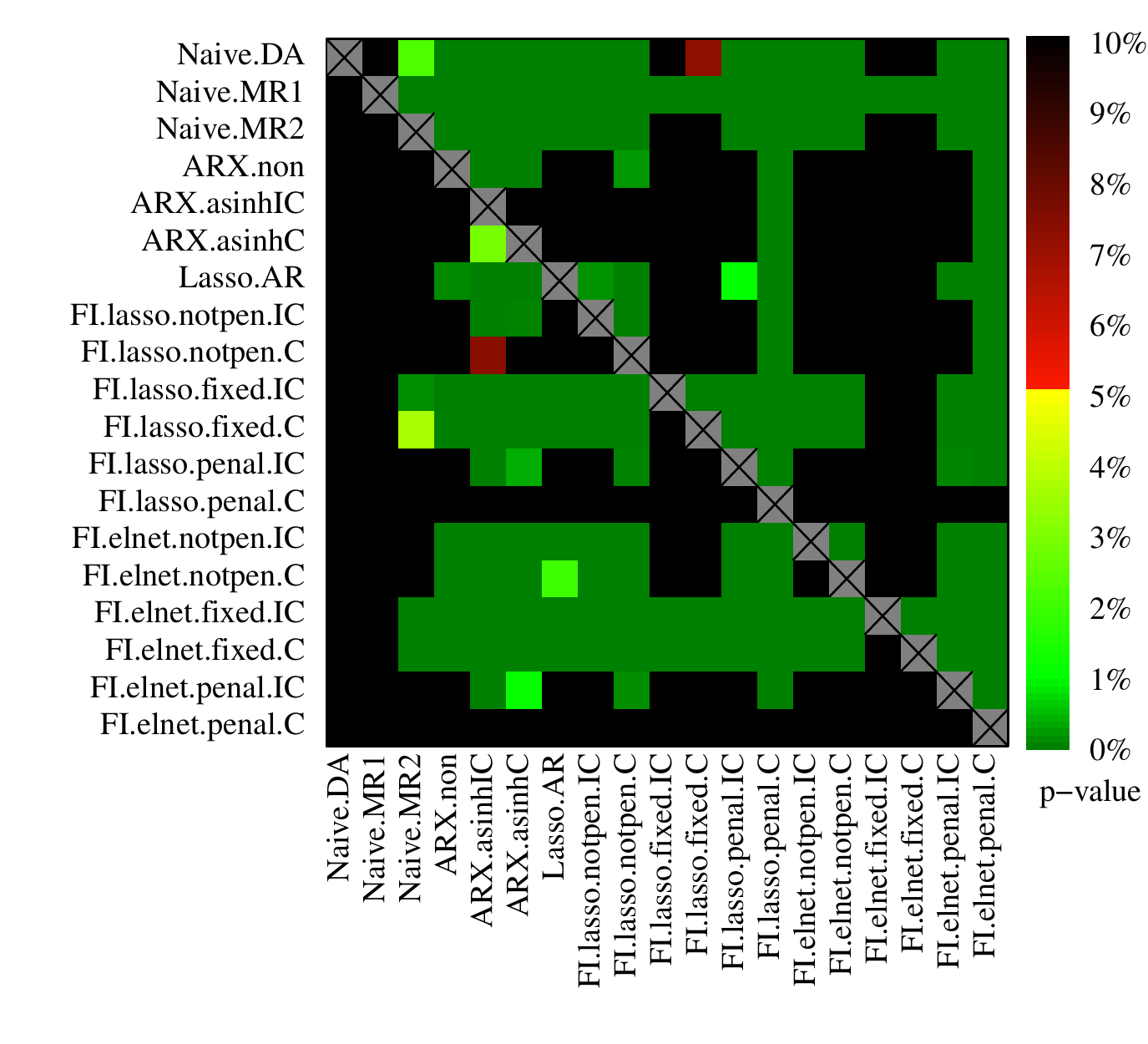}
			\label{fig:DM_MAE_qh}
		}
		\caption{Results of the Diebold-Mariano test for the quarter-hourly products. (a) presents the $p$-values for the $||\cdot||_2$ norm, (b) the values for the $||\cdot||_1$ norm. The figures use a heat map to indicate the range of the $p$-values. The closer they are to zero ($\to$ dark green), the more significant the difference is between forecasts of X-axis model (better) and forecasts of the Y-axis model (worse).}
		\label{fig:DMtest_qh}
	\end{figure*}

	The other full information models, which do not penalize for using the most recent value or have the corresponding coefficient fixed to~1, perform not much worse. In terms of MAE and RMSE, all of them are not significantly different from the \textbf{FI.lasso.fixed.IC} and of course than the \textbf{Naive.MR1}. Based on the DM test, the forecasts of the models with the most recent value coefficient set to 1 are not significantly worse than the forecasts of the \textbf{Naive.MR1}. On the other hand, the \textbf{FI.lasso.fixed.IC} appears to be significantly the best among all the considered models, excluding the naive most recent value. The expert models perform very well too, while the \textbf{Naive.DA} and the \textbf{Lasso.AR} perform the worst. Let us note that these models are the only ones that do not contain the most recent value, which turns out to explain the final ID$_3$ price very well. Let us emphasize that the correct way of the backward transformation in most cases returns lower errors and significantly better forecasts than the incorrect way.
	
	Figure \ref{fig:RMSEMAEh} shows models' performance for the hourly products separately over the day. We see there clearly that the models, that do not consist of the \textbf{Naive.MR1}, perform the worst in every hour. The ones that use the most recent value, but penalize its usage, give better results, but still, they are worse than the models that favour the most recent value or the  \textbf{Naive.MR1} itself. Among these, it is not easy to distinguish the best one based on Figure \ref{fig:RMSEMAEh}. Moreover, all models perform better during night hours than day hours, especially the peak ones.
	
	The situation is slightly different for the quarter-hourly products. Based on Table \ref{tab:maermse} the \textbf{Naive.MR1} does not perform that well. Also, the DM test results, presented in Figure \ref{fig:DMtest_qh}, indicate its poor performance. This suggests that the quarter-hourly intraday market is not weak-form efficient and still an autoregressive or deterministic structure can be found there. In terms of the MAE, we receive the lowest error for the full information model estimated using the elastic net with a standard penalty and correctly back-transformed, i.e. \textbf{FI.elnet.penal.C}. Very similar and not significantly different value is received for the \textbf{FI.lasso.penal.C} model. In terms of the RMSE, the best performing model turns out to be the asinh-transformed and correctly back-transformed expert model. Although, the other expert models and the aforementioned full information models have their RMSEs in the 2-sigma range of the \textbf{ARX.asinhC}. Based on Figure \ref{fig:DMtest_qh}, the models whose forecasts significantly outperform other models' forecasts are the \textbf{FI.lasso.penal.C} and the \textbf{FI.elnet.penal.C}. According to the results of the DM test, the forecasts of these models are not significantly different. The other models that perform well are the \textbf{ARX} models, while the worst performing are the naives and the full information models that have the most recent value coefficient fixed to~1.

		\begin{figure*}[t]
		\centering
		\subfloat[]{
			\includegraphics[width = 0.475\linewidth]{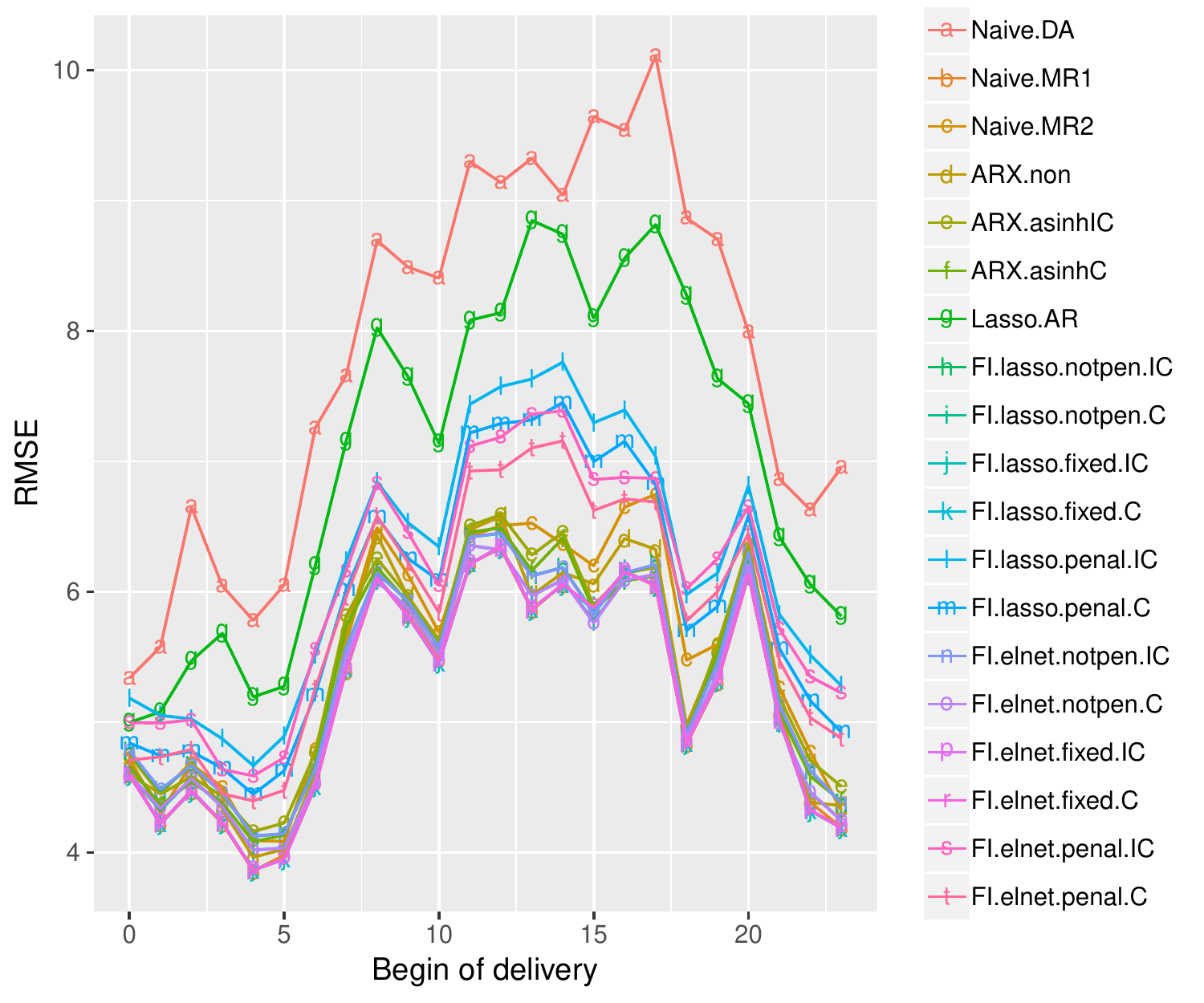}
			\label{fig:RMSEh}
		}%
		~ 
		\subfloat[]{
			\includegraphics[width = 0.475\linewidth]{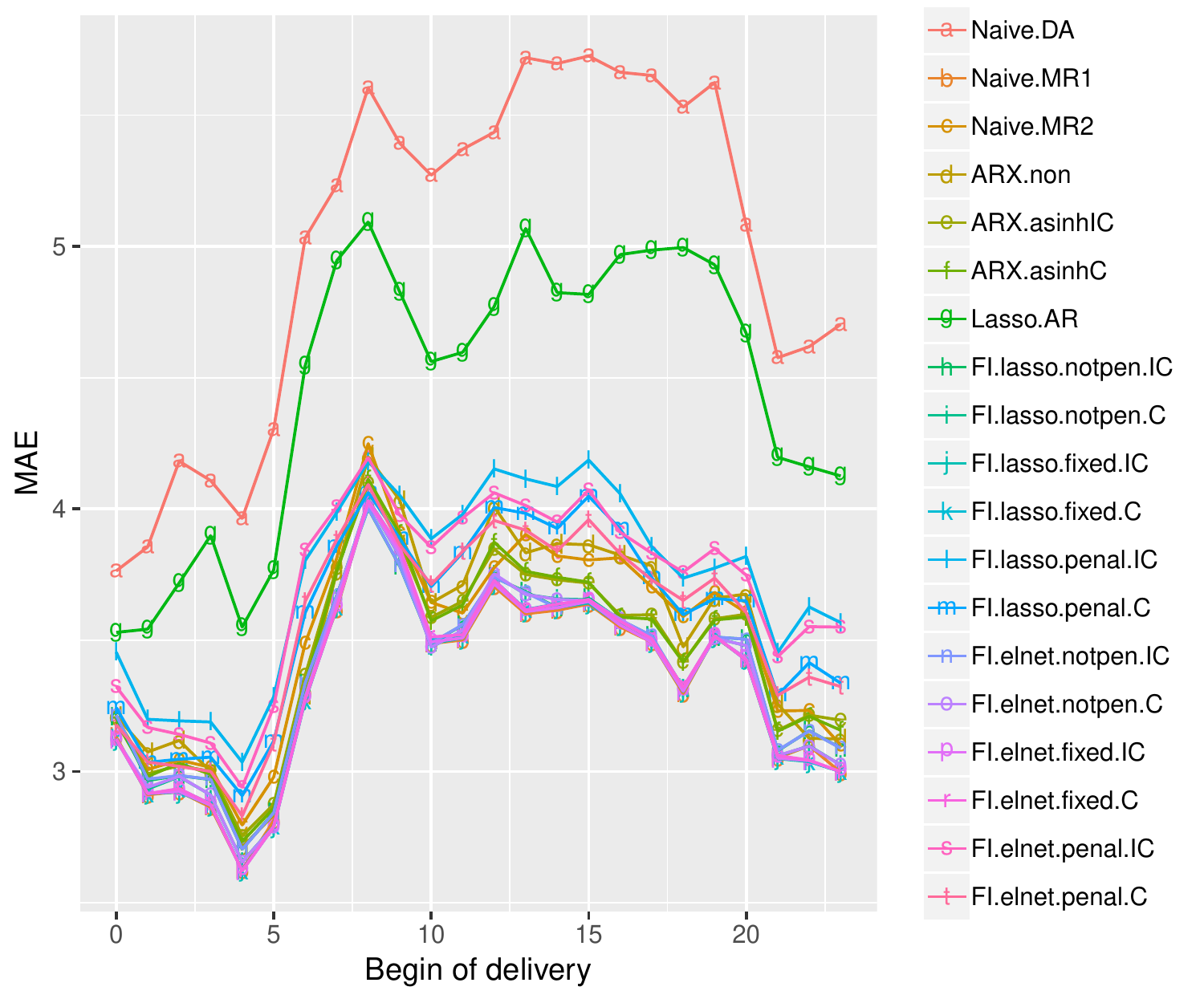}
			\label{fig:MAEh}
		}
		\caption{Performance measures for hourly products}
		\label{fig:RMSEMAEh}
	\end{figure*}
	
	\begin{figure*}[t!]
		\centering
		\subfloat[]{
			\includegraphics[width = 0.475\linewidth]{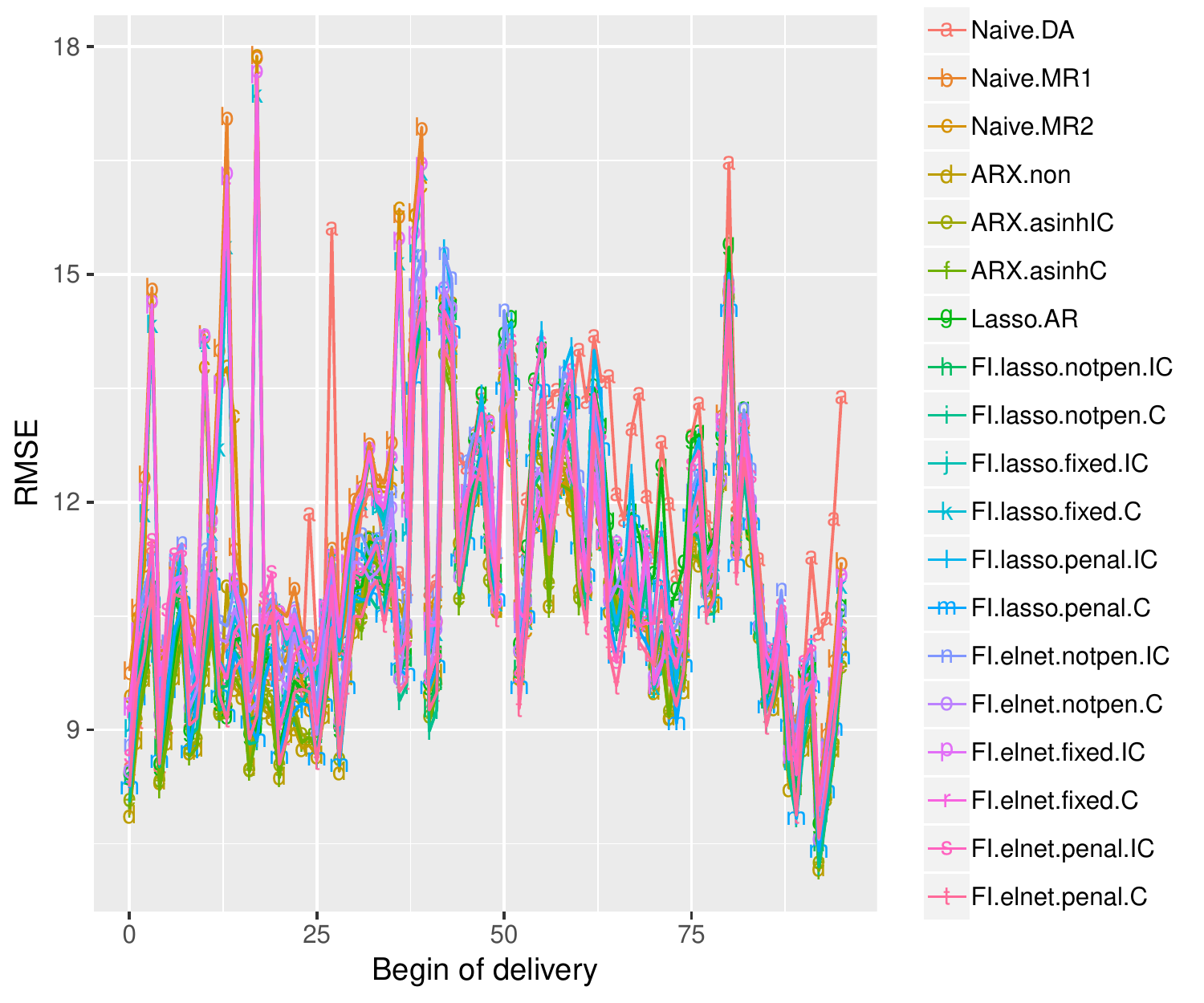}
			\label{fig:RMSEqh}
		}%
		~ 
		\subfloat[]{
			\includegraphics[width = 0.475\linewidth]{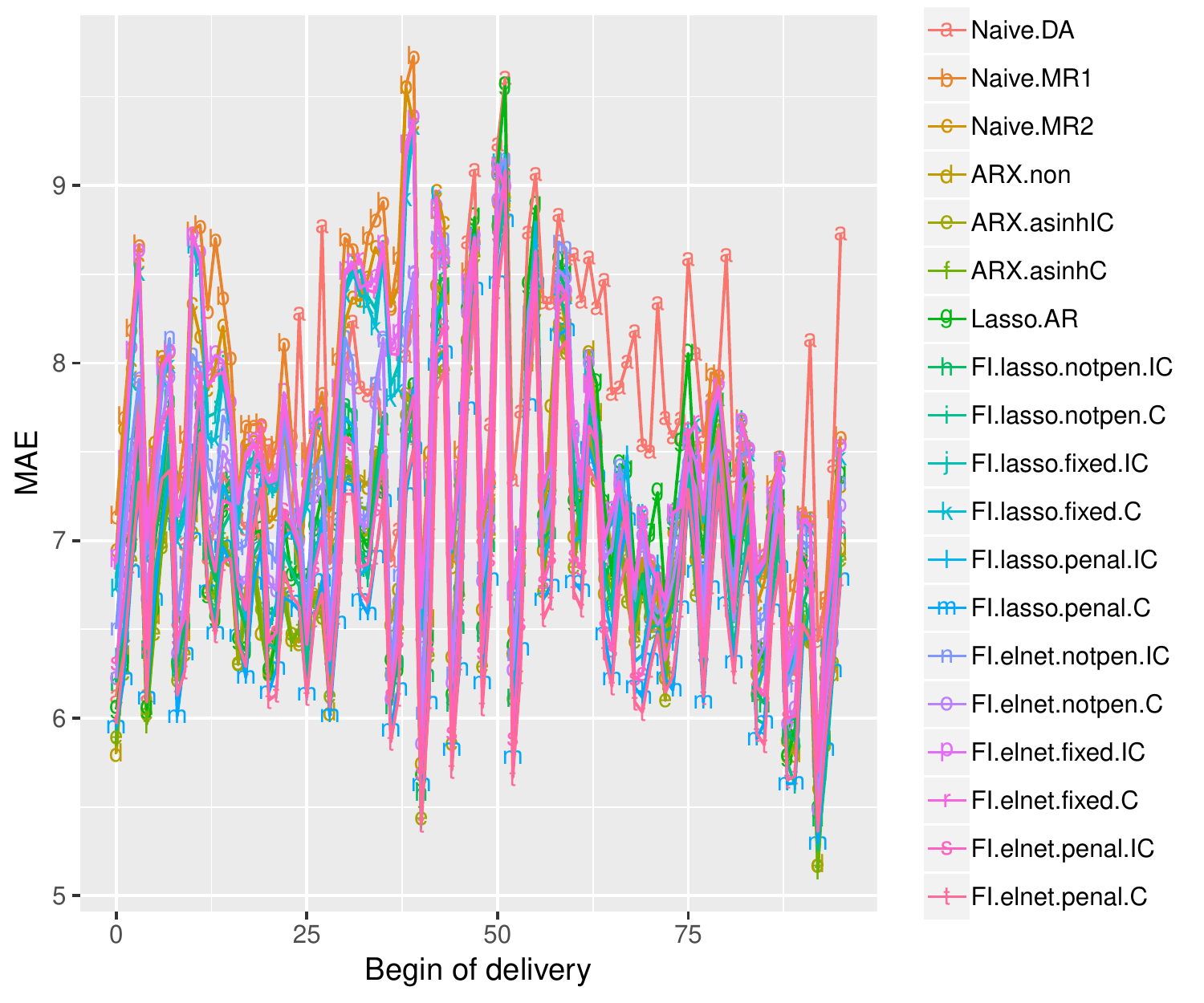}
			\label{fig:MAEqh}
		}
		\caption{Performance measures for quarter-hourly products}
		\label{fig:RMSEMAEqh}
	\end{figure*}

	Figure \ref{fig:RMSEMAEqh} shows the models' performance for the quarter-hourly products over the day. Let us note that here the results are not that explicit, as in Figure \ref{fig:RMSEMAEh}. The error values exhibit the jigsaw pattern that is present for the quarter-hourly prices in Figure \ref{fig:ID3weeklyqh}. It means that the models do not handle well this pattern. The considered models perform much worse for the quarter-hourly than for hourly products. The reason for this may be the higher volatility of the quarter-hourly market, especially the higher rate of the outliers occurrence. Another reason may be the lower liquidity of this market 3 hours and 15 minutes before the delivery. Figure \ref{fig:numberoftrans} shows that by this time, a very little number of transactions takes place in this market. This may explain also the poor performance of the \textbf{Naive.MR1}.

	\subsection{Variable selection}

	\begin{table}[!hb]
		\centering
		\begingroup\footnotesize
		\begin{tabular}{rllll}
			\hline
			& Importance: 1 & Importance: 2 & Importance: 3  & Importance: 4\\ 
			\hline
			00:00 & \cellcolor[rgb]{0.5,0.9,0.5} {$\prescript{H}{3.25}{\text{ID}}_{0.25}^{d, \text{00:00}}$ (99.86)} & \cellcolor[rgb]{1,0.505,0.5} {$\text{BV}^{d-1,\text{17:15}}$ (0.04)} & \cellcolor[rgb]{1,0.504,0.5} {$\prescript{H}{4.75}{\text{ID}}_{0.25}^{d, \text{00:00}}$ (0.03)} & \cellcolor[rgb]{1,0.503,0.5} {$\prescript{QH}{9.25}{\text{ID}}_{0.25}^{d-1, \text{20:45}}$ (0.03)} \\ 
			01:00 & \cellcolor[rgb]{0.5,0.9,0.5} {$\prescript{H}{3.25}{\text{ID}}_{0.25}^{d, \text{01:00}}$ (99.7)} & \cellcolor[rgb]{1,0.523,0.5} {$\prescript{QH}{0.5}{\text{ID}}_{0.25}^{d-1, \text{22:15}}$ (0.18)} & \cellcolor[rgb]{1,0.507,0.5} {$\prescript{QH}{}{\text{ID}}_{3}^{d-9,\text{07:30}}$ (0.06)} & \cellcolor[rgb]{1,0.505,0.5} {$\text{IA}^{d-2,\text{08:00}}$ (0.04)} \\ 
			02:00 & \cellcolor[rgb]{0.5,0.9,0.5} {$\prescript{H}{3.25}{\text{ID}}_{0.25}^{d, \text{02:00}}$ (99.09)} & \cellcolor[rgb]{1,0.569,0.5} {$\text{DA}^{d-13,\text{06:00}}$ (0.54)} & \cellcolor[rgb]{1,0.532,0.5} {$\prescript{QH}{14.5}{\text{ID}}_{0.25}^{d, \text{06:45}}$ (0.25)} & \cellcolor[rgb]{1,0.506,0.5} {$\text{BV}^{d-1,\text{22:15}}$ (0.04)} \\ 
			03:00 & \cellcolor[rgb]{0.5,0.9,0.5} {$\prescript{H}{3.25}{\text{ID}}_{0.25}^{d, \text{03:00}}$ (99.92)} & \cellcolor[rgb]{1,0.507,0.5} {$\text{BV}^{d-8,\text{06:00}}$ (0.06)} & \cellcolor[rgb]{1,0.502,0.5} {$\prescript{QH}{14.5}{\text{ID}}_{0.25}^{d-1, \text{07:00}}$ (0.02)} & \cellcolor[rgb]{1,0.501,0.5} {$\prescript{QH}{5.25}{\text{ID}}_{0.25}^{d-1, \text{04:45}}$ (0)} \\ 
			04:00 & \cellcolor[rgb]{0.5,0.9,0.5} {$\prescript{H}{3.25}{\text{ID}}_{0.25}^{d, \text{04:00}}$ (99.92)} & \cellcolor[rgb]{1,0.505,0.5} {$\text{BV}^{d-3,\text{15:45}}$ (0.04)} & \cellcolor[rgb]{1,0.502,0.5} {$\prescript{QH}{0}{\text{ID}}_{0.25}^{d-1, \text{13:45}}$ (0.02)} & \cellcolor[rgb]{1,0.502,0.5} {$\text{IA}^{d-14,\text{21:45}}$ (0.02)} \\ 
			05:00 & \cellcolor[rgb]{0.5,0.9,0.5} {$\prescript{H}{3.25}{\text{ID}}_{0.25}^{d, \text{05:00}}$ (99.9)} & \cellcolor[rgb]{1,0.508,0.5} {$\text{BV}^{d-2,\text{02:30}}$ (0.06)} & \cellcolor[rgb]{1,0.502,0.5} {$\text{IA}^{d,\text{00:45}}$ (0.02)} & \cellcolor[rgb]{1,0.501,0.5} {$\prescript{QH}{8.25}{\text{ID}}_{0.25}^{d, \text{02:45}}$ (0.01)} \\ 
			06:00 & \cellcolor[rgb]{0.5,0.9,0.5} {$\prescript{H}{3.25}{\text{ID}}_{0.25}^{d, \text{06:00}}$ (99.47)} & \cellcolor[rgb]{1,0.522,0.5} {$\prescript{QH}{12.75}{\text{ID}}_{0.25}^{d-1, \text{23:15}}$ (0.17)} & \cellcolor[rgb]{1,0.519,0.5} {$\prescript{QH}{6}{\text{ID}}_{0.25}^{d-1, \text{23:45}}$ (0.15)} & \cellcolor[rgb]{1,0.515,0.5} {$\prescript{QH}{6.25}{\text{ID}}_{0.25}^{d-1, \text{23:45}}$ (0.11)} \\ 
			07:00 & \cellcolor[rgb]{0.5,0.9,0.5} {$\prescript{H}{3.25}{\text{ID}}_{0.25}^{d, \text{07:00}}$ (99.7)} & \cellcolor[rgb]{1,0.516,0.5} {$\prescript{QH}{5.75}{\text{ID}}_{0.25}^{d-1, \text{19:30}}$ (0.12)} & \cellcolor[rgb]{1,0.509,0.5} {$\prescript{QH}{8.5}{\text{ID}}_{0.25}^{d, \text{00:45}}$ (0.07)} & \cellcolor[rgb]{1,0.503,0.5} {$\prescript{QH}{12.5}{\text{ID}}_{0.25}^{d-1, \text{22:00}}$ (0.03)} \\ 
			08:00 & \cellcolor[rgb]{0.5,0.9,0.5} {$\prescript{H}{3.25}{\text{ID}}_{0.25}^{d, \text{08:00}}$ (99.86)} & \cellcolor[rgb]{1,0.503,0.5} {$\prescript{QH}{9.25}{\text{ID}}_{0.25}^{d, \text{02:45}}$ (0.03)} & \cellcolor[rgb]{1,0.503,0.5} {$\prescript{QH}{5}{\text{ID}}_{0.25}^{d-1, \text{15:45}}$ (0.02)} & \cellcolor[rgb]{1,0.503,0.5} {$\prescript{QH}{14}{\text{ID}}_{0.25}^{d-1, \text{21:30}}$ (0.02)} \\ 
			09:00 & \cellcolor[rgb]{0.5,0.9,0.5} {$\prescript{H}{3.25}{\text{ID}}_{0.25}^{d, \text{09:00}}$ (99.46)} & \cellcolor[rgb]{1,0.561,0.5} {$\prescript{QH}{}{\text{ID}}_{3}^{d-6,\text{21:30}}$ (0.48)} & \cellcolor[rgb]{1,0.505,0.5} {$\prescript{QH}{16.5}{\text{ID}}_{0.25}^{d, \text{09:00}}$ (0.04)} & \cellcolor[rgb]{1,0.501,0.5} {$\prescript{H}{26}{\text{ID}}_{0.25}^{d-1, \text{23:00}}$ (0)} \\ 
			10:00 & \cellcolor[rgb]{0.5,0.9,0.5} {$\prescript{H}{3.25}{\text{ID}}_{0.25}^{d, \text{10:00}}$ (99.08)} & \cellcolor[rgb]{1,0.53,0.5} {$\prescript{QH}{4}{\text{ID}}_{0.25}^{d, \text{09:45}}$ (0.24)} & \cellcolor[rgb]{1,0.514,0.5} {$\prescript{QH}{4.25}{\text{ID}}_{0.25}^{d, \text{10:00}}$ (0.11)} & \cellcolor[rgb]{1,0.514,0.5} {$\prescript{QH}{2.75}{\text{ID}}_{0.25}^{d, \text{02:45}}$ (0.11)} \\ 
			11:00 & \cellcolor[rgb]{0.5,0.9,0.5} {$\prescript{H}{3.25}{\text{ID}}_{0.25}^{d, \text{11:00}}$ (99.78)} & \cellcolor[rgb]{1,0.528,0.5} {$\prescript{QH}{3}{\text{ID}}_{0.25}^{d, \text{00:00}}$ (0.22)} & \cellcolor[rgb]{1,0.5,0.5} { } & \cellcolor[rgb]{1,0.5,0.5} { } \\ 
			12:00 & \cellcolor[rgb]{0.5,0.9,0.5} {$\prescript{H}{3.25}{\text{ID}}_{0.25}^{d, \text{12:00}}$ (99.73)} & \cellcolor[rgb]{1,0.534,0.5} {$\text{BV}^{d-7,\text{20:15}}$ (0.27)} & \cellcolor[rgb]{1,0.5,0.5} {$\prescript{QH}{}{\text{ID}}_{3}^{d-9,\text{00:15}}$ (0)} & \cellcolor[rgb]{1,0.5,0.5} {$\text{BV}^{d-8,\text{19:15}}$ (0)} \\ 
			13:00 & \cellcolor[rgb]{0.5,0.9,0.5} {$\prescript{H}{3.25}{\text{ID}}_{0.25}^{d, \text{13:00}}$ (99.82)} & \cellcolor[rgb]{1,0.516,0.5} {$\prescript{QH}{}{\text{ID}}_{3}^{d-9,\text{07:15}}$ (0.12)} & \cellcolor[rgb]{1,0.503,0.5} {$\prescript{QH}{}{\text{ID}}_{3}^{d-9,\text{09:30}}$ (0.03)} & \cellcolor[rgb]{1,0.502,0.5} {$\prescript{QH}{}{\text{ID}}_{3}^{d-12,\text{21:15}}$ (0.02)} \\ 
			14:00 & \cellcolor[rgb]{0.5,0.9,0.5} {$\prescript{H}{3.25}{\text{ID}}_{0.25}^{d, \text{14:00}}$ (99.92)} & \cellcolor[rgb]{1,0.506,0.5} {$\prescript{QH}{3}{\text{ID}}_{0.25}^{d-1, \text{11:00}}$ (0.05)} & \cellcolor[rgb]{1,0.502,0.5} {$\text{BV}^{d-1,\text{21:00}}$ (0.01)} & \cellcolor[rgb]{1,0.502,0.5} {$\text{BV}^{d-1,\text{16:30}}$ (0.01)} \\ 
			15:00 & \cellcolor[rgb]{0.5,0.9,0.5} {$\prescript{H}{3.25}{\text{ID}}_{0.25}^{d, \text{15:00}}$ (97.57)} & \cellcolor[rgb]{1,0.64,0.5} {$\prescript{QH}{}{\text{ID}}_{3}^{d-7,\text{06:45}}$ (1.09)} & \cellcolor[rgb]{1,0.553,0.5} {$\text{IA}^{d-9,\text{06:00}}$ (0.42)} & \cellcolor[rgb]{1,0.532,0.5} {$\text{BV}^{d-9,\text{23:30}}$ (0.25)} \\ 
			16:00 & \cellcolor[rgb]{0.5,0.9,0.5} {$\prescript{H}{3.25}{\text{ID}}_{0.25}^{d, \text{16:00}}$ (99.03)} & \cellcolor[rgb]{1,0.53,0.5} {$\prescript{QH}{}{\text{ID}}_{3}^{d-4,\text{13:45}}$ (0.23)} & \cellcolor[rgb]{1,0.521,0.5} {$\text{IA}^{d-7,\text{20:00}}$ (0.16)} & \cellcolor[rgb]{1,0.513,0.5} {$\text{IA}^{d-7,\text{17:00}}$ (0.1)} \\ 
			17:00 & \cellcolor[rgb]{0.5,0.9,0.5} {$\prescript{H}{3.25}{\text{ID}}_{0.25}^{d, \text{17:00}}$ (98.53)} & \cellcolor[rgb]{1,0.642,0.5} {$\prescript{QH}{5.25}{\text{ID}}_{0.25}^{d, \text{06:15}}$ (1.11)} & \cellcolor[rgb]{1,0.511,0.5} {$\prescript{QH}{4.75}{\text{ID}}_{0.25}^{d, \text{05:30}}$ (0.09)} & \cellcolor[rgb]{1,0.506,0.5} {$\prescript{QH}{}{\text{ID}}_{3}^{d-7,\text{21:30}}$ (0.05)} \\ 
			18:00 & \cellcolor[rgb]{0.5,0.9,0.5} {$\prescript{H}{3.25}{\text{ID}}_{0.25}^{d, \text{18:00}}$ (99.67)} & \cellcolor[rgb]{1,0.514,0.5} {$\text{IA}^{d-10,\text{09:00}}$ (0.11)} & \cellcolor[rgb]{1,0.509,0.5} {$\prescript{QH}{0.25}{\text{ID}}_{0.25}^{d-1, \text{15:30}}$ (0.07)} & \cellcolor[rgb]{1,0.507,0.5} {$\prescript{QH}{16}{\text{ID}}_{0.25}^{d, \text{08:30}}$ (0.06)} \\ 
			19:00 & \cellcolor[rgb]{0.5,0.9,0.5} {$\prescript{H}{3.25}{\text{ID}}_{0.25}^{d, \text{19:00}}$ (98.68)} & \cellcolor[rgb]{1,0.557,0.5} {$\prescript{QH}{3.5}{\text{ID}}_{0.25}^{d-1, \text{16:45}}$ (0.44)} & \cellcolor[rgb]{1,0.548,0.5} {$\prescript{QH}{2.75}{\text{ID}}_{0.25}^{d-1, \text{16:45}}$ (0.38)} & \cellcolor[rgb]{1,0.519,0.5} {$\prescript{QH}{1.75}{\text{ID}}_{0.25}^{d-1, \text{16:45}}$ (0.15)} \\ 
			20:00 & \cellcolor[rgb]{0.5,0.9,0.5} {$\prescript{H}{3.25}{\text{ID}}_{0.25}^{d, \text{20:00}}$ (99.92)} & \cellcolor[rgb]{1,0.505,0.5} {$\prescript{QH}{7.25}{\text{ID}}_{0.25}^{d, \text{07:30}}$ (0.04)} & \cellcolor[rgb]{1,0.502,0.5} {$\prescript{QH}{2.5}{\text{ID}}_{0.25}^{d-1, \text{16:45}}$ (0.01)} & \cellcolor[rgb]{1,0.501,0.5} {$\text{BV}^{d-14,\text{19:30}}$ (0.01)} \\ 
			21:00 & \cellcolor[rgb]{0.5,0.9,0.5} {$\prescript{H}{3.25}{\text{ID}}_{0.25}^{d, \text{21:00}}$ (98.55)} & \cellcolor[rgb]{1,0.564,0.5} {$\text{IA}^{d1,\text{07:00}}$ (0.5)} & \cellcolor[rgb]{1,0.546,0.5} {$\text{IA}^{d-10,\text{18:00}}$ (0.36)} & \cellcolor[rgb]{1,0.543,0.5} {$\prescript{QH}{2.25}{\text{ID}}_{0.25}^{d, \text{19:45}}$ (0.34)} \\ 
			22:00 & \cellcolor[rgb]{0.5,0.9,0.5} {$\prescript{H}{3.25}{\text{ID}}_{0.25}^{d, \text{22:00}}$ (99.02)} & \cellcolor[rgb]{1,0.526,0.5} {$\prescript{QH}{25.5}{\text{ID}}_{0.25}^{d1, \text{20:00}}$ (0.21)} & \cellcolor[rgb]{1,0.523,0.5} {$\text{BV}^{d-6,\text{00:00}}$ (0.18)} & \cellcolor[rgb]{1,0.519,0.5} {$\text{IA}^{d1,\text{01:45}}$ (0.15)} \\ 
			23:00 & \cellcolor[rgb]{0.5,0.9,0.5} {$\prescript{H}{3.25}{\text{ID}}_{0.25}^{d, \text{23:00}}$ (99.96)} & \cellcolor[rgb]{1,0.504,0.5} {$\prescript{QH}{3.25}{\text{ID}}_{0.25}^{d, \text{03:15}}$ (0.03)} & \cellcolor[rgb]{1,0.5,0.5} {$\text{BV}^{d-5,\text{22:00}}$ (0)} & \cellcolor[rgb]{1,0.5,0.5} {$\text{BV}^{d-10,\text{03:00}}$ (0)} \\ 
			\hline
		\end{tabular}
		\endgroup
		\caption{Most relevant coefficients in the model \textbf{FI.lasso.fixed.IC} for hourly products} 
		\label{tab:beta_h}
	\end{table}
	
	We perform a variable selection analysis for the best full information model for the hourly products, i.e. \textbf{FI.lasso.fixed.IC} and for the best performing model for the quarter-hourly products, i.e. \textbf{FI.elnet.penal.C}. Table \ref{tab:beta_h} shows the four most relevant coefficients in the model \textbf{FI.lasso.fixed.IC} for each hourly product. Let us note that for almost every hour the most recent value has the importance of more than $99\%$. This explains why this model is not significantly different from the \textbf{Naive.MR1}. This may be an important hint that there is no more information that we can get from this data, despite the most recent value. There is no clear pattern in the selection of the second and the third most relevant coefficient, especially that the corresponding values are very small.

	The tables regarding the importance of parameters for the model \textbf{FI.elnet.penal.C} for the quarter-hourly products can be found in the Appendix. Table \ref{tab:beta_qh_1} shows that for the night and early-morning products the corresponding Intraday Auction value is a very relevant variable. This pattern appears, but is not that strong in Tables \ref{tab:beta_qh_2} and \ref{tab:beta_qh_3}. The most recent value of the corresponding product and the most recent value of the closest hourly product seem relevant for many products too. This suggests that the most recent value is an informative variable, even though it is not that good alone, as it is for the hourly products. Let us note that the information from the balancing market and from the hourly day-ahead market is non-existent in the 4 most important variables. Most of the regressors that appear there are the $_x\text{ID}_y$ prices from the Intraday Continuous market. The pattern of choice of the $_x\text{ID}_y$ prices is not clear, except the fact that most of the important variables have their delivery before the corresponding product.
	
	\section{Summary and Conclusion}
	We conducted an electricity price forecasting study in the German Intraday Continuous market. We utilized a new model that makes use of the market's continuity and estimated it using well-known techniques. We compared it with seven benchmark models to measure its performance. We performed the analysis separately for the hourly and quarter-hourly products. The results for the hourly products suggest that there is no more information that we can get from the transactions data, despite the most recent price, which means that here we deal with a weak-form efficient market. Therefore, none of the considered models performed better than the naive most recent value. This is very similar to the relation between the German-Austrian EXAA and EPEX day-ahead markets that is exhibited by \cite{ziel2015forecasting}.
	
	On the other hand, the results for the quarter-hourly products have shown that there is some space for improvement. For this market, the most recent value did not give satisfying results. The reason for this may be a lower number of transactions by the time of forecasting when comparing with the market of the hourly products. 
	In the case of the quarter-hourly products, the full information model estimated using the elastic net with a standard penalty and correctly back-transformed performed the best. The variable selection analysis of this model has shown that the most relevant regressors are: the corresponding Intraday Auction price, the most recent value of the corresponding product, and the most recent value of the closest hourly product.
	
	An important outcome of the following paper is the analysis of the asinh's backward transformation. We have shown that the mathematically correct approach to this problem gives significantly better forecasts than the common, incorrect one. This is clearly depicted in Figures \ref{fig:DMtest_h} and \ref{fig:DMtest_qh}. We see that comparing the same models back-transformed correctly and incorrectly results in most cases in significantly better forecasts when using the correct backward transformation. Having this and the approach's simplicity on our minds, we strongly encourage to use it.
	
	Future research can go in different directions. Likely the proposed model can be improved using other estimation methods. The next natural step forward is the inclusion of fundamental variables in the model, like e.g. weather forecasts or outages of power plants. Another possible direction is probabilistic forecasting, including a detailed volatility analysis. It is certain that due to the growth of the intraday market and the scarcity of the literature on this subject, still a lot of research needs to be conducted.
	
	\section*{Acknowledgments}
	
	This research article was partially supported by the German Research Foundation (DFG, Germany) and the National Science Center (NCN, Poland) through BEETHOVEN grant no. 2016/23/G/HS4/01005.

	\section*{Appendix}
	\subsection*{The $_x\text{ID}_y$ weighted-average additivity property}
	\begin{proposition}
		Let us consider a disjoint split of the $\mathbb{T}_{x,y}^{d,s}$ period
		\begin{equation*}
		\mathbb{T}_{x,y}^{d,s}  =  \dot{\bigcup_{j}} \mathbb{T}_{x_j,y_j}^{d,s} 
		= \dot{\bigcup_{j}} \left[b(d,s) - x_j - y_j, b(d,s) - x_j \right),			
		\end{equation*}
		where $j \in \{0, 1, \dots, J\}$, $x_0 + y_0 = x + y$ and $x_J = x$. Then
		\begin{equation}
		{}_{x}\text{ID}^{d,s}_y = \frac{  \sum_{j} {}_{x_j}\text{ID}^{d,s}_{y_j} \mathbb{V}_{x_j,y_j}^{d,s}}{\sum_{j} \mathbb{V}_{x_j,y_j}^{d,s}},
		\end{equation}
		where $\mathbb{V}_{x,y}^{d,s} = \sum_{k \in \mathbb{T}_{x,y}^{d,s} \cap \mathcal{T}^{d,s}} V_k^{d,s}$.
	\end{proposition}
	
	\begin{proof}
		To prove the property, we apply the above disjoint of the time frame $\mathbb{T}_{x,y}^{d,s}$ to the definition (\ref{eq:xIDy}) of the $_x\text{ID} _y$. 
		\begin{equation}
		\begin{aligned}
		{}^{}_x\text{ID}_y^{d,s} & := \frac{1}{\sum_{k \in \mathbb{T}_{x,y}^{d,s}\cap \mathcal{T}^{d,s}} V_k^{d,s}} \sum_{k \in \mathbb{T}_{x,y}^{d,s}\cap \mathcal{T}^{d,s}} V_k^{d,s}P_k^{d,s} \\
		& = \frac{1}{\sum_{k \in \left(\dot{\bigcup}_{j} \mathbb{T}_{x_j,y_j}^{d,s} \right) \cap \mathcal{T}^{d,s}} V_k^{d,s}} \sum_{k \in \left(\dot{\bigcup}_{j} \mathbb{T}_{x_j,y_j}^{d,s} \right) \cap \mathcal{T}^{d,s}} V_k^{d,s}P_k^{d,s}	\\	
		& = \frac{1}{\sum_{k \in \dot{\bigcup}_{j} \left( \mathbb{T}_{x_j,y_j}^{d,s}  \cap \mathcal{T}^{d,s} \right)} V_k^{d,s}} \sum_{k \in \dot{\bigcup}_{j} \left( \mathbb{T}_{x_j,y_j}^{d,s}  \cap \mathcal{T}^{d,s} \right)} V_k^{d,s}P_k^{d,s}	\\
		& = \frac{1}{\sum_j \sum_{k \in   \mathbb{T}_{x_j,y_j}^{d,s}  \cap \mathcal{T}^{d,s} } V_k^{d,s}}  \sum_j \sum_{k \in \mathbb{T}_{x_j,y_j}^{d,s}  \cap \mathcal{T}^{d,s} } V_k^{d,s}P_k^{d,s}	\\
		& = \frac{1}{\sum_j \mathbb{V}_{x_j,y_j}^{d,s}}  \sum_j \left( \frac{\mathbb{V}_{x_j,y_j}^{d,s}}{\mathbb{V}_{x_j,y_j}^{d,s}} \sum_{k \in \mathbb{T}_{x_j,y_j}^{d,s}  \cap \mathcal{T}^{d,s} } V_k^{d,s}P_k^{d,s}  \right)\\
		& = \frac{1}{\sum_j \mathbb{V}_{x_j,y_j}^{d,s}}  \sum_j \mathbb{V}_{x_j,y_j}^{d,s} {}_{x_j}\text{ID}^{d,s}_{y_j} = \frac{\sum_j {}_{x_j}\text{ID}^{d,s}_{y_j} \mathbb{V}_{x_j,y_j}^{d,s}}{\sum_j \mathbb{V}_{x_j,y_j}^{d,s}}.
		\end{aligned}
		\end{equation}
	\end{proof}
	
	\subsection*{Quarter-hourly products coefficients relevance}
	In this section, we present the tables consisting of the most relevant coefficients in the model \textbf{FI.elnet.penal.C} for each quarter-hourly product. 
	\begin{table}[!ht]
		\centering
		\begingroup\footnotesize
		\begin{tabular}{rllll}
			\hline
			& Importance: 1 & Importance: 2 & Importance: 3  & Importance: 4\\ 
	\hline
	00:00 & \cellcolor[rgb]{0.5,0.9,0.5} {$\text{IA}^{d,\text{00:00}}$ (27.23)} & \cellcolor[rgb]{0.826,1,0.5} {$\prescript{QH}{3.25}{\text{ID}}_{0.25}^{d, \text{00:00}}$ (11.35)} & \cellcolor[rgb]{0.996,1,0.5} {$\prescript{QH}{0.5}{\text{ID}}_{0.25}^{d-1, \text{21:15}}$ (7.09)} & \cellcolor[rgb]{1,0.926,0.5} {$\prescript{QH}{3.5}{\text{ID}}_{0.25}^{d, \text{00:00}}$ (5.52)} \\ 
	00:15 & \cellcolor[rgb]{0.952,1,0.5} {$\prescript{H}{1}{\text{ID}}_{0.25}^{d-1, \text{22:00}}$ (8.2)} & \cellcolor[rgb]{0.972,1,0.5} {$\text{IA}^{d,\text{00:15}}$ (7.71)} & \cellcolor[rgb]{1,1,0.5} {$\prescript{H}{3}{\text{ID}}_{0.25}^{d, \text{00:00}}$ (6.99)} & \cellcolor[rgb]{1,0.954,0.5} {$\prescript{H}{2}{\text{ID}}_{0.25}^{d-1, \text{23:00}}$ (6.09)} \\ 
	00:30 & \cellcolor[rgb]{0.973,1,0.5} {$\prescript{H}{2.75}{\text{ID}}_{0.25}^{d, \text{00:00}}$ (7.67)} & \cellcolor[rgb]{1,0.99,0.5} {$\prescript{H}{3.75}{\text{ID}}_{0.25}^{d, \text{01:00}}$ (6.8)} & \cellcolor[rgb]{1,0.982,0.5} {$\prescript{H}{4}{\text{ID}}_{0.25}^{d, \text{01:00}}$ (6.65)} & \cellcolor[rgb]{1,0.854,0.5} {$\text{IA}^{d,\text{00:30}}$ (4.07)} \\ 
	00:45 & \cellcolor[rgb]{0.506,0.902,0.5} {$\text{IA}^{d,\text{00:45}}$ (24.74)} & \cellcolor[rgb]{0.941,1,0.5} {$\prescript{H}{3.5}{\text{ID}}_{0.25}^{d, \text{01:00}}$ (8.48)} & \cellcolor[rgb]{0.995,1,0.5} {$\prescript{H}{2.5}{\text{ID}}_{0.25}^{d, \text{00:00}}$ (7.12)} & \cellcolor[rgb]{1,0.995,0.5} {$\prescript{QH}{3.25}{\text{ID}}_{0.25}^{d, \text{00:45}}$ (6.9)} \\ 
	01:00 & \cellcolor[rgb]{0.713,0.971,0.5} {$\text{IA}^{d,\text{01:00}}$ (15.75)} & \cellcolor[rgb]{0.958,1,0.5} {$\prescript{H}{2.25}{\text{ID}}_{0.25}^{d, \text{00:00}}$ (8.06)} & \cellcolor[rgb]{1,0.952,0.5} {$\prescript{QH}{3.25}{\text{ID}}_{0.25}^{d, \text{01:00}}$ (6.04)} & \cellcolor[rgb]{1,0.843,0.5} {$\prescript{H}{1.25}{\text{ID}}_{0.25}^{d-1, \text{23:00}}$ (3.86)} \\ 
	01:15 & \cellcolor[rgb]{0.664,0.955,0.5} {$\prescript{H}{2.25}{\text{ID}}_{0.25}^{d, \text{00:00}}$ (17.91)} & \cellcolor[rgb]{0.755,0.985,0.5} {$\prescript{H}{2}{\text{ID}}_{0.25}^{d, \text{00:00}}$ (13.95)} & \cellcolor[rgb]{0.799,1,0.5} {$\prescript{H}{3}{\text{ID}}_{0.25}^{d, \text{01:00}}$ (12.03)} & \cellcolor[rgb]{0.875,1,0.5} {$\prescript{H}{3.25}{\text{ID}}_{0.25}^{d, \text{01:00}}$ (10.13)} \\ 
	01:30 & \cellcolor[rgb]{0.795,0.998,0.5} {$\prescript{H}{2.75}{\text{ID}}_{0.25}^{d, \text{01:00}}$ (12.22)} & \cellcolor[rgb]{0.927,1,0.5} {$\prescript{H}{1.75}{\text{ID}}_{0.25}^{d, \text{00:00}}$ (8.82)} & \cellcolor[rgb]{0.958,1,0.5} {$\prescript{H}{3}{\text{ID}}_{0.25}^{d, \text{01:00}}$ (8.04)} & \cellcolor[rgb]{0.998,1,0.5} {$\prescript{H}{3.25}{\text{ID}}_{0.25}^{d, \text{01:00}}$ (7.06)} \\ 
	01:45 & \cellcolor[rgb]{0.618,0.939,0.5} {$\text{IA}^{d,\text{01:45}}$ (19.91)} & \cellcolor[rgb]{0.751,0.984,0.5} {$\prescript{H}{2.5}{\text{ID}}_{0.25}^{d, \text{01:00}}$ (14.14)} & \cellcolor[rgb]{1,0.955,0.5} {$\prescript{H}{1.5}{\text{ID}}_{0.25}^{d, \text{00:00}}$ (6.09)} & \cellcolor[rgb]{1,0.848,0.5} {$\prescript{QH}{2.25}{\text{ID}}_{0.25}^{d, \text{00:45}}$ (3.96)} \\ 
	02:00 & \cellcolor[rgb]{0.832,1,0.5} {$\text{IA}^{d,\text{02:00}}$ (11.19)} & \cellcolor[rgb]{0.88,1,0.5} {$\prescript{H}{1.25}{\text{ID}}_{0.25}^{d, \text{00:00}}$ (10)} & \cellcolor[rgb]{0.987,1,0.5} {$\prescript{H}{2.25}{\text{ID}}_{0.25}^{d, \text{01:00}}$ (7.33)} & \cellcolor[rgb]{1,0.932,0.5} {$\prescript{QH}{9.75}{\text{ID}}_{0.25}^{d, \text{02:00}}$ (5.64)} \\ 
	02:15 & \cellcolor[rgb]{0.828,1,0.5} {$\prescript{H}{1}{\text{ID}}_{0.25}^{d, \text{00:00}}$ (11.29)} & \cellcolor[rgb]{0.934,1,0.5} {$\prescript{H}{2}{\text{ID}}_{0.25}^{d, \text{01:00}}$ (8.64)} & \cellcolor[rgb]{0.954,1,0.5} {$\prescript{H}{3}{\text{ID}}_{0.25}^{d, \text{02:00}}$ (8.14)} & \cellcolor[rgb]{1,0.88,0.5} {$\prescript{H}{2.25}{\text{ID}}_{0.25}^{d, \text{01:00}}$ (4.59)} \\ 
	02:30 & \cellcolor[rgb]{0.87,1,0.5} {$\prescript{H}{1.75}{\text{ID}}_{0.25}^{d, \text{01:00}}$ (10.26)} & \cellcolor[rgb]{1,0.976,0.5} {$\prescript{H}{2.75}{\text{ID}}_{0.25}^{d, \text{02:00}}$ (6.53)} & \cellcolor[rgb]{1,0.905,0.5} {$\prescript{H}{5.75}{\text{ID}}_{0.25}^{d, \text{05:00}}$ (5.1)} & \cellcolor[rgb]{1,0.865,0.5} {$\prescript{H}{2.5}{\text{ID}}_{0.25}^{d, \text{01:00}}$ (4.29)} \\ 
	02:45 & \cellcolor[rgb]{0.88,1,0.5} {$\prescript{QH}{1.25}{\text{ID}}_{0.25}^{d, \text{00:45}}$ (9.99)} & \cellcolor[rgb]{0.886,1,0.5} {$\text{IA}^{d,\text{02:45}}$ (9.86)} & \cellcolor[rgb]{1,0.961,0.5} {$\text{IA}^{d,\text{03:45}}$ (6.21)} & \cellcolor[rgb]{1,0.92,0.5} {$\prescript{H}{2.5}{\text{ID}}_{0.25}^{d, \text{02:00}}$ (5.41)} \\ 
	03:00 & \cellcolor[rgb]{0.789,0.996,0.5} {$\text{IA}^{d,\text{03:00}}$ (12.46)} & \cellcolor[rgb]{0.999,1,0.5} {$\prescript{H}{1.25}{\text{ID}}_{0.25}^{d, \text{01:00}}$ (7.03)} & \cellcolor[rgb]{1,0.997,0.5} {$\prescript{QH}{0.5}{\text{ID}}_{0.25}^{d, \text{00:15}}$ (6.95)} & \cellcolor[rgb]{1,0.946,0.5} {$\prescript{QH}{3.25}{\text{ID}}_{0.25}^{d, \text{03:00}}$ (5.93)} \\ 
	03:15 & \cellcolor[rgb]{0.792,0.997,0.5} {$\prescript{H}{1}{\text{ID}}_{0.25}^{d, \text{01:00}}$ (12.36)} & \cellcolor[rgb]{0.926,1,0.5} {$\prescript{H}{1.25}{\text{ID}}_{0.25}^{d, \text{01:00}}$ (8.85)} & \cellcolor[rgb]{1,0.947,0.5} {$\prescript{H}{3}{\text{ID}}_{0.25}^{d, \text{03:00}}$ (5.95)} & \cellcolor[rgb]{1,0.87,0.5} {$\prescript{H}{2}{\text{ID}}_{0.25}^{d, \text{02:00}}$ (4.39)} \\ 
	03:30 & \cellcolor[rgb]{0.963,1,0.5} {$\prescript{H}{2.75}{\text{ID}}_{0.25}^{d, \text{03:00}}$ (7.92)} & \cellcolor[rgb]{1,0.905,0.5} {$\prescript{H}{4.75}{\text{ID}}_{0.25}^{d, \text{05:00}}$ (5.11)} & \cellcolor[rgb]{1,0.873,0.5} {$\prescript{H}{5}{\text{ID}}_{0.25}^{d, \text{05:00}}$ (4.45)} & \cellcolor[rgb]{1,0.845,0.5} {$\prescript{H}{1.75}{\text{ID}}_{0.25}^{d, \text{02:00}}$ (3.91)} \\ 
	03:45 & \cellcolor[rgb]{0.84,1,0.5} {$\text{IA}^{d,\text{03:45}}$ (10.99)} & \cellcolor[rgb]{0.956,1,0.5} {$\prescript{QH}{3.25}{\text{ID}}_{0.25}^{d, \text{03:45}}$ (8.11)} & \cellcolor[rgb]{1,0.973,0.5} {$\prescript{QH}{2.25}{\text{ID}}_{0.25}^{d, \text{02:45}}$ (6.46)} & \cellcolor[rgb]{1,0.967,0.5} {$\prescript{H}{4.5}{\text{ID}}_{0.25}^{d, \text{05:00}}$ (6.34)} \\ 
	04:00 & \cellcolor[rgb]{0.594,0.931,0.5} {$\text{IA}^{d,\text{04:00}}$ (20.91)} & \cellcolor[rgb]{0.88,1,0.5} {$\prescript{QH}{11.75}{\text{ID}}_{0.25}^{d, \text{04:00}}$ (9.99)} & \cellcolor[rgb]{1,0.989,0.5} {$\prescript{QH}{3.25}{\text{ID}}_{0.25}^{d, \text{04:00}}$ (6.78)} & \cellcolor[rgb]{1,0.829,0.5} {$\prescript{H}{2.75}{\text{ID}}_{0.25}^{d, \text{03:00}}$ (3.57)} \\ 
	04:15 & \cellcolor[rgb]{0.846,1,0.5} {$\prescript{H}{2}{\text{ID}}_{0.25}^{d, \text{03:00}}$ (10.84)} & \cellcolor[rgb]{0.923,1,0.5} {$\prescript{H}{3}{\text{ID}}_{0.25}^{d, \text{04:00}}$ (8.93)} & \cellcolor[rgb]{0.934,1,0.5} {$\prescript{H}{1}{\text{ID}}_{0.25}^{d, \text{02:00}}$ (8.64)} & \cellcolor[rgb]{1,0.873,0.5} {$\prescript{H}{2.25}{\text{ID}}_{0.25}^{d, \text{03:00}}$ (4.46)} \\ 
	04:30 & \cellcolor[rgb]{0.818,1,0.5} {$\prescript{H}{2.75}{\text{ID}}_{0.25}^{d, \text{04:00}}$ (11.54)} & \cellcolor[rgb]{0.977,1,0.5} {$\prescript{H}{3}{\text{ID}}_{0.25}^{d, \text{04:00}}$ (7.57)} & \cellcolor[rgb]{1,0.971,0.5} {$\prescript{H}{4.25}{\text{ID}}_{0.25}^{d, \text{05:00}}$ (6.43)} & \cellcolor[rgb]{1,0.959,0.5} {$\prescript{H}{4.5}{\text{ID}}_{0.25}^{d, \text{05:00}}$ (6.17)} \\ 
	04:45 & \cellcolor[rgb]{0.832,1,0.5} {$\text{IA}^{d,\text{04:45}}$ (11.21)} & \cellcolor[rgb]{0.938,1,0.5} {$\prescript{QH}{3.25}{\text{ID}}_{0.25}^{d, \text{04:45}}$ (8.54)} & \cellcolor[rgb]{1,0.899,0.5} {$\prescript{H}{3.5}{\text{ID}}_{0.25}^{d, \text{05:00}}$ (4.98)} & \cellcolor[rgb]{1,0.878,0.5} {$\prescript{H}{4.5}{\text{ID}}_{0.25}^{d, \text{05:00}}$ (4.55)} \\ 
	05:00 & \cellcolor[rgb]{0.664,0.955,0.5} {$\text{IA}^{d,\text{05:00}}$ (17.9)} & \cellcolor[rgb]{0.945,1,0.5} {$\prescript{QH}{3.25}{\text{ID}}_{0.25}^{d, \text{05:00}}$ (8.38)} & \cellcolor[rgb]{1,0.981,0.5} {$\prescript{H}{1.25}{\text{ID}}_{0.25}^{d, \text{03:00}}$ (6.62)} & \cellcolor[rgb]{1,0.798,0.5} {$\prescript{H}{2.25}{\text{ID}}_{0.25}^{d, \text{04:00}}$ (2.98)} \\ 
	05:15 & \cellcolor[rgb]{0.932,1,0.5} {$\prescript{H}{3}{\text{ID}}_{0.25}^{d, \text{05:00}}$ (8.7)} & \cellcolor[rgb]{0.967,1,0.5} {$\prescript{QH}{3.25}{\text{ID}}_{0.25}^{d, \text{05:15}}$ (7.82)} & \cellcolor[rgb]{1,0.901,0.5} {$\prescript{H}{2}{\text{ID}}_{0.25}^{d, \text{04:00}}$ (5.02)} & \cellcolor[rgb]{1,0.87,0.5} {$\prescript{H}{1}{\text{ID}}_{0.25}^{d, \text{03:00}}$ (4.39)} \\ 
	05:30 & \cellcolor[rgb]{0.932,1,0.5} {$\prescript{H}{3.75}{\text{ID}}_{0.25}^{d, \text{06:00}}$ (8.7)} & \cellcolor[rgb]{1,0.994,0.5} {$\prescript{H}{3}{\text{ID}}_{0.25}^{d, \text{05:00}}$ (6.87)} & \cellcolor[rgb]{1,0.913,0.5} {$\prescript{H}{2.75}{\text{ID}}_{0.25}^{d, \text{05:00}}$ (5.26)} & \cellcolor[rgb]{1,0.886,0.5} {$\prescript{QH}{3.25}{\text{ID}}_{0.25}^{d, \text{05:30}}$ (4.71)} \\ 
	05:45 & \cellcolor[rgb]{0.923,1,0.5} {$\text{IA}^{d,\text{05:45}}$ (8.93)} & \cellcolor[rgb]{0.963,1,0.5} {$\prescript{QH}{3.25}{\text{ID}}_{0.25}^{d, \text{05:45}}$ (7.93)} & \cellcolor[rgb]{1,0.972,0.5} {$\prescript{QH}{4.25}{\text{ID}}_{0.25}^{d, \text{05:45}}$ (6.43)} & \cellcolor[rgb]{1,0.964,0.5} {$\prescript{H}{3.5}{\text{ID}}_{0.25}^{d, \text{06:00}}$ (6.28)} \\ 
	06:00 & \cellcolor[rgb]{0.709,0.97,0.5} {$\prescript{QH}{3.25}{\text{ID}}_{0.25}^{d, \text{06:00}}$ (15.94)} & \cellcolor[rgb]{0.725,0.975,0.5} {$\text{IA}^{d,\text{06:00}}$ (15.24)} & \cellcolor[rgb]{1,0.916,0.5} {$\prescript{QH}{4.75}{\text{ID}}_{0.25}^{d, \text{06:00}}$ (5.32)} & \cellcolor[rgb]{1,0.864,0.5} {$\prescript{QH}{4.5}{\text{ID}}_{0.25}^{d, \text{06:00}}$ (4.27)} \\ 
	06:15 & \cellcolor[rgb]{0.706,0.969,0.5} {$\prescript{H}{3}{\text{ID}}_{0.25}^{d, \text{06:00}}$ (16.08)} & \cellcolor[rgb]{0.74,0.98,0.5} {$\prescript{H}{3.25}{\text{ID}}_{0.25}^{d, \text{06:00}}$ (14.59)} & \cellcolor[rgb]{0.912,1,0.5} {$\prescript{H}{3.5}{\text{ID}}_{0.25}^{d, \text{06:00}}$ (9.19)} & \cellcolor[rgb]{0.936,1,0.5} {$\prescript{QH}{3.25}{\text{ID}}_{0.25}^{d, \text{06:15}}$ (8.6)} \\ 
	06:30 & \cellcolor[rgb]{0.901,1,0.5} {$\prescript{H}{2.75}{\text{ID}}_{0.25}^{d, \text{06:00}}$ (9.46)} & \cellcolor[rgb]{0.943,1,0.5} {$\prescript{QH}{3.25}{\text{ID}}_{0.25}^{d, \text{06:30}}$ (8.42)} & \cellcolor[rgb]{1,0.959,0.5} {$\prescript{H}{3}{\text{ID}}_{0.25}^{d, \text{06:00}}$ (6.19)} & \cellcolor[rgb]{1,0.896,0.5} {$\prescript{H}{3.75}{\text{ID}}_{0.25}^{d, \text{07:00}}$ (4.91)} \\ 
	06:45 & \cellcolor[rgb]{0.748,0.983,0.5} {$\text{IA}^{d,\text{06:45}}$ (14.27)} & \cellcolor[rgb]{0.786,0.995,0.5} {$\prescript{QH}{3.25}{\text{ID}}_{0.25}^{d, \text{06:45}}$ (12.63)} & \cellcolor[rgb]{0.982,1,0.5} {$\prescript{QH}{3.75}{\text{ID}}_{0.25}^{d, \text{06:45}}$ (7.45)} & \cellcolor[rgb]{1,1,0.5} {$\prescript{H}{2.5}{\text{ID}}_{0.25}^{d, \text{06:00}}$ (7)} \\ 
	07:00 & \cellcolor[rgb]{0.52,0.907,0.5} {$\text{IA}^{d,\text{07:00}}$ (24.12)} & \cellcolor[rgb]{0.834,1,0.5} {$\prescript{H}{2.25}{\text{ID}}_{0.25}^{d, \text{06:00}}$ (11.14)} & \cellcolor[rgb]{0.919,1,0.5} {$\prescript{QH}{3.25}{\text{ID}}_{0.25}^{d, \text{07:00}}$ (9.02)} & \cellcolor[rgb]{0.98,1,0.5} {$\prescript{H}{2.5}{\text{ID}}_{0.25}^{d, \text{06:00}}$ (7.5)} \\ 
	07:15 & \cellcolor[rgb]{0.744,0.981,0.5} {$\prescript{H}{3}{\text{ID}}_{0.25}^{d, \text{07:00}}$ (14.41)} & \cellcolor[rgb]{0.849,1,0.5} {$\prescript{H}{3.25}{\text{ID}}_{0.25}^{d, \text{07:00}}$ (10.77)} & \cellcolor[rgb]{0.879,1,0.5} {$\prescript{H}{2}{\text{ID}}_{0.25}^{d, \text{06:00}}$ (10.03)} & \cellcolor[rgb]{1,0.934,0.5} {$\prescript{H}{3.5}{\text{ID}}_{0.25}^{d, \text{07:00}}$ (5.68)} \\ 
	07:30 & \cellcolor[rgb]{0.828,1,0.5} {$\prescript{H}{2.75}{\text{ID}}_{0.25}^{d, \text{07:00}}$ (11.29)} & \cellcolor[rgb]{0.863,1,0.5} {$\prescript{H}{3}{\text{ID}}_{0.25}^{d, \text{07:00}}$ (10.43)} & \cellcolor[rgb]{1,0.991,0.5} {$\prescript{H}{3.25}{\text{ID}}_{0.25}^{d, \text{07:00}}$ (6.81)} & \cellcolor[rgb]{1,0.977,0.5} {$\prescript{H}{1.75}{\text{ID}}_{0.25}^{d, \text{06:00}}$ (6.53)} \\ 
	07:45 & \cellcolor[rgb]{0.847,1,0.5} {$\text{IA}^{d,\text{07:45}}$ (10.82)} & \cellcolor[rgb]{0.853,1,0.5} {$\prescript{H}{2.5}{\text{ID}}_{0.25}^{d, \text{07:00}}$ (10.67)} & \cellcolor[rgb]{0.938,1,0.5} {$\prescript{QH}{3.25}{\text{ID}}_{0.25}^{d, \text{07:45}}$ (8.55)} & \cellcolor[rgb]{1,0.913,0.5} {$\prescript{QH}{2.25}{\text{ID}}_{0.25}^{d, \text{06:45}}$ (5.25)} \\ 
			\hline
		\end{tabular}
		\endgroup
		\caption{Most relevant coefficients in the model \textbf{FI.elnet.penal.C} for each quarter-hourly product from 00:00 to 07:45} 
		\label{tab:beta_qh_1}
	\end{table}
	
	\begin{table}[!ht]
		\centering
		\begingroup\footnotesize
		\begin{tabular}{rllll}
			\hline
			& Importance: 1 & Importance: 2 & Importance: 3  & Importance: 4\\ 
			\hline
			08:00 & \cellcolor[rgb]{0.611,0.937,0.5} {$\text{IA}^{d,\text{08:00}}$ (20.21)} & \cellcolor[rgb]{0.787,0.996,0.5} {$\prescript{H}{2.25}{\text{ID}}_{0.25}^{d, \text{07:00}}$ (12.55)} & \cellcolor[rgb]{0.99,1,0.5} {$\prescript{H}{2.5}{\text{ID}}_{0.25}^{d, \text{07:00}}$ (7.26)} & \cellcolor[rgb]{1,0.988,0.5} {$\prescript{H}{1.25}{\text{ID}}_{0.25}^{d, \text{06:00}}$ (6.77)} \\ 
			08:15 & \cellcolor[rgb]{0.784,0.995,0.5} {$\prescript{H}{2}{\text{ID}}_{0.25}^{d, \text{07:00}}$ (12.69)} & \cellcolor[rgb]{0.785,0.995,0.5} {$\prescript{H}{3}{\text{ID}}_{0.25}^{d, \text{08:00}}$ (12.65)} & \cellcolor[rgb]{0.924,1,0.5} {$\prescript{H}{2.25}{\text{ID}}_{0.25}^{d, \text{07:00}}$ (8.89)} & \cellcolor[rgb]{0.931,1,0.5} {$\prescript{H}{3.25}{\text{ID}}_{0.25}^{d, \text{08:00}}$ (8.72)} \\ 
			08:30 & \cellcolor[rgb]{0.924,1,0.5} {$\prescript{H}{2.75}{\text{ID}}_{0.25}^{d, \text{08:00}}$ (8.9)} & \cellcolor[rgb]{0.932,1,0.5} {$\prescript{H}{3.75}{\text{ID}}_{0.25}^{d, \text{09:00}}$ (8.71)} & \cellcolor[rgb]{0.96,1,0.5} {$\prescript{H}{1.75}{\text{ID}}_{0.25}^{d, \text{07:00}}$ (8.01)} & \cellcolor[rgb]{1,0.94,0.5} {$\prescript{H}{4}{\text{ID}}_{0.25}^{d, \text{09:00}}$ (5.79)} \\ 
			08:45 & \cellcolor[rgb]{0.959,1,0.5} {$\prescript{QH}{3.25}{\text{ID}}_{0.25}^{d, \text{08:45}}$ (8.03)} & \cellcolor[rgb]{1,0.956,0.5} {$\prescript{QH}{3.5}{\text{ID}}_{0.25}^{d, \text{08:45}}$ (6.11)} & \cellcolor[rgb]{1,0.908,0.5} {$\prescript{H}{4.5}{\text{ID}}_{0.25}^{d, \text{10:00}}$ (5.16)} & \cellcolor[rgb]{1,0.9,0.5} {$\prescript{H}{3.5}{\text{ID}}_{0.25}^{d, \text{09:00}}$ (4.99)} \\ 
			09:00 & \cellcolor[rgb]{0.793,0.998,0.5} {$\prescript{H}{2.25}{\text{ID}}_{0.25}^{d, \text{08:00}}$ (12.31)} & \cellcolor[rgb]{0.868,1,0.5} {$\text{IA}^{d,\text{09:00}}$ (10.3)} & \cellcolor[rgb]{0.881,1,0.5} {$\prescript{H}{1.25}{\text{ID}}_{0.25}^{d, \text{07:00}}$ (9.98)} & \cellcolor[rgb]{1,0.978,0.5} {$\prescript{H}{2.5}{\text{ID}}_{0.25}^{d, \text{08:00}}$ (6.55)} \\ 
			09:15 & \cellcolor[rgb]{0.787,0.996,0.5} {$\prescript{H}{3}{\text{ID}}_{0.25}^{d, \text{09:00}}$ (12.57)} & \cellcolor[rgb]{0.95,1,0.5} {$\prescript{H}{2}{\text{ID}}_{0.25}^{d, \text{08:00}}$ (8.24)} & \cellcolor[rgb]{0.986,1,0.5} {$\prescript{H}{3.25}{\text{ID}}_{0.25}^{d, \text{09:00}}$ (7.34)} & \cellcolor[rgb]{1,0.912,0.5} {$\prescript{H}{3.5}{\text{ID}}_{0.25}^{d, \text{09:00}}$ (5.24)} \\ 
			09:30 & \cellcolor[rgb]{0.819,1,0.5} {$\prescript{H}{2.75}{\text{ID}}_{0.25}^{d, \text{09:00}}$ (11.52)} & \cellcolor[rgb]{0.861,1,0.5} {$\prescript{H}{3.75}{\text{ID}}_{0.25}^{d, \text{10:00}}$ (10.47)} & \cellcolor[rgb]{1,0.998,0.5} {$\prescript{H}{4.75}{\text{ID}}_{0.25}^{d, \text{11:00}}$ (6.95)} & \cellcolor[rgb]{1,0.925,0.5} {$\prescript{QH}{1.5}{\text{ID}}_{0.25}^{d, \text{07:45}}$ (5.5)} \\ 
			09:45 & \cellcolor[rgb]{0.854,1,0.5} {$\prescript{QH}{2.25}{\text{ID}}_{0.25}^{d, \text{08:45}}$ (10.65)} & \cellcolor[rgb]{0.864,1,0.5} {$\prescript{H}{3.5}{\text{ID}}_{0.25}^{d, \text{10:00}}$ (10.41)} & \cellcolor[rgb]{0.95,1,0.5} {$\prescript{H}{4.5}{\text{ID}}_{0.25}^{d, \text{11:00}}$ (8.26)} & \cellcolor[rgb]{1,0.978,0.5} {$\prescript{H}{2.5}{\text{ID}}_{0.25}^{d, \text{09:00}}$ (6.57)} \\ 
			10:00 & \cellcolor[rgb]{0.918,1,0.5} {$\prescript{H}{2.25}{\text{ID}}_{0.25}^{d, \text{09:00}}$ (9.05)} & \cellcolor[rgb]{1,0.986,0.5} {$\prescript{H}{1.25}{\text{ID}}_{0.25}^{d, \text{08:00}}$ (6.71)} & \cellcolor[rgb]{1,0.972,0.5} {$\prescript{QH}{3.25}{\text{ID}}_{0.25}^{d, \text{10:00}}$ (6.44)} & \cellcolor[rgb]{1,0.878,0.5} {$\prescript{H}{3.25}{\text{ID}}_{0.25}^{d, \text{10:00}}$ (4.55)} \\ 
			10:15 & \cellcolor[rgb]{0.845,1,0.5} {$\prescript{H}{3}{\text{ID}}_{0.25}^{d, \text{10:00}}$ (10.87)} & \cellcolor[rgb]{0.891,1,0.5} {$\prescript{H}{2}{\text{ID}}_{0.25}^{d, \text{09:00}}$ (9.71)} & \cellcolor[rgb]{1,0.99,0.5} {$\prescript{H}{3.25}{\text{ID}}_{0.25}^{d, \text{10:00}}$ (6.8)} & \cellcolor[rgb]{1,0.965,0.5} {$\prescript{H}{2.25}{\text{ID}}_{0.25}^{d, \text{09:00}}$ (6.3)} \\ 
			10:30 & \cellcolor[rgb]{0.977,1,0.5} {$\prescript{H}{3.75}{\text{ID}}_{0.25}^{d, \text{11:00}}$ (7.57)} & \cellcolor[rgb]{1,0.976,0.5} {$\prescript{H}{4}{\text{ID}}_{0.25}^{d, \text{11:00}}$ (6.52)} & \cellcolor[rgb]{1,0.966,0.5} {$\prescript{QH}{1.5}{\text{ID}}_{0.25}^{d, \text{08:45}}$ (6.31)} & \cellcolor[rgb]{1,0.949,0.5} {$\prescript{H}{2.75}{\text{ID}}_{0.25}^{d, \text{10:00}}$ (5.97)} \\ 
			10:45 & \cellcolor[rgb]{0.891,1,0.5} {$\text{IA}^{d,\text{10:45}}$ (9.72)} & \cellcolor[rgb]{0.929,1,0.5} {$\prescript{H}{3.5}{\text{ID}}_{0.25}^{d, \text{11:00}}$ (8.78)} & \cellcolor[rgb]{0.938,1,0.5} {$\prescript{QH}{3.25}{\text{ID}}_{0.25}^{d, \text{10:45}}$ (8.55)} & \cellcolor[rgb]{0.977,1,0.5} {$\prescript{QH}{1.25}{\text{ID}}_{0.25}^{d, \text{08:45}}$ (7.58)} \\ 
			11:00 & \cellcolor[rgb]{0.914,1,0.5} {$\prescript{H}{2.25}{\text{ID}}_{0.25}^{d, \text{10:00}}$ (9.14)} & \cellcolor[rgb]{0.956,1,0.5} {$\prescript{H}{3.25}{\text{ID}}_{0.25}^{d, \text{11:00}}$ (8.1)} & \cellcolor[rgb]{0.988,1,0.5} {$\prescript{H}{1.25}{\text{ID}}_{0.25}^{d, \text{09:00}}$ (7.3)} & \cellcolor[rgb]{1,0.881,0.5} {$\prescript{QH}{0.75}{\text{ID}}_{0.25}^{d, \text{08:30}}$ (4.61)} \\ 
			11:15 & \cellcolor[rgb]{0.82,1,0.5} {$\prescript{H}{3}{\text{ID}}_{0.25}^{d, \text{11:00}}$ (11.51)} & \cellcolor[rgb]{0.943,1,0.5} {$\prescript{H}{2}{\text{ID}}_{0.25}^{d, \text{10:00}}$ (8.43)} & \cellcolor[rgb]{0.964,1,0.5} {$\prescript{H}{3.25}{\text{ID}}_{0.25}^{d, \text{11:00}}$ (7.91)} & \cellcolor[rgb]{1,0.913,0.5} {$\prescript{H}{4}{\text{ID}}_{0.25}^{d, \text{12:00}}$ (5.25)} \\ 
			11:30 & \cellcolor[rgb]{0.926,1,0.5} {$\prescript{H}{2.75}{\text{ID}}_{0.25}^{d, \text{11:00}}$ (8.85)} & \cellcolor[rgb]{0.937,1,0.5} {$\prescript{QH}{2.5}{\text{ID}}_{0.25}^{d, \text{10:45}}$ (8.58)} & \cellcolor[rgb]{0.967,1,0.5} {$\prescript{H}{3.75}{\text{ID}}_{0.25}^{d, \text{12:00}}$ (7.82)} & \cellcolor[rgb]{1,0.97,0.5} {$\prescript{H}{3}{\text{ID}}_{0.25}^{d, \text{11:00}}$ (6.4)} \\ 
			11:45 & \cellcolor[rgb]{0.975,1,0.5} {$\prescript{QH}{3.25}{\text{ID}}_{0.25}^{d, \text{11:45}}$ (7.62)} & \cellcolor[rgb]{1,0.976,0.5} {$\prescript{QH}{2.25}{\text{ID}}_{0.25}^{d, \text{10:45}}$ (6.51)} & \cellcolor[rgb]{1,0.92,0.5} {$\prescript{H}{4.5}{\text{ID}}_{0.25}^{d, \text{13:00}}$ (5.39)} & \cellcolor[rgb]{1,0.915,0.5} {$\prescript{H}{3.5}{\text{ID}}_{0.25}^{d, \text{12:00}}$ (5.3)} \\ 
			12:00 & \cellcolor[rgb]{0.985,1,0.5} {$\prescript{QH}{3.25}{\text{ID}}_{0.25}^{d, \text{12:00}}$ (7.38)} & \cellcolor[rgb]{1,0.953,0.5} {$\prescript{H}{2.25}{\text{ID}}_{0.25}^{d, \text{11:00}}$ (6.06)} & \cellcolor[rgb]{1,0.898,0.5} {$\prescript{H}{1.25}{\text{ID}}_{0.25}^{d, \text{10:00}}$ (4.97)} & \cellcolor[rgb]{1,0.873,0.5} {$\prescript{QH}{1.75}{\text{ID}}_{0.25}^{d, \text{10:15}}$ (4.47)} \\ 
			12:15 & \cellcolor[rgb]{0.925,1,0.5} {$\prescript{QH}{1.75}{\text{ID}}_{0.25}^{d, \text{10:45}}$ (8.87)} & \cellcolor[rgb]{0.974,1,0.5} {$\prescript{H}{3}{\text{ID}}_{0.25}^{d, \text{12:00}}$ (7.66)} & \cellcolor[rgb]{1,0.894,0.5} {$\prescript{H}{4}{\text{ID}}_{0.25}^{d, \text{13:00}}$ (4.88)} & \cellcolor[rgb]{1,0.893,0.5} {$\prescript{H}{2}{\text{ID}}_{0.25}^{d, \text{11:00}}$ (4.85)} \\ 
			12:30 & \cellcolor[rgb]{0.886,1,0.5} {$\prescript{H}{3.75}{\text{ID}}_{0.25}^{d, \text{13:00}}$ (9.85)} & \cellcolor[rgb]{1,0.983,0.5} {$\prescript{H}{2.75}{\text{ID}}_{0.25}^{d, \text{12:00}}$ (6.66)} & \cellcolor[rgb]{1,0.958,0.5} {$\prescript{QH}{2.5}{\text{ID}}_{0.25}^{d, \text{11:45}}$ (6.17)} & \cellcolor[rgb]{1,0.932,0.5} {$\prescript{H}{4.75}{\text{ID}}_{0.25}^{d, \text{14:00}}$ (5.64)} \\ 
			12:45 & \cellcolor[rgb]{0.909,1,0.5} {$\prescript{QH}{2.25}{\text{ID}}_{0.25}^{d, \text{11:45}}$ (9.27)} & \cellcolor[rgb]{0.925,1,0.5} {$\prescript{H}{3.5}{\text{ID}}_{0.25}^{d, \text{13:00}}$ (8.86)} & \cellcolor[rgb]{0.964,1,0.5} {$\prescript{H}{4.5}{\text{ID}}_{0.25}^{d, \text{14:00}}$ (7.9)} & \cellcolor[rgb]{1,0.876,0.5} {$\prescript{QH}{1.25}{\text{ID}}_{0.25}^{d, \text{10:45}}$ (4.52)} \\ 
			13:00 & \cellcolor[rgb]{1,0.997,0.5} {$\prescript{QH}{3.25}{\text{ID}}_{0.25}^{d, \text{13:00}}$ (6.93)} & \cellcolor[rgb]{1,0.959,0.5} {$\prescript{H}{2.25}{\text{ID}}_{0.25}^{d, \text{12:00}}$ (6.17)} & \cellcolor[rgb]{1,0.911,0.5} {$\prescript{QH}{0.75}{\text{ID}}_{0.25}^{d, \text{10:30}}$ (5.23)} & \cellcolor[rgb]{1,0.886,0.5} {$\prescript{QH}{2.25}{\text{ID}}_{0.25}^{d, \text{12:00}}$ (4.72)} \\ 
			13:15 & \cellcolor[rgb]{0.866,1,0.5} {$\prescript{H}{3}{\text{ID}}_{0.25}^{d, \text{13:00}}$ (10.34)} & \cellcolor[rgb]{1,0.992,0.5} {$\prescript{H}{2}{\text{ID}}_{0.25}^{d, \text{12:00}}$ (6.85)} & \cellcolor[rgb]{1,0.9,0.5} {$\prescript{QH}{1.5}{\text{ID}}_{0.25}^{d, \text{11:30}}$ (5.01)} & \cellcolor[rgb]{1,0.874,0.5} {$\prescript{H}{3.25}{\text{ID}}_{0.25}^{d, \text{13:00}}$ (4.48)} \\ 
			13:30 & \cellcolor[rgb]{0.972,1,0.5} {$\prescript{H}{2.75}{\text{ID}}_{0.25}^{d, \text{13:00}}$ (7.7)} & \cellcolor[rgb]{0.994,1,0.5} {$\prescript{QH}{2.5}{\text{ID}}_{0.25}^{d, \text{12:45}}$ (7.14)} & \cellcolor[rgb]{1,0.947,0.5} {$\prescript{H}{3.75}{\text{ID}}_{0.25}^{d, \text{14:00}}$ (5.94)} & \cellcolor[rgb]{1,0.917,0.5} {$\prescript{H}{3}{\text{ID}}_{0.25}^{d, \text{13:00}}$ (5.34)} \\ 
			13:45 & \cellcolor[rgb]{0.852,1,0.5} {$\prescript{QH}{3.25}{\text{ID}}_{0.25}^{d, \text{13:45}}$ (10.69)} & \cellcolor[rgb]{0.927,1,0.5} {$\prescript{QH}{2.25}{\text{ID}}_{0.25}^{d, \text{12:45}}$ (8.82)} & \cellcolor[rgb]{1,0.928,0.5} {$\text{IA}^{d,\text{13:45}}$ (5.56)} & \cellcolor[rgb]{1,0.878,0.5} {$\prescript{H}{3.5}{\text{ID}}_{0.25}^{d, \text{14:00}}$ (4.57)} \\ 
			14:00 & \cellcolor[rgb]{0.926,1,0.5} {$\prescript{QH}{1}{\text{ID}}_{0.25}^{d, \text{11:45}}$ (8.86)} & \cellcolor[rgb]{0.983,1,0.5} {$\prescript{QH}{3.25}{\text{ID}}_{0.25}^{d, \text{14:00}}$ (7.43)} & \cellcolor[rgb]{1,0.937,0.5} {$\prescript{QH}{0.75}{\text{ID}}_{0.25}^{d, \text{11:30}}$ (5.75)} & \cellcolor[rgb]{1,0.931,0.5} {$\prescript{H}{1.25}{\text{ID}}_{0.25}^{d, \text{12:00}}$ (5.62)} \\ 
			14:15 & \cellcolor[rgb]{0.954,1,0.5} {$\prescript{QH}{1.75}{\text{ID}}_{0.25}^{d, \text{12:45}}$ (8.16)} & \cellcolor[rgb]{0.955,1,0.5} {$\prescript{H}{3}{\text{ID}}_{0.25}^{d, \text{14:00}}$ (8.12)} & \cellcolor[rgb]{1,0.953,0.5} {$\prescript{QH}{0.75}{\text{ID}}_{0.25}^{d, \text{11:45}}$ (6.06)} & \cellcolor[rgb]{1,0.899,0.5} {$\prescript{H}{2}{\text{ID}}_{0.25}^{d, \text{13:00}}$ (4.97)} \\ 
			14:30 & \cellcolor[rgb]{1,0.922,0.5} {$\prescript{QH}{2.5}{\text{ID}}_{0.25}^{d, \text{13:45}}$ (5.43)} & \cellcolor[rgb]{1,0.851,0.5} {$\prescript{QH}{1.5}{\text{ID}}_{0.25}^{d, \text{12:45}}$ (4.02)} & \cellcolor[rgb]{1,0.85,0.5} {$\prescript{H}{3.75}{\text{ID}}_{0.25}^{d, \text{15:00}}$ (4)} & \cellcolor[rgb]{1,0.8,0.5} {$\prescript{H}{2.75}{\text{ID}}_{0.25}^{d, \text{14:00}}$ (3.01)} \\ 
			14:45 & \cellcolor[rgb]{0.913,1,0.5} {$\prescript{QH}{3.25}{\text{ID}}_{0.25}^{d, \text{14:45}}$ (9.19)} & \cellcolor[rgb]{0.913,1,0.5} {$\text{IA}^{d,\text{14:45}}$ (9.18)} & \cellcolor[rgb]{1,0.97,0.5} {$\prescript{QH}{2.25}{\text{ID}}_{0.25}^{d, \text{13:45}}$ (6.4)} & \cellcolor[rgb]{1,0.852,0.5} {$\prescript{H}{4.5}{\text{ID}}_{0.25}^{d, \text{16:00}}$ (4.04)} \\ 
			15:00 & \cellcolor[rgb]{0.893,1,0.5} {$\text{IA}^{d,\text{15:00}}$ (9.68)} & \cellcolor[rgb]{0.965,1,0.5} {$\prescript{QH}{1}{\text{ID}}_{0.25}^{d, \text{12:45}}$ (7.86)} & \cellcolor[rgb]{0.975,1,0.5} {$\prescript{QH}{3.25}{\text{ID}}_{0.25}^{d, \text{15:00}}$ (7.63)} & \cellcolor[rgb]{0.993,1,0.5} {$\prescript{QH}{3.5}{\text{ID}}_{0.25}^{d, \text{15:00}}$ (7.16)} \\ 
			15:15 & \cellcolor[rgb]{0.935,1,0.5} {$\prescript{H}{3}{\text{ID}}_{0.25}^{d, \text{15:00}}$ (8.62)} & \cellcolor[rgb]{1,0.913,0.5} {$\prescript{H}{2}{\text{ID}}_{0.25}^{d, \text{14:00}}$ (5.26)} & \cellcolor[rgb]{1,0.849,0.5} {$\prescript{QH}{1.75}{\text{ID}}_{0.25}^{d, \text{13:45}}$ (3.98)} & \cellcolor[rgb]{1,0.834,0.5} {$\prescript{H}{3.25}{\text{ID}}_{0.25}^{d, \text{15:00}}$ (3.68)} \\ 
			15:30 & \cellcolor[rgb]{0.975,1,0.5} {$\prescript{H}{3.75}{\text{ID}}_{0.25}^{d, \text{16:00}}$ (7.62)} & \cellcolor[rgb]{1,0.875,0.5} {$\prescript{QH}{4.75}{\text{ID}}_{0.25}^{d, \text{17:00}}$ (4.5)} & \cellcolor[rgb]{1,0.856,0.5} {$\prescript{QH}{3.25}{\text{ID}}_{0.25}^{d, \text{15:30}}$ (4.12)} & \cellcolor[rgb]{1,0.835,0.5} {$\prescript{QH}{4.25}{\text{ID}}_{0.25}^{d, \text{16:30}}$ (3.69)} \\ 
			15:45 & \cellcolor[rgb]{0.835,1,0.5} {$\prescript{QH}{3.25}{\text{ID}}_{0.25}^{d, \text{15:45}}$ (11.13)} & \cellcolor[rgb]{0.965,1,0.5} {$\prescript{QH}{2.25}{\text{ID}}_{0.25}^{d, \text{14:45}}$ (7.88)} & \cellcolor[rgb]{1,0.97,0.5} {$\text{IA}^{d,\text{15:45}}$ (6.41)} & \cellcolor[rgb]{1,0.831,0.5} {$\prescript{QH}{23.25}{\text{ID}}_{0.25}^{d, \text{15:45}}$ (3.62)} \\ 
			\hline
		\end{tabular}
		\endgroup
		\caption{Most relevant coefficients in the model \textbf{FI.elnet.penal.C} for each quarter-hourly product from 08:00 to 15:45} 
		\label{tab:beta_qh_2}
	\end{table}

	\begin{table}[!ht]
		\centering
		\begingroup\footnotesize
		\begin{tabular}{rllll}
			\hline
			& Importance: 1 & Importance: 2 & Importance: 3  & Importance: 4\\ 
			\hline
			16:00 & \cellcolor[rgb]{0.93,1,0.5} {$\prescript{QH}{3.25}{\text{ID}}_{0.25}^{d, \text{16:00}}$ (8.75)} & \cellcolor[rgb]{1,0.946,0.5} {$\prescript{QH}{2.25}{\text{ID}}_{0.25}^{d, \text{15:00}}$ (5.93)} & \cellcolor[rgb]{1,0.938,0.5} {$\text{IA}^{d,\text{16:00}}$ (5.76)} & \cellcolor[rgb]{1,0.919,0.5} {$\prescript{QH}{1}{\text{ID}}_{0.25}^{d, \text{13:45}}$ (5.37)} \\ 
			16:15 & \cellcolor[rgb]{0.855,1,0.5} {$\prescript{H}{3}{\text{ID}}_{0.25}^{d, \text{16:00}}$ (10.64)} & \cellcolor[rgb]{0.934,1,0.5} {$\prescript{H}{2}{\text{ID}}_{0.25}^{d, \text{15:00}}$ (8.65)} & \cellcolor[rgb]{1,0.895,0.5} {$\prescript{H}{3.25}{\text{ID}}_{0.25}^{d, \text{16:00}}$ (4.9)} & \cellcolor[rgb]{1,0.891,0.5} {$\prescript{QH}{0.75}{\text{ID}}_{0.25}^{d, \text{13:45}}$ (4.83)} \\ 
			16:30 & \cellcolor[rgb]{0.926,1,0.5} {$\prescript{QH}{2.5}{\text{ID}}_{0.25}^{d, \text{15:45}}$ (8.85)} & \cellcolor[rgb]{1,0.99,0.5} {$\prescript{QH}{1.5}{\text{ID}}_{0.25}^{d, \text{14:45}}$ (6.8)} & \cellcolor[rgb]{1,0.985,0.5} {$\prescript{QH}{3.25}{\text{ID}}_{0.25}^{d, \text{16:30}}$ (6.69)} & \cellcolor[rgb]{1,0.971,0.5} {$\prescript{H}{3.75}{\text{ID}}_{0.25}^{d, \text{17:00}}$ (6.43)} \\ 
			16:45 & \cellcolor[rgb]{0.75,0.983,0.5} {$\prescript{QH}{3.25}{\text{ID}}_{0.25}^{d, \text{16:45}}$ (14.18)} & \cellcolor[rgb]{1,0.963,0.5} {$\prescript{QH}{2.25}{\text{ID}}_{0.25}^{d, \text{15:45}}$ (6.27)} & \cellcolor[rgb]{1,0.961,0.5} {$\prescript{QH}{3.5}{\text{ID}}_{0.25}^{d, \text{16:45}}$ (6.22)} & \cellcolor[rgb]{1,0.937,0.5} {$\text{IA}^{d,\text{16:45}}$ (5.75)} \\ 
			17:00 & \cellcolor[rgb]{0.788,0.996,0.5} {$\prescript{QH}{3.25}{\text{ID}}_{0.25}^{d, \text{17:00}}$ (12.52)} & \cellcolor[rgb]{0.979,1,0.5} {$\text{IA}^{d,\text{17:00}}$ (7.51)} & \cellcolor[rgb]{1,0.928,0.5} {$\prescript{QH}{0.75}{\text{ID}}_{0.25}^{d, \text{14:30}}$ (5.56)} & \cellcolor[rgb]{1,0.903,0.5} {$\prescript{QH}{2.25}{\text{ID}}_{0.25}^{d, \text{16:00}}$ (5.07)} \\ 
			17:15 & \cellcolor[rgb]{0.902,1,0.5} {$\prescript{H}{3}{\text{ID}}_{0.25}^{d, \text{17:00}}$ (9.46)} & \cellcolor[rgb]{1,0.993,0.5} {$\prescript{H}{2}{\text{ID}}_{0.25}^{d, \text{16:00}}$ (6.86)} & \cellcolor[rgb]{1,0.948,0.5} {$\prescript{QH}{3.25}{\text{ID}}_{0.25}^{d, \text{17:15}}$ (5.95)} & \cellcolor[rgb]{1,0.945,0.5} {$\prescript{H}{3.25}{\text{ID}}_{0.25}^{d, \text{17:00}}$ (5.89)} \\ 
			17:30 & \cellcolor[rgb]{0.824,1,0.5} {$\prescript{QH}{3.25}{\text{ID}}_{0.25}^{d, \text{17:30}}$ (11.39)} & \cellcolor[rgb]{0.935,1,0.5} {$\prescript{H}{2.75}{\text{ID}}_{0.25}^{d, \text{17:00}}$ (8.62)} & \cellcolor[rgb]{0.956,1,0.5} {$\prescript{H}{3.75}{\text{ID}}_{0.25}^{d, \text{18:00}}$ (8.1)} & \cellcolor[rgb]{0.971,1,0.5} {$\prescript{QH}{1.5}{\text{ID}}_{0.25}^{d, \text{15:45}}$ (7.71)} \\ 
			17:45 & \cellcolor[rgb]{0.686,0.962,0.5} {$\prescript{QH}{3.25}{\text{ID}}_{0.25}^{d, \text{17:45}}$ (16.95)} & \cellcolor[rgb]{0.974,1,0.5} {$\prescript{H}{3.5}{\text{ID}}_{0.25}^{d, \text{18:00}}$ (7.65)} & \cellcolor[rgb]{0.983,1,0.5} {$\prescript{QH}{3.5}{\text{ID}}_{0.25}^{d, \text{17:45}}$ (7.41)} & \cellcolor[rgb]{1,0.999,0.5} {$\prescript{QH}{4}{\text{ID}}_{0.25}^{d, \text{17:45}}$ (6.98)} \\ 
			18:00 & \cellcolor[rgb]{0.803,1,0.5} {$\prescript{QH}{3.25}{\text{ID}}_{0.25}^{d, \text{18:00}}$ (11.91)} & \cellcolor[rgb]{0.857,1,0.5} {$\text{IA}^{d,\text{18:00}}$ (10.56)} & \cellcolor[rgb]{0.932,1,0.5} {$\prescript{QH}{3.5}{\text{ID}}_{0.25}^{d, \text{18:00}}$ (8.7)} & \cellcolor[rgb]{1,0.92,0.5} {$\prescript{QH}{3.75}{\text{ID}}_{0.25}^{d, \text{18:00}}$ (5.39)} \\ 
			18:15 & \cellcolor[rgb]{0.788,0.996,0.5} {$\prescript{H}{3}{\text{ID}}_{0.25}^{d, \text{18:00}}$ (12.5)} & \cellcolor[rgb]{0.868,1,0.5} {$\prescript{H}{3.25}{\text{ID}}_{0.25}^{d, \text{18:00}}$ (10.3)} & \cellcolor[rgb]{1,0.95,0.5} {$\prescript{H}{3.5}{\text{ID}}_{0.25}^{d, \text{18:00}}$ (6)} & \cellcolor[rgb]{1,0.89,0.5} {$\prescript{H}{3.75}{\text{ID}}_{0.25}^{d, \text{18:00}}$ (4.8)} \\ 
			18:30 & \cellcolor[rgb]{0.864,1,0.5} {$\prescript{H}{3.75}{\text{ID}}_{0.25}^{d, \text{19:00}}$ (10.41)} & \cellcolor[rgb]{0.872,1,0.5} {$\prescript{H}{2.75}{\text{ID}}_{0.25}^{d, \text{18:00}}$ (10.2)} & \cellcolor[rgb]{0.978,1,0.5} {$\prescript{QH}{3.5}{\text{ID}}_{0.25}^{d, \text{18:45}}$ (7.56)} & \cellcolor[rgb]{1,0.967,0.5} {$\prescript{H}{4}{\text{ID}}_{0.25}^{d, \text{19:00}}$ (6.34)} \\ 
			18:45 & \cellcolor[rgb]{0.682,0.961,0.5} {$\prescript{QH}{3.25}{\text{ID}}_{0.25}^{d, \text{18:45}}$ (17.11)} & \cellcolor[rgb]{0.8,1,0.5} {$\prescript{QH}{3.5}{\text{ID}}_{0.25}^{d, \text{18:45}}$ (12)} & \cellcolor[rgb]{0.92,1,0.5} {$\prescript{H}{3.5}{\text{ID}}_{0.25}^{d, \text{19:00}}$ (8.99)} & \cellcolor[rgb]{1,0.994,0.5} {$\text{IA}^{d,\text{18:45}}$ (6.89)} \\ 
			19:00 & \cellcolor[rgb]{0.78,0.993,0.5} {$\text{IA}^{d,\text{19:00}}$ (12.86)} & \cellcolor[rgb]{0.862,1,0.5} {$\prescript{H}{2.25}{\text{ID}}_{0.25}^{d, \text{18:00}}$ (10.45)} & \cellcolor[rgb]{0.967,1,0.5} {$\prescript{QH}{3.5}{\text{ID}}_{0.25}^{d, \text{19:00}}$ (7.83)} & \cellcolor[rgb]{0.999,1,0.5} {$\prescript{H}{3.25}{\text{ID}}_{0.25}^{d, \text{19:00}}$ (7.02)} \\ 
			19:15 & \cellcolor[rgb]{0.688,0.963,0.5} {$\prescript{H}{3}{\text{ID}}_{0.25}^{d, \text{19:00}}$ (16.86)} & \cellcolor[rgb]{0.896,1,0.5} {$\prescript{H}{3.25}{\text{ID}}_{0.25}^{d, \text{19:00}}$ (9.59)} & \cellcolor[rgb]{0.959,1,0.5} {$\prescript{H}{3.5}{\text{ID}}_{0.25}^{d, \text{19:00}}$ (8.02)} & \cellcolor[rgb]{0.966,1,0.5} {$\prescript{H}{3.75}{\text{ID}}_{0.25}^{d, \text{19:00}}$ (7.86)} \\ 
			19:30 & \cellcolor[rgb]{0.842,1,0.5} {$\prescript{H}{2.75}{\text{ID}}_{0.25}^{d, \text{19:00}}$ (10.95)} & \cellcolor[rgb]{0.933,1,0.5} {$\prescript{H}{3.75}{\text{ID}}_{0.25}^{d, \text{20:00}}$ (8.68)} & \cellcolor[rgb]{0.96,1,0.5} {$\prescript{QH}{2.5}{\text{ID}}_{0.25}^{d, \text{18:45}}$ (8)} & \cellcolor[rgb]{1,0.929,0.5} {$\prescript{H}{4}{\text{ID}}_{0.25}^{d, \text{20:00}}$ (5.59)} \\ 
			19:45 & \cellcolor[rgb]{0.696,0.965,0.5} {$\prescript{QH}{3.25}{\text{ID}}_{0.25}^{d, \text{19:45}}$ (16.51)} & \cellcolor[rgb]{0.943,1,0.5} {$\text{IA}^{d,\text{19:45}}$ (8.42)} & \cellcolor[rgb]{0.977,1,0.5} {$\prescript{H}{3.5}{\text{ID}}_{0.25}^{d, \text{20:00}}$ (7.58)} & \cellcolor[rgb]{0.997,1,0.5} {$\prescript{QH}{3.5}{\text{ID}}_{0.25}^{d, \text{19:45}}$ (7.08)} \\ 
			20:00 & \cellcolor[rgb]{0.767,0.989,0.5} {$\prescript{QH}{3.25}{\text{ID}}_{0.25}^{d, \text{20:00}}$ (13.41)} & \cellcolor[rgb]{0.866,1,0.5} {$\prescript{H}{2.25}{\text{ID}}_{0.25}^{d, \text{19:00}}$ (10.34)} & \cellcolor[rgb]{0.95,1,0.5} {$\text{IA}^{d,\text{20:00}}$ (8.26)} & \cellcolor[rgb]{1,0.986,0.5} {$\prescript{H}{3.25}{\text{ID}}_{0.25}^{d, \text{20:00}}$ (6.71)} \\ 
			20:15 & \cellcolor[rgb]{0.75,0.983,0.5} {$\prescript{H}{3}{\text{ID}}_{0.25}^{d, \text{20:00}}$ (14.15)} & \cellcolor[rgb]{0.847,1,0.5} {$\prescript{H}{2}{\text{ID}}_{0.25}^{d, \text{19:00}}$ (10.82)} & \cellcolor[rgb]{0.941,1,0.5} {$\prescript{H}{3.25}{\text{ID}}_{0.25}^{d, \text{20:00}}$ (8.46)} & \cellcolor[rgb]{1,0.946,0.5} {$\prescript{H}{3.5}{\text{ID}}_{0.25}^{d, \text{20:00}}$ (5.92)} \\ 
			20:30 & \cellcolor[rgb]{0.787,0.996,0.5} {$\prescript{QH}{2.5}{\text{ID}}_{0.25}^{d, \text{19:45}}$ (12.58)} & \cellcolor[rgb]{0.857,1,0.5} {$\prescript{H}{3.75}{\text{ID}}_{0.25}^{d, \text{21:00}}$ (10.58)} & \cellcolor[rgb]{0.883,1,0.5} {$\prescript{H}{2.75}{\text{ID}}_{0.25}^{d, \text{20:00}}$ (9.93)} & \cellcolor[rgb]{0.989,1,0.5} {$\prescript{QH}{2.75}{\text{ID}}_{0.25}^{d, \text{19:45}}$ (7.27)} \\ 
			20:45 & \cellcolor[rgb]{0.734,0.978,0.5} {$\text{IA}^{d,\text{20:45}}$ (14.84)} & \cellcolor[rgb]{0.834,1,0.5} {$\prescript{QH}{2.25}{\text{ID}}_{0.25}^{d, \text{19:45}}$ (11.15)} & \cellcolor[rgb]{0.933,1,0.5} {$\prescript{H}{3.5}{\text{ID}}_{0.25}^{d, \text{21:00}}$ (8.67)} & \cellcolor[rgb]{0.988,1,0.5} {$\prescript{QH}{3.25}{\text{ID}}_{0.25}^{d, \text{20:45}}$ (7.31)} \\ 
			21:00 & \cellcolor[rgb]{0.773,0.991,0.5} {$\prescript{H}{3.25}{\text{ID}}_{0.25}^{d, \text{21:00}}$ (13.19)} & \cellcolor[rgb]{1,0.991,0.5} {$\prescript{H}{2.25}{\text{ID}}_{0.25}^{d, \text{20:00}}$ (6.82)} & \cellcolor[rgb]{1,0.966,0.5} {$\prescript{H}{3.5}{\text{ID}}_{0.25}^{d, \text{21:00}}$ (6.32)} & \cellcolor[rgb]{1,0.897,0.5} {$\prescript{QH}{0.75}{\text{ID}}_{0.25}^{d, \text{18:30}}$ (4.94)} \\ 
			21:15 & \cellcolor[rgb]{0.704,0.968,0.5} {$\prescript{H}{3}{\text{ID}}_{0.25}^{d, \text{21:00}}$ (16.15)} & \cellcolor[rgb]{0.978,1,0.5} {$\prescript{H}{3.25}{\text{ID}}_{0.25}^{d, \text{21:00}}$ (7.56)} & \cellcolor[rgb]{1,0.969,0.5} {$\prescript{QH}{3.25}{\text{ID}}_{0.25}^{d, \text{21:15}}$ (6.37)} & \cellcolor[rgb]{1,0.866,0.5} {$\prescript{H}{3.5}{\text{ID}}_{0.25}^{d, \text{21:00}}$ (4.32)} \\ 
			21:30 & \cellcolor[rgb]{0.87,1,0.5} {$\prescript{H}{2.75}{\text{ID}}_{0.25}^{d, \text{21:00}}$ (10.25)} & \cellcolor[rgb]{0.899,1,0.5} {$\prescript{H}{3}{\text{ID}}_{0.25}^{d, \text{21:00}}$ (9.52)} & \cellcolor[rgb]{0.927,1,0.5} {$\prescript{H}{3.75}{\text{ID}}_{0.25}^{d, \text{22:00}}$ (8.83)} & \cellcolor[rgb]{1,0.955,0.5} {$\prescript{H}{4}{\text{ID}}_{0.25}^{d, \text{22:00}}$ (6.1)} \\ 
			21:45 & \cellcolor[rgb]{0.928,1,0.5} {$\prescript{QH}{3.25}{\text{ID}}_{0.25}^{d, \text{21:45}}$ (8.81)} & \cellcolor[rgb]{0.941,1,0.5} {$\prescript{QH}{2.25}{\text{ID}}_{0.25}^{d, \text{20:45}}$ (8.47)} & \cellcolor[rgb]{0.942,1,0.5} {$\prescript{QH}{3.5}{\text{ID}}_{0.25}^{d, \text{21:45}}$ (8.45)} & \cellcolor[rgb]{1,0.886,0.5} {$\prescript{QH}{1.25}{\text{ID}}_{0.25}^{d, \text{19:45}}$ (4.72)} \\ 
			22:00 & \cellcolor[rgb]{0.806,1,0.5} {$\text{IA}^{d,\text{22:00}}$ (11.85)} & \cellcolor[rgb]{0.88,1,0.5} {$\prescript{QH}{3.25}{\text{ID}}_{0.25}^{d, \text{22:00}}$ (10.01)} & \cellcolor[rgb]{0.907,1,0.5} {$\prescript{H}{3.25}{\text{ID}}_{0.25}^{d, \text{22:00}}$ (9.32)} & \cellcolor[rgb]{0.945,1,0.5} {$\prescript{H}{3.5}{\text{ID}}_{0.25}^{d, \text{22:00}}$ (8.37)} \\ 
			22:15 & \cellcolor[rgb]{0.704,0.968,0.5} {$\prescript{H}{3}{\text{ID}}_{0.25}^{d, \text{22:00}}$ (16.17)} & \cellcolor[rgb]{1,0.993,0.5} {$\prescript{H}{3.25}{\text{ID}}_{0.25}^{d, \text{22:00}}$ (6.85)} & \cellcolor[rgb]{1,0.872,0.5} {$\prescript{QH}{3.25}{\text{ID}}_{0.25}^{d, \text{22:15}}$ (4.45)} & \cellcolor[rgb]{1,0.851,0.5} {$\prescript{H}{3.5}{\text{ID}}_{0.25}^{d, \text{22:00}}$ (4.03)} \\ 
			22:30 & \cellcolor[rgb]{0.77,0.99,0.5} {$\prescript{H}{2.75}{\text{ID}}_{0.25}^{d, \text{22:00}}$ (13.31)} & \cellcolor[rgb]{0.978,1,0.5} {$\prescript{H}{3}{\text{ID}}_{0.25}^{d, \text{22:00}}$ (7.54)} & \cellcolor[rgb]{1,0.964,0.5} {$\prescript{QH}{3.25}{\text{ID}}_{0.25}^{d, \text{22:30}}$ (6.29)} & \cellcolor[rgb]{1,0.896,0.5} {$\prescript{H}{4}{\text{ID}}_{0.25}^{d, \text{23:00}}$ (4.92)} \\ 
			22:45 & \cellcolor[rgb]{0.825,1,0.5} {$\prescript{QH}{3.25}{\text{ID}}_{0.25}^{d, \text{22:45}}$ (11.39)} & \cellcolor[rgb]{0.88,1,0.5} {$\prescript{QH}{3.5}{\text{ID}}_{0.25}^{d, \text{22:45}}$ (10.01)} & \cellcolor[rgb]{0.913,1,0.5} {$\prescript{H}{2.5}{\text{ID}}_{0.25}^{d, \text{22:00}}$ (9.18)} & \cellcolor[rgb]{0.936,1,0.5} {$\prescript{H}{3.5}{\text{ID}}_{0.25}^{d, \text{23:00}}$ (8.59)} \\ 
			23:00 & \cellcolor[rgb]{0.608,0.936,0.5} {$\prescript{QH}{3.25}{\text{ID}}_{0.25}^{d, \text{23:00}}$ (20.34)} & \cellcolor[rgb]{0.972,1,0.5} {$\text{IA}^{d,\text{23:00}}$ (7.69)} & \cellcolor[rgb]{1,0.98,0.5} {$\prescript{H}{2.25}{\text{ID}}_{0.25}^{d, \text{22:00}}$ (6.6)} & \cellcolor[rgb]{1,0.968,0.5} {$\prescript{QH}{3.5}{\text{ID}}_{0.25}^{d, \text{23:00}}$ (6.36)} \\ 
			23:15 & \cellcolor[rgb]{0.822,1,0.5} {$\prescript{H}{3}{\text{ID}}_{0.25}^{d, \text{23:00}}$ (11.44)} & \cellcolor[rgb]{0.883,1,0.5} {$\prescript{H}{2}{\text{ID}}_{0.25}^{d, \text{22:00}}$ (9.93)} & \cellcolor[rgb]{0.985,1,0.5} {$\prescript{QH}{3.25}{\text{ID}}_{0.25}^{d, \text{23:15}}$ (7.38)} & \cellcolor[rgb]{1,0.979,0.5} {$\prescript{H}{3.5}{\text{ID}}_{0.25}^{d, \text{23:00}}$ (6.58)} \\ 
			23:30 & \cellcolor[rgb]{0.742,0.981,0.5} {$\prescript{H}{3}{\text{ID}}_{0.25}^{d, \text{23:00}}$ (14.49)} & \cellcolor[rgb]{0.749,0.983,0.5} {$\prescript{H}{2.75}{\text{ID}}_{0.25}^{d, \text{23:00}}$ (14.23)} & \cellcolor[rgb]{0.833,1,0.5} {$\prescript{H}{3.25}{\text{ID}}_{0.25}^{d, \text{23:00}}$ (11.18)} & \cellcolor[rgb]{0.927,1,0.5} {$\prescript{H}{3.5}{\text{ID}}_{0.25}^{d, \text{23:00}}$ (8.83)} \\ 
			23:45 & \cellcolor[rgb]{0.726,0.975,0.5} {$\prescript{QH}{3.25}{\text{ID}}_{0.25}^{d, \text{23:45}}$ (15.19)} & \cellcolor[rgb]{0.842,1,0.5} {$\prescript{H}{2.5}{\text{ID}}_{0.25}^{d, \text{23:00}}$ (10.95)} & \cellcolor[rgb]{0.897,1,0.5} {$\prescript{H}{2.75}{\text{ID}}_{0.25}^{d, \text{23:00}}$ (9.58)} & \cellcolor[rgb]{1,0.935,0.5} {$\text{IA}^{d,\text{23:45}}$ (5.7)} \\ 
			\hline
		\end{tabular}
		\endgroup
		\caption{Most relevant coefficients in the model \textbf{FI.elnet.penal.C} for each quarter-hourly product from 16:00 to 23:45} 
		\label{tab:beta_qh_3}
	\end{table}

	\clearpage
	
	\bibliographystyle{chicago}
	
	\bibliography{bibliography}	
	
\end{document}